\documentclass[11pt]{article}
\usepackage{fullpage}
\usepackage{amsmath,amssymb,amsthm,bbm,mathtools}
\usepackage{algorithm,algpseudocode}
\usepackage{graphicx,subfigure}
\usepackage{threeparttable}
\usepackage{enumitem}
\usepackage{multirow}
\usepackage{url}
\usepackage{bm}
\usepackage{color}
\usepackage{afterpage}

\graphicspath{{figs/}}

\newcommand{\Romannumber}[1]{\uppercase\expandafter{\romannumeral #1}}

\newcommand{\real}{\mathbb{R}}

\newcommand{\ud}[1]{\underline{#1}}

\newcommand{\bessel}{\mathsf{K}}

\newcommand{\symln}[1]{{\color{blue} #1}}

\DeclareMathOperator{\diag}{diag}

\DeclareMathOperator{\sgn}{sgn}

\theoremstyle{definition} \newtheorem{definition}{Definition}
\theoremstyle{remark}     
\theoremstyle{remark}     
\theoremstyle{plain}      \newtheorem{theorem}{Theorem}
\theoremstyle{plain}      
\theoremstyle{plain}      
\theoremstyle{plain}      
\theoremstyle{plain}      \newtheorem{lemma}[theorem]{Lemma}

        {\vskip 10pt \begin{algorithmic}[1]}%
        {\end{algorithmic} \vskip 10pt}


\newcounter{parentnumber}

\usepackage{natbib}

\begin{document}

\title{\textsf{Linear-Cost Covariance Functions for Gaussian Random Fields}}
\author{Jie Chen\thanks{MIT-IBM Watson AI Lab, IBM Research. Email: \texttt{chenjie@us.ibm.com}}
  \and Michael L. Stein\thanks{University of Chicago. Emails: \texttt{stein@galton.uchicago.edu}}}
\maketitle

\begin{abstract}
Gaussian random fields (GRF) are a fundamental stochastic model for spatiotemporal data analysis. An essential ingredient of GRF is the covariance function that characterizes the joint Gaussian distribution of the field. Commonly used covariance functions give rise to fully dense and unstructured covariance matrices, for which required calculations are notoriously expensive to carry out for large data. In this work, we propose a construction of covariance functions that result in matrices with a hierarchical structure. Empowered by matrix algorithms that scale linearly with the matrix dimension, the hierarchical structure is proved to be efficient for a variety of random field computations, including sampling, kriging, and likelihood evaluation. Specifically, with $n$ scattered sites, sampling and likelihood evaluation has an $O(n)$ cost and kriging has an $O(\log n)$ cost after preprocessing, particularly favorable for the kriging of an extremely large number of sites (e.g., predicting on more sites than observed). We demonstrate comprehensive numerical experiments to show the use of the constructed covariance functions and their appealing computation time. Numerical examples on a laptop include simulated data of size up to one million, as well as a climate data product with over two million observations.
\end{abstract}

\section*{Keywords}
Gaussian sampling; Kriging; Maximum likelihood estimation; Hierarchical matrix; Climate data

\section{Introduction}
A Gaussian random field (GRF) $Z(\bm{x}):\real^d\to\real$ is a random field where all of its finite-dimensional distributions are Gaussian. Often termed as \emph{Gaussian processes}, GRFs are widely adopted as a practical model in areas ranging from spatial statistics~\citep{Stein1999}, geology~\citep{Chiles2012}, computer experiments~\citep{Koehler1996}, uncertainty quantification~\citep{Smith2013}, to machine learning~\citep{Rasmussen2006}. Among the many reasons for its popularity, a computational advantage is that the Gaussian assumption enables many computations to be done with basic numerical linear algebra.

Although numerical linear algebra~\citep{Golub1996} is a mature discipline and decades of research efforts result in highly efficient and reliable software libraries (e.g., BLAS~\citep{Goto2008} and LAPACK~\citep{Anderson1999})\footnote{These libraries are the elementary components of commonly used software such as R, Matlab, and python.}, the computation of GRF models cannot overcome a fundamental scalability barrier. For a collection of $n$ scattered sites $\bm{x}_1$, $\bm{x}_2$, \ldots, $\bm{x}_n$, the computation typically requires $O(n^2)$ storage and $O(n^2)$ to $O(n^3)$ arithmetic operations, which easily hit the capacity of modern computers when $n$ is large. In what follows, we review the basic notation and a few computational components that underlie this challenge.

Denote by $\mu(\bm{x}):\real^d\to\real$ the mean function and $k(\bm{x},\bm{x}'):\real^d\times\real^d\to\real$ the covariance function, which is (strictly) positive definite. Let $X=\{\bm{x}_i\}_{i=1}^n$ be a set of sampling sites and let $\bm{z}=[Z(\bm{x}_1),\ldots,Z(\bm{x}_n)]^T$ (column vector) be a realization of the random field at $X$. Additionally, denote by $\bm{\mu}$ the mean vector with elements $\mu_i=\mu(\bm{x}_i)$ and by $K$ the covariance matrix with elements $K_{ij}=k(\bm{x}_i,\bm{x}_j)$.

\begin{description}
\item[Sampling] Realizing a GRF amounts to sampling the multivariate normal distribution $\mathcal{N}(\bm{\mu},K)$. To this end, one performs a matrix factorization $K=GG^T$ (e.g., Cholesky), samples a vector $\bm{y}$ from the standard normal, and computes
\begin{equation}\label{eqn:sampling}
\bm{z}=\bm{\mu}+G\bm{y}.
\end{equation}

\item[Kriging] Kriging is the estimation of $Z(\bm{x}_0)$ at a new site $\bm{x}_0$. Other terminology includes \emph{interpolation}, \emph{regression}\footnote{Regression often assumes a noise term that we omit here for simplicity. An alternative way to view the noise term is that the covariance function has a nugget.}, and \emph{prediction}. The random variable $Z(\bm{x}_0)$ conditioned on the observation $\bm{z}$ admits a normal distribution $\mathcal{N}(\mu_0,\sigma_0^2)$ with
\begin{equation}\label{eqn:kriging}
\mu_0=\mu(\bm{x}_0)+\bm{k}_0^TK^{-1}(\bm{z}-\bm{\mu})
\quad\text{and}\quad
\sigma_0^2=k(\bm{x}_0,\bm{x}_0)-\bm{k}_0^TK^{-1}\bm{k}_0,
\end{equation}
where $\bm{k}_0$ is the column vector $[k(\bm{x}_1,\bm{x}_0),k(\bm{x}_2,\bm{x}_0),\ldots,k(\bm{x}_n,\bm{x}_0)]^T$.

\item[Log-likelihood] The log-likelihood function of a Gaussian distribution $\mathcal{N}(\bm{\mu},K)$ is
\begin{equation}\label{eqn:loglik}
\mathcal{L}=-\frac{1}{2}(\bm{z}-\bm{\mu})^TK^{-1}(\bm{z}-\bm{\mu})-\frac{1}{2}\log\det K-\frac{n}{2}\log2\pi.
\end{equation}
The log-likelihood $\mathcal{L}$ is a function of $\bm{\theta}\in\real^p$ that parameterizes the mean function $\mu$ and the covariance function $k$. The evaluation of $\mathcal{L}$ is an essential ingredient in maximum likelihood estimation and Bayesian inference.
\end{description}

A common characteristic of these examples is the expensive numerical linear algebra computations: Cholesky-like factorization in~\eqref{eqn:sampling}, linear system solutions in~\eqref{eqn:kriging} and~\eqref{eqn:loglik}, and determinant computation in~\eqref{eqn:loglik}. In general, the covariance matrix $K$ is dense and thus these computations have $O(n^2)$ memory cost and $O(n^3)$ arithmetic cost. Moreover, a subtlety occurs in the kriging of more than a few sites. In dense linear algebra, a preferred approach for solving linear systems is not to form the matrix inverse explicitly; rather, one factorizes the matrix as a product of two triangular matrices with $O(n^3)$ cost, followed by triangular solves whose costs are only $O(n^2)$. Then, if one wants to krige $m=O(n)$ sites, the formulas in~\eqref{eqn:kriging}, particularly the variance calculation, have a total cost of $O(n^2m)=O(n^3)$. This cost indicates that speeding up matrix factorization alone is insufficient for kriging, because $m$ vectors $\bm{k}_0$ create another computational bottleneck.

\subsection{Existing Approaches}
Scaling up the computations for GRF models has been a topic of great interest in the statistics community for many years and has recently attracted the attention of the numerical linear algebra community. Whereas it is not the focus of this work to extensively survey the literature, we discuss a few representative approaches and their pros and cons.

A general idea for reducing the computations is to restrict oneself to covariance matrices $K$ that have an exploitable structure, e.g., sparse, low-rank, or block-diagonal. Covariance tapering~\citep{Furrer2006,Kaufman2008,Wang2011,Stein2013a} approximates a covariance function $k$ by multiplying it with another one $k_{\text{t}}$ that has a compact support. The resulting compactly supported function $kk_{\text{t}}$ potentially introduces sparsity to the matrix. However, often the appropriate support for statistical purposes is not narrow, which undermines the use of sparse linear algebra to speed up computation. Low-rank approximations~\citep{Cressie2008,Eidsvik2012} generally approximate $K$ by using a low-rank matrix plus a diagonal matrix. In many applications, such an approximation is quite limited, especially when the diagonal component of $K$ does not dominate the small-scale variation of the random field~\citep{Stein2008,Stein2014}. In machine learning under the context of kernel methods, a number of randomized low-rank approximation techniques were proposed (e.g., Nystr\"{o}m approximation~\citep{Drineas2005} and random Fourier features~\citep{Rahimi2007}). In these methods, often the rank may need to be fairly large relative to $n$ for a good approximation, particularly in high dimensions~\citep{Huang2014}. Moreover, not every low-rank approximation can krige $m=O(n)$ sites efficiently. The block-diagonal approximation casts an artificial independence assumption across blocks, which is unappealing, although this simple approach can outperform covariance tapering and low-rank methods in many circumstances \citep{Stein2008,Stein2014}.

Additionally, a number of methods have been proposed through exploiting other computationally friendly structures on the Gaussian process. Notable examples include LatticeKrig~\citep{Nychka2015}, predictive process~\citep{Finley2009}, nearest neighbor Gaussian process~\citep{Datta2016,Datta2016a}, stochastic PDE~\citep{Rue2009}, periodic embedding~\citep{Guinness2017,Guinness2019}, Metakriging~\citep{Minsker2015,Minsker2017}, Gapfill~\citep{Gerber2018}, and local approximate Gaussian process~\citep{Gramacy2015}. See the case study by~\citet{Heaton2019} and references therein for a more complete list of computational methods and empirical comparisons.

There also exists a rich literature focusing on only the parameter estimation of $\bm{\theta}$. Among them, spectral methods~\citep{Whittle1954,Guyon1982,Dahlhaus1987} deal with the data in the Fourier domain. These methods work less well for high dimensions~\citep{Stein1995} or when the data are ungridded~\citep{Fuentes2007}. Several methods focus on the approximation of the likelihood, wherein the log-determinant term~\eqref{eqn:loglik} may be approximated by using Taylor expansions~\citep{Zhang2006} or Hutchinson approximations~\citep{Aune2014,Han2017,Dong2017,Ubaru2017}. The composite-likelihood approach~\citep{Vecchia1988,Stein2004a,Caragea2007,Varin2011} partitions $X$ into subsets and expands the likelihood by using the law of successive conditioning. Then, the conditional likelihoods in the product chain are approximated by dropping the conditional dependence on faraway subsets. This approach is often competitive. Yet another approach is to solve unbiased estimating equations~\citep{Anitescu2012,Stein2013,Anitescu2017} instead of maximizing the log-likelihood $\mathcal{L}$. This approach rids the computation of the determinant term, but its effectiveness relies on fast matrix-vector multiplications~\citep{Chen2014b} and effective preconditioning of the covariance matrix~\citep{Stein2012a,Chen2013}.

Recently, a multi-resolution approach~\citep{Katzfuss2017} based on successive conditioning was proposed, wherein the covariance structure is approximated in a hierarchical manner. The remainder of the approximation at the coarse level is filled by the finer level. This approach shares quite a few characteristics with our approach, which falls under the umbrella of ``hierarchical matrices'' in numerical linear algebra. Whereas the structure of~\citet{Katzfuss2017} is obtained in a coarse-to-fine fashion, our approach derives the structure in a fine-to-coarse manner, allowing translations to a type of hierarchical matrices that admit $O(n)$ cost without $\log n$ factors. Comparison of kriging and likelihood performance can be found in Section~\ref{sec:exp.mra}.

\subsection{Proposed Approach}
In this work, we take a holistic view and propose an approach applicable to the various computational components of GRF. The idea is to construct covariance functions that render a linear storage and arithmetic cost for (at least) the computations occurring in~\eqref{eqn:sampling} to~\eqref{eqn:loglik}. Specifically, for any (strictly) positive definite function $k(\cdot,\cdot)$, which we call the ``base function,'' we propose a recipe to construct (strictly) positive definite functions $k_{\rm{h}}(\cdot,\cdot)$ as alternatives. The base function $k$ is not necessarily stationary. The subscript ``h'' standards for ``hierarchical,'' because the first step of the construction is a hierarchical partitioning of the computation domain. With the subscript ``h'', the storage of the corresponding covariance matrix $K_{\rm{h}}$, as well as the additional storage requirement incurred in matrix computations, is $O(n)$. Additionally,
\begin{enumerate}
\item the arithmetic costs of matrix construction $K_{\rm{h}}$, factorization $K_{\rm{h}}=G_{\rm{h}}G_{\rm{h}}^T$, explicit inversion $K_{\rm{h}}^{-1}$, and determinant calculation $\det(K_{\rm{h}})$ are $O(n)$;
\item for any dense vector $\bm{y}$ of matching dimension, the arithmetic costs of matrix-vector multiplications $G_{\rm{h}}\bm{y}$ and $K_{\rm{h}}^{-1}\bm{y}$ are $O(n)$; and
\item for any dense vector $\bm{w}$ of matching dimension, the arithmetic costs of the inner product $\bm{k}_{\rm{h},0}^T\bm{w}$ and the quadratic form $\bm{k}_{\rm{h},0}^TK_{\rm{h}}^{-1}\bm{k}_{\rm{h},0}$ are $O(\log n)$, provided that an $O(n)$ preprocessing is done independently of the new site $\bm{x}_0$.
\end{enumerate}
The last property indicates that the overall cost of kriging $m=O(n)$ sites and estimating the uncertainties is $O(n\log n)$, which dominates the preprocessing $O(n)$.

The essence of this computationally attractive approach is a special covariance structure that we coin ``recursively low-rank.'' Informally speaking, a matrix $A$ is recursively low-rank if it is a block-diagonal matrix plus a low-rank matrix, with such a structure re-occurring in each main diagonal block of the matrix. The ``recursive'' part mandates that the low-rank factors share the same subspace across levels. The matrix $K_{\rm{h}}$ resulting from the proposed covariance function $k_{\rm{h}}$ is a symmetric positive definite version of recursively low-rank matrices. Interesting properties of the recursively low-rank structure of $A$ include that $A^{-1}$ admits exactly the same structure, and that if $A$ is symmetric positive definite, it may be factorized as $GG^T$ where $G$ also admits the same structure, albeit not being symmetric. These are the essential properties that allow for the development of $O(n)$ algorithms throughout. Moreover, the recursively low-rank structure is carried out to the out-of-sample vector $\bm{k}_{\rm{h},0}$, which makes it possible to compute inner products $\bm{k}_{\rm{h},0}^T\bm{w}$ and quadratic forms $\bm{k}_{\rm{h},0}^TK_{\rm{h}}^{-1}\bm{k}_{\rm{h},0}$ in an $O(\log n)$ cost, asymptotically lower than $O(n)$.

This matrix structure is closely connected to the rich literature of fast kernel approximation methods in scientific computing, reflected through a similar hierarchical framework but fine distinctions in design choices. A holistic design that aims at fitting the many computational components of GRF simultaneously however narrows down the possible choices and rationalizes the one that we take. After the presentation of the technical details, we will discuss in depth the subtle distinctions with many related hierarchical matrix approaches in Section~\ref{sec:compare.hierarchical.matrix}.

Note that although the proposal is based on approximations, the constructed covariance function $k_{\rm{h}}$ is valid for any ``rank'' and the involved linear algebra algorithms compute exact quantities (under infinite precision). The properties of $k_{\rm{h}}$ can be far from those of $k$ owing to the hierarchical nature. In practice, one should fix the rank and let the data size grow, subject to computational budget. Treat $k_{\rm{h}}$ as a covariance model by itself and perform model selection, rather than increasing the rank to chase approximation quality.

\section{Recursively Low-Rank Covariance Function}
Let $k:S\times S\to\real$ be positive definite for some domain $S$; that is, for any set of points $\bm{x}_1,\ldots,\bm{x}_n\in S$ and any set of coefficients $\alpha_1,\ldots,\alpha_n\in\real$, the quadratic form $\sum_{ij}\alpha_i\alpha_jk(\bm{x}_i,\bm{x}_j)\ge0$. We say that $k$ is \emph{strictly} positive definite if the quadratic form is strictly greater than $0$ whenever the $\bm{x}$'s are distinct and not all of the $\alpha_i$'s are 0. Given any $k$ and $S$, in this section we propose a recipe for constructing functions $k_{\rm{h}}$ that are (strictly) positive definite if $k$ is so.
We note the often confusing terminology that a strictly positive definite function always yields a positive definite covariance matrix for $n$ distinct observations, whereas, for a positive definite function, this matrix is only required to be positive semi-definite.

Some notations are necessary. Let $X$ be an ordered list of points in $S$. We will use $k(X,X)$ to denote the matrix with elements $k(\bm{x},\bm{x}')$ for all pairs $\bm{x},\bm{x}'\in X$. Similarly, we use $k(X,\bm{x})$ and $k(\bm{x},X)$ to denote a column and a row vector, respectively, when one of the arguments passed to $k$ contains a singleton $\{\bm{x}\}$.
These notations apply to any function $k$ (including the constructed $k_{\rm{h}}$ and the $\psi^{(i)}$ defined later) and any domain $S$ (including subdomains of it).

The construction of $k_{\rm{h}}$ is based on a hierarchical partitioning of $S$. For simplicity, let us first consider a partitioning with only one level. Let $S$ be partitioned into disjoint subdomains $S_1,\ldots,S_t$ such that $S=S_1\cup\cdots\cup S_t$. Let $\ud{X}$ be a set of $r$ distinct points in $S$. If $k(\ud{X},\ud{X})$ is invertible, define
\begin{equation}\label{eqn:k.one.level}
k_{\rm{h}}(\bm{x},\bm{x}')=
\begin{cases}
k(\bm{x},\bm{x}'), & \text{if $\bm{x},\bm{x}'\in S_j$ for some $j$},\\
k(\bm{x},\ud{X})k(\ud{X},\ud{X})^{-1}k(\ud{X},\bm{x}'), & \text{otherwise}.
\end{cases}
\end{equation}
In words, \eqref{eqn:k.one.level} states that the covariance for a pair of sites $\bm{x},\bm{x}'$ is equal to $k(\bm{x},\bm{x}')$ if they are located in the same subdomain; otherwise, it is replaced by the Nystr\"{o}m approximation $k(\bm{x},\ud{X})k(\ud{X},\ud{X})^{-1}k(\ud{X},\bm{x}')$. The Nystr\"{o}m approximation is always no greater than $k(\bm{x},\bm{x}')$ and when $k$ is strictly positive definite, it attains $k(\bm{x},\bm{x}')$ only when either $\bm{x}$ or $\bm{x}'$ belongs to $\ud{X}$. Following convention, we call the $r$ points in $\ud{X}$ \emph{landmark points}. Throughout this work, we will reserve underscores to indicate a list of landmark points.  The term ``low-rank'' comes from the fact that a matrix generated from Nystr\"{o}m approximation generically has rank $r$ (when $n\ge r$), regardless of how large $n$ is.

The positive definiteness of $k_{\rm{h}}$ follows a simple Schur-complement split. Furthermore, we have a stronger result when $k$ is assumed to be strictly positive definite; in this case, $k_{\rm{h}}$ carries over the strictness. We summarize this property in the following theorem, whose proof is given in the appendix.

\begin{theorem}\label{thm:k.one.level}
The function $k_{\rm{h}}$ defined in~\eqref{eqn:k.one.level} is positive definite if $k$ is positive definite and $k(\ud{X},\ud{X})$ is invertible. Moreover, $k_{\rm{h}}$ is strictly positive definite if $k$ is so.
\end{theorem}

We now proceed to hierarchical partitioning. Such a partitioning of the domain $S$ may be represented by a partitioning tree $T$. We name the tree nodes by using lower case letters such as $j$ and let the subdomain it corresponds to be $S_j$. The root is always $j=1$ and hence $S\equiv S_1$. We write $\text{Ch}(j)$ to denote the set of all child nodes of $j$. Equivalently, this means that a (sub)domain $S_j$ is partitioned into disjoint subdomains $S_l$ for all $l\in\text{Ch}(j)$. An example is illustrated in Figure~\ref{fig:tree}, where $S_1=S_2\cup S_3\cup S_4$, $S_2=S_5\cup S_6\cup S_7$, and $S_4=S_8\cup S_9$.

\begin{figure}[ht]
\centering
\includegraphics[width=.48\linewidth]{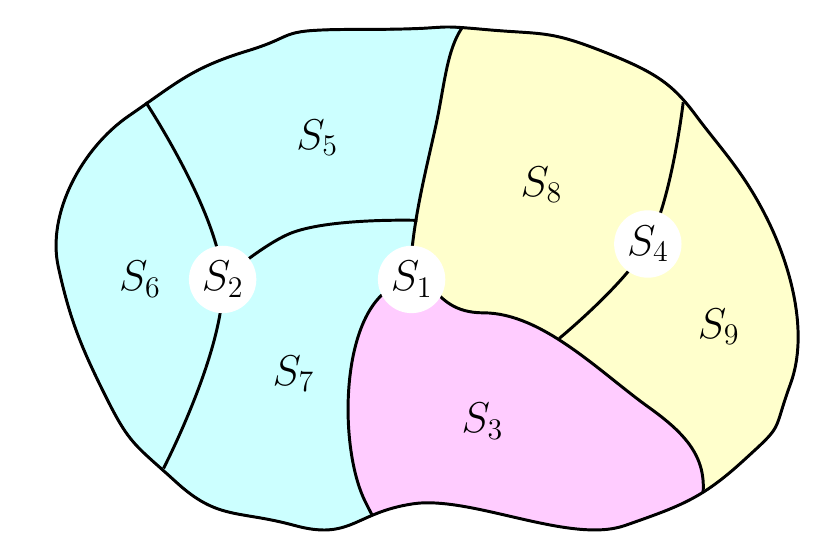}
\includegraphics[width=.48\linewidth]{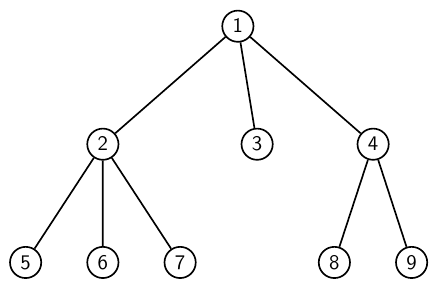}
\caption{Domain $S$ and partitioning tree $T$.}
\label{fig:tree}
\end{figure}

We now define a covariance function $k_{\rm{h}}$ based on hierarchical partitioning. For each nonleaf node $i$, let $\ud{X}_i$ be a set of $r$ landmark points in $S_i$ and assume that $k(\ud{X}_i,\ud{X}_i)$ is invertible. The main idea is to cascade the definition of covariance to those of the child subdomains. Thus, we recursively define a function $k_{\rm{h}}^{(i)}:S_i\times S_i\to\real$ such that if $\bm{x}$ and $\bm{x}'$ belong to the same child subdomain $S_j$ of $S_i$, then $k_{\rm{h}}^{(i)}(\bm{x},\bm{x}')=k_{\rm{h}}^{(j)}(\bm{x},\bm{x}')$; otherwise, $k_{\rm{h}}^{(i)}(\bm{x},\bm{x}')$ resembles a Nystr\"{o}m approximation. Formally, our covariance function
\begin{equation}\label{eqn:k.multilevel.part1}
k_{\rm{h}}\equiv k_{\rm{h}}^{(1)},
\end{equation}
where for any tree node $i$,
\begin{equation}\label{eqn:k.multilevel.part2}
k_{\rm{h}}^{(i)}(\bm{x},\bm{x}')=
\begin{cases}
k(\bm{x},\bm{x}'), & \text{if $i$ is leaf},\\
k_{\rm{h}}^{(j)}(\bm{x},\bm{x}'), & \text{if $\bm{x},\bm{x}'\in S_j$ for some $j\in\text{Ch}(i)$},\\
\psi^{(i)}(\bm{x},\ud{X}_i)k(\ud{X}_i,\ud{X}_i)^{-1}\psi^{(i)}(\ud{X}_i,\bm{x}'), & \text{otherwise}.
\end{cases}
\end{equation}
The auxiliary function $\psi^{(i)}(\bm{x},\ud{X}_i)$ cannot be the same as $k(\bm{x},\ud{X}_i)$, because positive definiteness will be lost. Instead, we make the following recursive definition when $\bm{x}\in S_i$:
\begin{equation}\label{eqn:k.multilevel.part3}
\psi^{(i)}(\bm{x},\ud{X}_i)=
\begin{cases}
k(\bm{x},\ud{X}_i), & \text{if $\bm{x}\in S_j$ for some $j\in\text{Ch}(i)$ and $j$ is leaf},\\
\psi^{(j)}(\bm{x},\ud{X}_j)k(\ud{X}_j,\ud{X}_j)^{-1}k(\ud{X}_j,\ud{X}_i), & \text{if $\bm{x}\in S_j$ for some $j\in\text{Ch}(i)$ but $j$ is not leaf}.
\end{cases}
\end{equation}

To understand the definition, we expand the recursive formulas~\eqref{eqn:k.multilevel.part1}--\eqref{eqn:k.multilevel.part3} for a pair of points $\bm{x}\in S_j$ and $\bm{x}'\in S_l$, where $j$ and $l$ are two leaf nodes. If $j=l$, it is trivial that $k_{\rm{h}}(\bm{x},\bm{x}')=k(\bm{x},\bm{x}')$. Otherwise, they have a unique least common ancestor $p$. Then,
\begin{multline}\label{eqn:expand}
k_{\rm{h}}(\bm{x},\bm{x}')=k_{\rm{h}}^{(p)}(\bm{x},\bm{x}')\\
=\underbrace{k(\bm{x},\ud{X}_{j_1})k(\ud{X}_{j_1},\ud{X}_{j_1})^{-1}k(\ud{X}_{j_1},\ud{X}_{j_2})\cdots k(\ud{X}_{j_s},\ud{X}_{j_s})^{-1}k(\ud{X}_{j_s},\ud{X}_{p})}_{\psi^{(p)}(\bm{x},\ud{X}_p)}k(\ud{X}_{p},\ud{X}_{p})^{-1}\\
\cdot\underbrace{k(\ud{X}_{p},\ud{X}_{l_t})k(\ud{X}_{l_t},\ud{X}_{l_t})^{-1}\cdots k(\ud{X}_{l_2},\ud{X}_{l_1})k(\ud{X}_{l_1},\ud{X}_{l_1})^{-1}k(\ud{X}_{l_1},\bm{x}')}_{\psi^{(p)}(\ud{X}_p,\bm{x}')},
\end{multline}
where $(j,j_1,j_2,\ldots,j_s,p)$ is the path in the tree connecting $j$ and $p$ and similarly $(l,l_1,l_2,\ldots,l_t,p)$ is the path connecting $l$ and $p$. The vectors $\psi^{(p)}(\bm{x},\ud{X}_p)$ and $\psi^{(p)}(\ud{X}_p,\bm{x}')$ on the two sides of $k(\ud{X}_{p},\ud{X}_{p})^{-1}$ come from recursively applying~\eqref{eqn:k.multilevel.part3}.

The definition~\eqref{eqn:k.multilevel.part1}--\eqref{eqn:k.multilevel.part3} admits a covariance decomposition that progressively includes cross-covariances for larger and larger subdomains up the tree. Let us define a function $\xi^{(i)}:S\times S\to\real$ for each node $i$, which has a support on only $S_i\times S_i$; that is, $\xi^{(i)}(\bm{x},\bm{x}') = 0$ if either $\bm{x}$ or $\bm{x}'\notin S_i$. When both $\bm{x}$ and $\bm{x}'$ belong to $S_i$,
\begin{equation}\label{eqn:xi}
\xi^{(i)}(\bm{x},\bm{x}') =
\begin{cases}
k(\bm{x},\bm{x}')-k(\bm{x},\ud{X}_{p})k(\ud{X}_{p},\ud{X}_{p})^{-1}k(\ud{X}_{p},\bm{x}'), & \text{if $i$ is leaf},\\
\psi^{(i)}(\bm{x},\ud{X}_i)k(\ud{X}_i,\ud{X}_i)^{-1}
\Delta
k(\ud{X}_i,\ud{X}_i)^{-1}\psi^{(i)}(\ud{X}_i,\bm{x}'), & \text{if $i$ is neither leaf nor root},\\
\psi^{(i)}(\bm{x},\ud{X}_i)k(\ud{X}_i,\ud{X}_i)^{-1}\psi^{(i)}(\ud{X}_i,\bm{x}'), & \text{if $i$ is root},
\end{cases}
\end{equation}
where $\Delta = k(\ud{X}_i,\ud{X}_i)-k(\ud{X}_i,\ud{X}_{p})k(\ud{X}_{p},\ud{X}_{p})^{-1}k(\ud{X}_{p},\ud{X}_i)$ and $p$ denotes the parent of $i$. Through telescoping, one sees that $k_{\rm{h}}$ is the sum of $\xi^{(i)}$ for all nodes $i$ in the tree: $k(\bm{x},\bm{x}')=\sum_{i \in T}\xi^{(i)}(\bm{x},\bm{x}')$. Intuitively, at a leaf node $i$, $\xi^{(i)}$ is the covariance of the posterior Gaussian conditioned on the landmark set $\ud{X}_p$. Moving up one level, $\xi^{(p)}$ defines not only the cross-covariance between subdomains of $S_p$, but also modifies the covariance inside each subdomain, say $i$, into $k(\bm{x},\bm{x}')-\psi^{(q)}(\bm{x},\ud{X}_{q})k(\ud{X}_{q},\ud{X}_{q})^{-1}\psi^{(q)}(\ud{X}_{q},\bm{x}')$, where $q$ is the parent of $p$, when $\xi^{(p)}$ is added to $\xi^{(i)}$. Iteratively adding the $\xi$'s from leaf to root, we have all the cross-covariances defined and subsequently modified, as well as the covariance inside each leaf node modified to $k(\bm{x},\bm{x}')$.

Similar to Theorem~\ref{thm:k.one.level}, the positive definiteness of $k$ follows from recursive Schur-complement splits across the hierarchy tree. Furthermore, we have that $k_{\rm{h}}$ is strictly positive definite if $k$ is so. We summarize the overall result in the following theorem, whose proof is given in the appendix.

\begin{theorem}\label{thm:k.multilevel}
The function $k_{\rm{h}}$ defined in~\eqref{eqn:k.multilevel.part1}--\eqref{eqn:k.multilevel.part3} is positive definite if $k$ is positive definite and $k(\ud{X}_i,\ud{X}_i)$ is invertible for all nonleaf nodes $i$. Moreover, $k_{\rm{h}}$ is strictly positive definite if $k$ is so.
\end{theorem}

\begin{figure}[ht]
\centering
\subfigure[$k_{\rm{h}}$, $r=8$]{
  \includegraphics[width=.23\linewidth]{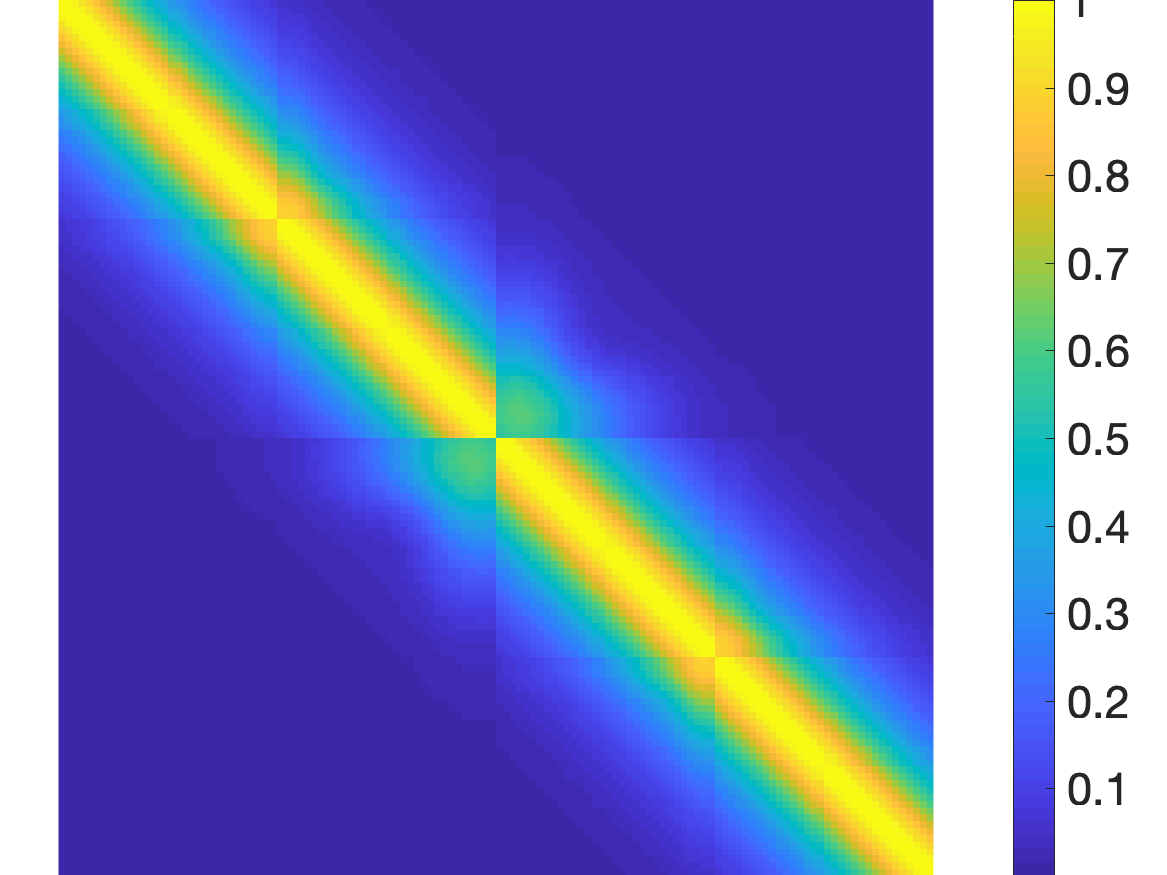}}
\subfigure[$k_{\rm{h}}$, $r=16$]{
  \includegraphics[width=.23\linewidth]{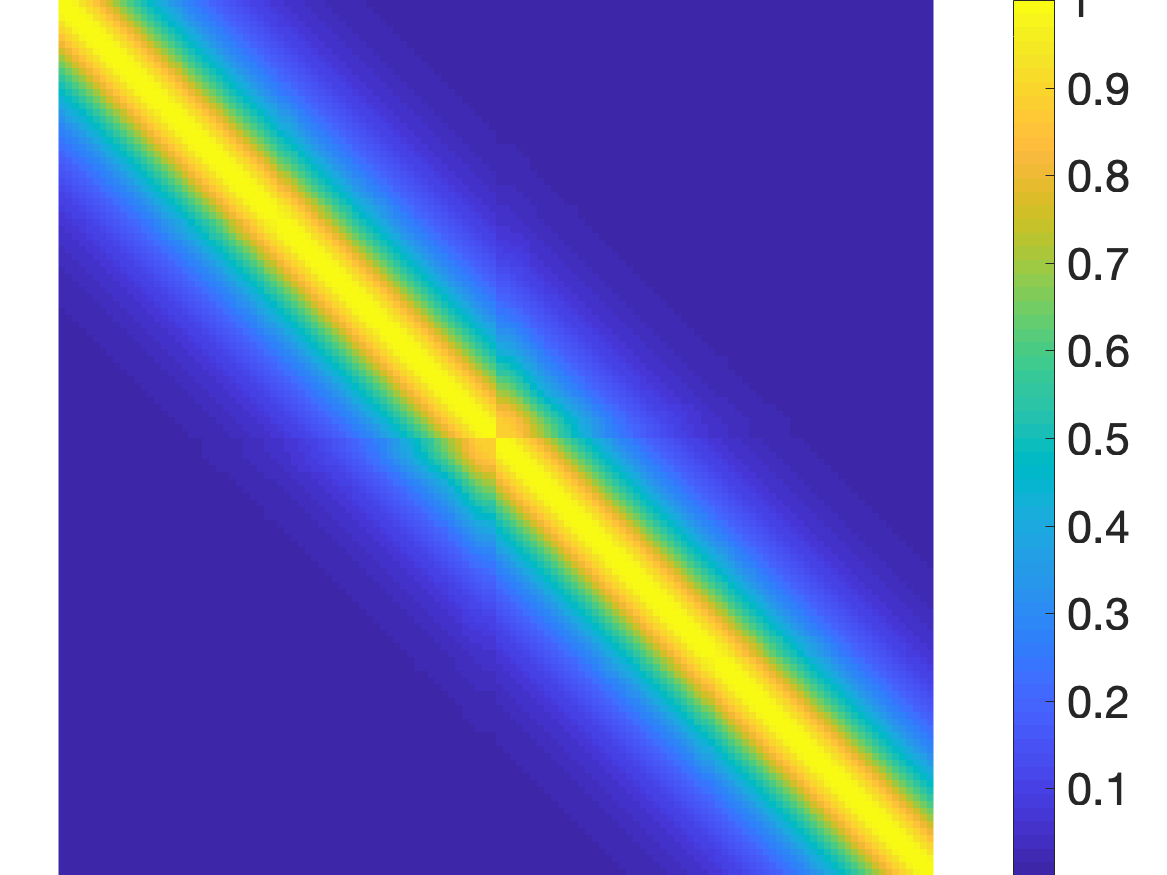}}
\subfigure[$k_{\rm{h}}$, $r=32$]{
  \includegraphics[width=.23\linewidth]{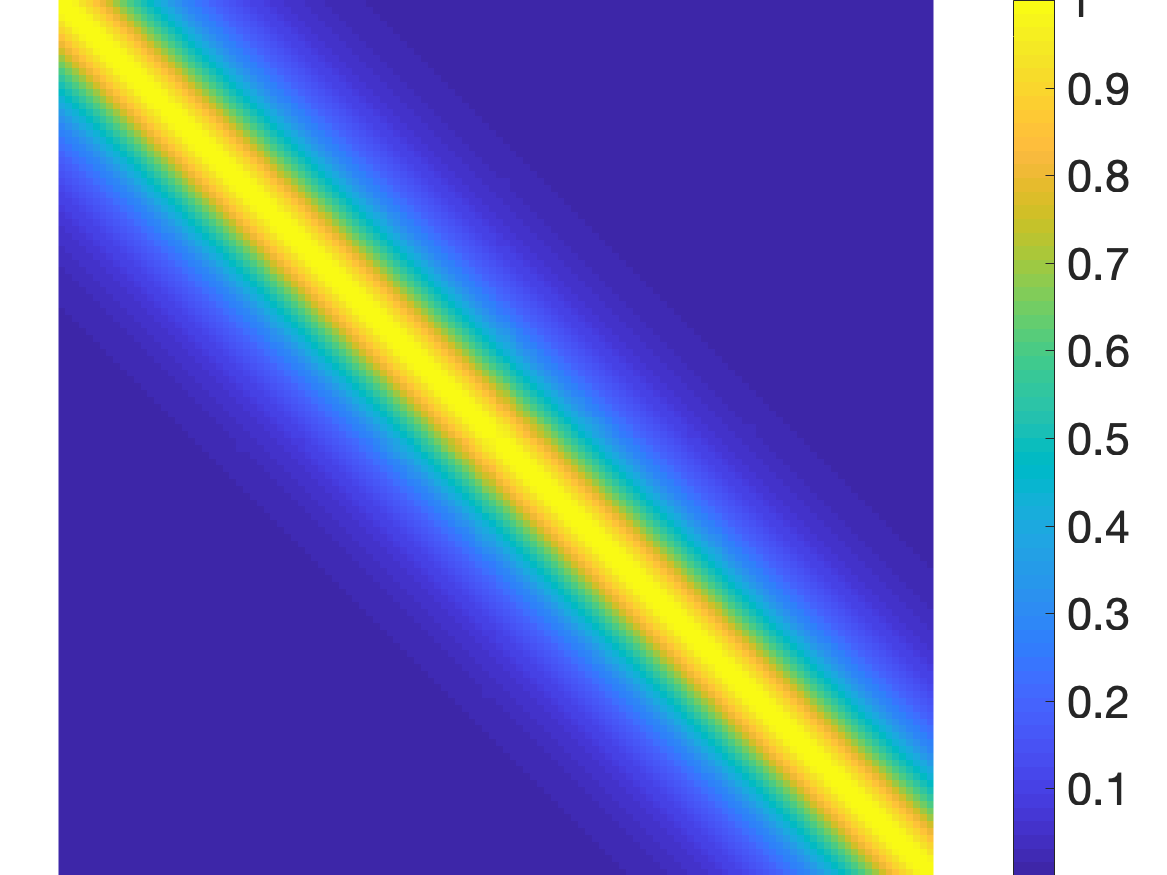}}
\subfigure[$k$]{
  \includegraphics[width=.23\linewidth]{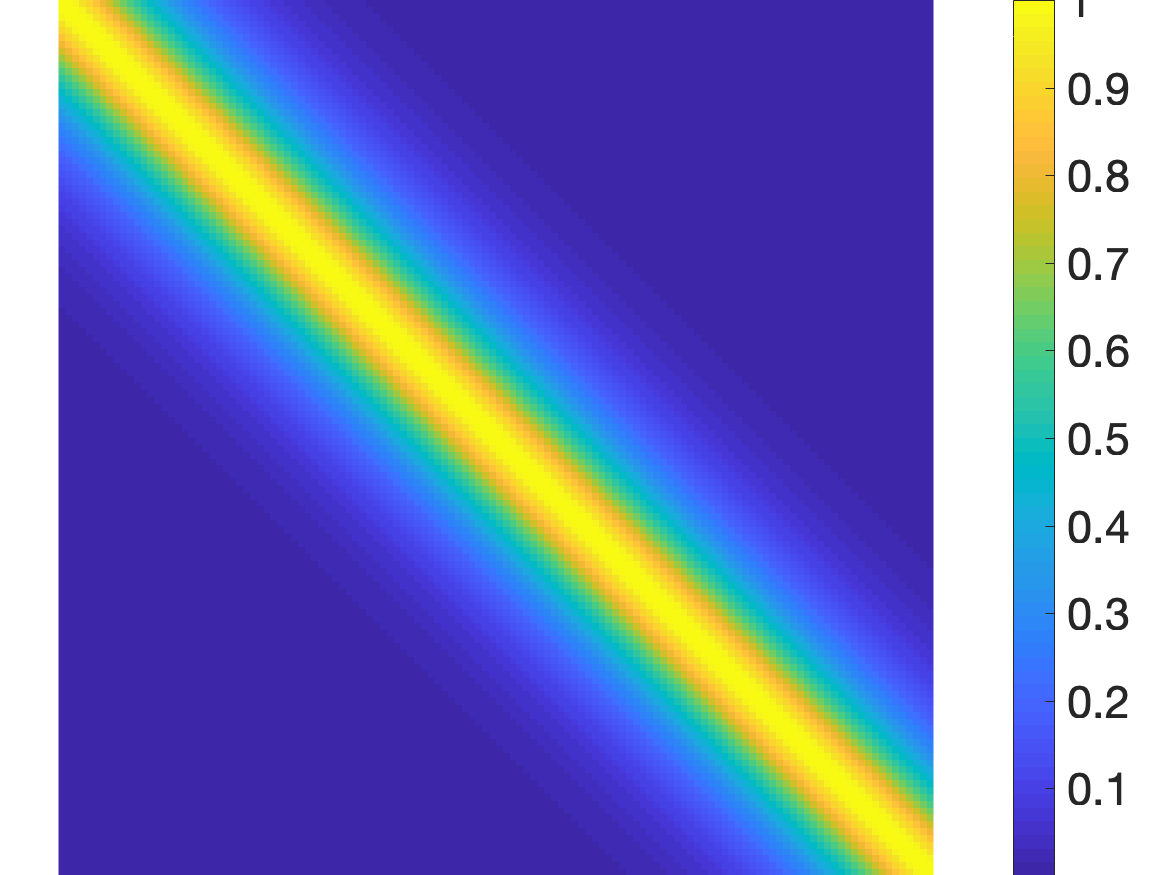}}
\caption{An example Mat\'{e}rn covariance function $k(\cdot,\cdot)$ and the constructed $k_{\rm{h}}(\cdot,\cdot)$'s in $[0,1]\times[0,1]$.}
\label{fig:kernel}
\end{figure}

In Figure~\ref{fig:kernel}, we show an example covariance function $k$ on $\real^1\times\real^1$ together with the constructed $k_{\rm{h}}$'s with different number of landmark points, $r$. The base $k$ is the Mat\'{e}rn covariance function (see~\eqref{eqn:matern} for definition) with sill $1.0$, range $0.2$, smoothness $1.5$, and nugget $0$. The considered domain $S=[0,1]$ is partitioned into equal halves recursively three times, resulting in eight leaf subdomains. Although $k$ is stationary, $k_{\rm{h}}$ is not and thus the visualization does not show a diagonally constant pattern.

The visual cues offered by plot (a) reveal a recursive blocking structure of $k_{\rm{h}}$, whereby the main diagonal blocks hold covariances inside the same subdomain and off-diagonal blocks across subdomains. For a pair of points in the same leaf subdomain, their covariance retains the value $k$. When they belong to different subdomains, low-rank approximation takes effects; the higher level in the hierarchy tree, the more aggressive is the approximation (see~\eqref{eqn:expand}). Naturally, when $r$ is small, the aggressive approximation renders a noticeable departure from the value $k$, as evident in the off-diagonal blocks toward the center of the plot. As $r$ increases, such a discrepancy is mitigated. When $r=32$, one sees barely any difference between plots (c) and (d).

It is important to note that the approximation does not depend on the number of sites, $n$. More importantly, $k_{\rm{h}}$ is a valid covariance function for any positive integer $r$. Rather than interpreting $k_{\rm{h}}$ as an approximation of $k$, one can treat $k_{\rm{h}}$ as a new covariance model and select models through comparing likelihoods. Given a fixed $r$, $k_{\rm{h}}$ can be applied to arbitrary data size $n$. The appealing $O(n)$ computational cost elucidated subsequently allows for efficient likelihood comparison.

\section{Recursively Low-Rank Matrix $A$}
An advantage of the proposed covariance function $k_{\rm{h}}$ is that when the number of landmark points in each subdomain is considered fixed, the covariance matrix $K_{\rm{h}}\equiv k_{\rm{h}}(X,X)$ for a set $X$ of $n$ points admits computational costs only linear in $n$. Such a desirable scaling comes from the fact that $K_{\rm{h}}$ is a special case of \emph{recursively low-rank matrices} whose computational costs are linear in the matrix dimension. In this section, we discuss these matrices and their operations (such as factorization and inversion). Then, in the section that follows, we will show the specialization of $K_{\rm{h}}$ and discuss additional vector operations tied to $k_{\rm{h}}$.

Let us first introduce some notation. Let $I=\{1,\ldots,n\}$. The index set $I$ may be recursively (permuted and) partitioned, resulting in a hierarchical formation that resembles the second panel of Figure~\ref{fig:tree}. Then, corresponding to a node $i$ is a subset $I_i\subset I$. Moreover, we have $I_i=\cup_{j\in\text{Ch}(i)} I_j$ where the $I_j$'s under union are disjoint. For an $n\times n$ real matrix $A$, we use $A(I_j,I_l)$ to denote a submatrix whose rows correspond to the index set $I_j$ and columns to $I_l$. We also follow the Matlab convention and use $:$ to mean all rows/columns when extracting submatrices. Further, we use $|I|$ to denote the cardinality of an index set $I$. We now define a recursively low-rank matrix.

\begin{definition}\label{def:matrix}
A matrix $A\in\real^{n\times n}$ is said to be \emph{recursively low-rank} with a partitioning tree $T$ and a positive integer $r$ if
\begin{enumerate}
\item for every pair of sibling nodes $i$ and $j$ with parent $p$, the block $A(I_i,I_j)$ admits a factorization
\[
A(I_i,I_j)=U_i\Sigma_pV_j^T
\]
for some $U_i\in\real^{|I_i|\times r}$, $\Sigma_p\in\real^{r\times r}$, and $V_j\in\real^{|I_j|\times r}$; and
\item for every pair of child node $i$ and parent node $p$ not being the root, the factors
\[
U_p(I_i,:)=U_iW_p \quad\text{and}\quad V_p(I_i,:)=V_iZ_p
\]
for some $W_p,Z_p\in\real^{r\times r}$.
\end{enumerate}
\end{definition}

In Definition~\ref{def:matrix}, the first item states that each off-diagonal block of $A$ is a rank-$r$ matrix. The middle factor $\Sigma_p$ is shared by all children of the same parent $p$, whereas the left factor $U_i$ and the right factor $V_j$ may be obtained through a change of basis from the corresponding factors in the child level, as detailed by the second item of the definition. As a consequence, if $\text{Ch}(i)=\{i_1,\ldots,i_s\}$ and $\text{Ch}(j)=\{j_1,\ldots,j_t\}$, then
\[
A(I_i,I_j)=
\underbrace{\begin{bmatrix} U_{i_1} \\ \vdots \\ U_{i_s} \end{bmatrix}W_i}_{U_i}
\Sigma_p
\underbrace{Z_j^T\begin{bmatrix} V_{j_1}^T & \cdots & V_{j_t}^T \end{bmatrix}}_{V_j^T}.
\]

From now on, we use the shorthand notation $A_{ii}$ to denote a diagonal block $A(I_i,I_i)$ and $A_{ij}$ to denote an off-diagonal block $A(I_i,I_j)$. A pictorial illustration of $A$, which corresponds to the tree in Figure~\ref{fig:tree}, is given in Figure~\ref{fig:matrix}. Then, $A$ is completely represented by the factors
\begin{equation}\label{eqn:A.factors}
\{A_{ii}, U_i, V_i, \Sigma_p, W_q, Z_q \mid i \text{ is leaf, } p \text{ is nonleaf, } q \text{ is neither leaf nor root}\}.
\end{equation}
In computer implementation, we store these factors in the corresponding nodes of the tree. See Figure~\ref{fig:tree2} for an extended example of Figure~\ref{fig:tree}. Clearly, $A$ is symmetric when $A_{ii}$ and $\Sigma_p$ are symmetric, $U_i=V_i$, and $W_q=Z_q$ for all appropriate nodes $i$, $p$, and $q$. In this case, the computer storage can be reduced by approximately a factor of $1/3$ through omitting the $V_i$'s and $Z_q$'s; meanwhile, matrix operations with $A$ often have a reduced cost, too.

\begin{figure}[ht]
\centering
\includegraphics[width=.45\linewidth]{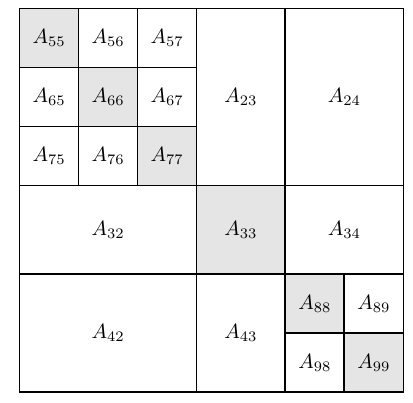}
\caption{The matrix $A$ corresponding to the partitioning tree in Figure~\ref{fig:tree}.}
\label{fig:matrix}
\end{figure}

\begin{figure}[ht]
\centering
\includegraphics{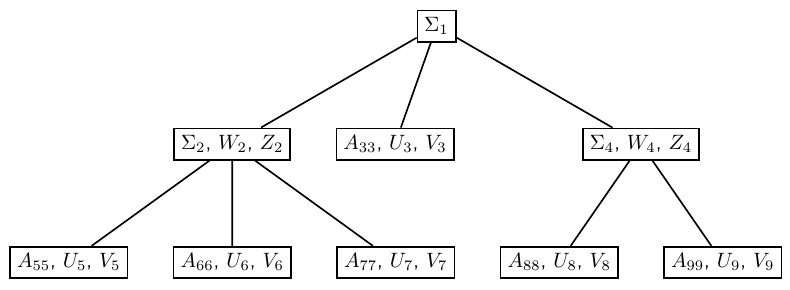}
\caption{Data structure for storing $A$. The partitioning tree is the same as that in Figure~\ref{fig:tree}.}
\label{fig:tree2}
\end{figure}

It is useful to note that not all matrix computations concerned in this paper are done with a symmetric matrix, although the covariance matrix is always so. One instance with unsymmetric matrices is sampling, where the matrix is a Cholesky-like factor of the covariance matrix. Hence, in this section, general algorithms are derived whenever $A$ may be unsymmetric, but we note the simplification for the symmetric case as appropriate.

The four matrix operations under consideration are:
\begin{enumerate}
\item matrix-vector multiplication $\bm{y}=A\bm{b}$;
\item matrix inversion $\widetilde{A}=A^{-1}$;
\item determinant $\det(A)$; and
\item Cholesky-like factorization $A=GG^T$ (when $A$ is symmetric positive definite).
\end{enumerate}
The detailed algorithms are presented in the appendix. Suffice it to mention here that interestingly, all algorithms are in the form of tree walks (e.g., preorder or postorder traversals) that heavily use the tree data structure illustrated in Figure~\ref{fig:tree2}. The inversion and Cholesky-like factorization rely on existence results summarized in the following. The proofs of these theorems are constructive, which simultaneously produce the algorithms. Hence, one may find the proofs inside the algorithms given in the appendix.

\begin{theorem}\label{thm:invA}
Let $A$ be recursively low-rank with a partitioning tree $T$ and a positive integer $r$. If $A$ is invertible and additionally, $A_{ii}-U_i\Sigma_pV_i^T$ is also invertible for all pairs of nonroot node $i$ and parent $p$, then there exists a recursively low-rank matrix $\widetilde{A}$ with the same partitioning tree $T$ and integer $r$, such that $\widetilde{A}=A^{-1}$. Following~\eqref{eqn:A.factors}, we denote the corresponding factors of $\widetilde{A}$ to be
\[
\{\widetilde{A}_{ii}, \widetilde{U}_i, \widetilde{V}_i, \widetilde{\Sigma}_p, \widetilde{W}_q, \widetilde{Z}_q \mid i \text{ is leaf, } p \text{ is nonleaf, } q \text{ is neither leaf nor root}\}.
\]
\end{theorem}

\begin{theorem}\label{thm:cholA}
Let $A$ be recursively low-rank with a partitioning tree $T$ and a positive integer $r$. If $A$ is symmetric, by convention let $A$ be represented by the factors
\[
\{A_{ii}, U_i, U_i, \Sigma_p, W_q, W_q \mid i \text{ is leaf, } p \text{ is nonleaf, } q \text{ is neither leaf nor root}\}.
\]
Furthermore, if $A$ is positive definite and additionally, $A_{ii}-U_i\Sigma_pU_i^T$ is also positive definite for all pairs of nonroot node $i$ and parent $p$, then there exists a recursively low-rank matrix $G$ with the same partitioning tree $T$ and integer $r$, and with factors
\[
\{G_{ii}, U_i, V_i, \Omega_p, W_q, Z_q \mid i \text{ is leaf, } p \text{ is nonleaf, } q \text{ is neither leaf nor root}\},
\]
such that $A=GG^T$.
\end{theorem}

\section{Covariance Matrix $K_{\rm{h}}$ as a Special Case of $A$ and Out-Of-Sample Extension}\label{sec:out.of.sample}
As noted at the beginning of the preceding section, the covariance matrix $K_{\rm{h}}=k_{\rm{h}}(X,X)$ is a special case of recursively low-rank matrices. This fact may be easily verified through populating the factors of $A$ defined in Definition~\ref{def:matrix}. Specifically, let $X$ be a set of $n$ distinct points in $S$ and let $X_j=X\cap S_j$ for all (sub)domains $S_j$. To avoid degeneracy assume $X_j\ne\emptyset$ for all $j$. Assign a recursively low-rank matrix $A$ in the following manner:
\begin{enumerate}
\item for every leaf node $i$, let $A_{ii}=k(X_i,X_i)$;
\item for every nonleaf node $p$, let $\Sigma_p=k(\ud{X}_p,\ud{X}_p)$;
\item for every leaf node $i$, let $U_i=V_i=k(X_i,\ud{X}_p)k(\ud{X}_p,\ud{X}_p)^{-1}$ where $p$ is the parent of $i$; and
\item for every nonleaf node $p$ not being the root, let $W_p=Z_p=k(\ud{X}_p,\ud{X}_q)k(\ud{X}_q,\ud{X}_q)^{-1}$ where $q$ is the parent of $p$.
\end{enumerate}
Then, one sees that $A=K_{\rm{h}}$. Clearly, $A$ is symmetric. Moreover, such a construction ensures that the preconditions of Theorems~\ref{thm:invA} and~\ref{thm:cholA} be satisfied.

In this section, we consider two operations with the vector $\bm{v}=k_{\rm{h}}(X,\bm{x})$, where $\bm{x}\notin X$ is an out-of-sample (i.e., unobserved site). The quantities of interest are
\begin{enumerate}
\item the inner product $\bm{w}^T\bm{v}$ for a general length-$n$ vector $\bm{w}$; and
\item the quadratic form $\bm{v}^T\widetilde{A}\bm{v}$, where $\widetilde{A}$ is a \emph{symmetric} recursively low-rank matrix with the same partitioning tree $T$ and integer $r$ as that used for constructing $k_{\rm{h}}$.
\end{enumerate}
For the quadratic form, in practical use $\widetilde{A}=K_{\rm{h}}^{-1}$, but the algorithm we develop here applies to a general symmetric $\widetilde{A}$. The inner product is used to compute prediction (first equation of~\eqref{eqn:kriging}) whereas the quadratic form is used to estimate standard error (second equation of~\eqref{eqn:kriging}).

The detailed algorithms are presented in the appendix. Similar to those in the preceding section, they are organized as tree algorithms. The difference is that both algorithms in this section are split into a preprocessing computation independent of $\bm{x}$ and a separate $\bm{x}$-dependent computation. The preprocessing still consists of tree traversals that visit all nodes of the hierarchy tree, but the $\bm{x}$-dependent computation visits only one path that connects the root and the leaf node that $\bm{x}$ lies in. In all cases, one needs not explicitly construct the vector $\bm{v}$, which otherwise costs $O(n)$ storage.

\section{Cost Analysis}\label{sec:cost.analysis}
All the recipes and algorithms developed in this work apply to a general partitioning of the domain $S$. As is usual, if the tree is arbitrary, cost analysis of many tree-based algorithms is unnecessarily complex. To convey informative results, here we assume that the partitioning tree $T$ is binary and perfect and the associated partitioning of the point set $X$ is balanced. That is, with some positive integer $n_0$, $|X_i|=n_0$ for all leaf nodes $i$. Then, with a partitioning tree of height $h$, the number of points is $|X|=n=n_02^h$. We assume that the number of landmark points, $r$, is equal to $n_0$ for simplicity.

Since the factors $A_{ii}$, $U_i$ and $V_i$ are stored in the leaf nodes $i$ and $\Sigma_p$, $W_p$, and $Z_p$ are stored in the nonleaf nodes $p$ (in fact, at the root there is no $W_p$ or $Z_p$), the storage is clearly
\[
\underbrace{(2^h)(n_0^2)}_{\text{for } A_{ii}}
+\underbrace{2(2^h)(n_0r)}_{\text{for } U_i \text{ and } V_i}
+\underbrace{(2^h-1)(r^2)}_{\text{for } \Sigma_p}
+\underbrace{2(2^h-2)(r^2)}_{\text{for } W_p \text{ and } Z_p}
=O(nr).
\]
An alternative way to conclude this result is that the tree has $O(n/r)$ nodes, each of which contains an $O(1)$ number of matrices of size $r\times r$. Therefore, the storage is $O(n/r\times r^2)=O(nr)$. This viewpoint also applies to the additional storage needed when executing all the matrix algorithms, wherein temporary vectors and matrices are allocated. This additional storage is $O(r)$ or $O(r^2)$ per node, hence it does not affect the overall assessment $O(nr)$.

The analysis of the arithmetic cost of each matrix operation is presented in the appendix. In brief summary, matrix construction is $O(n\log n+nr^2)$, matrix-vector multiplication is $O(nr)$, matrix inversion and Cholesky-like factorization are $O(nr^2)$, determinant computation is $O(n/r)$, inner product is $O(r^2\log_2(n/r))$ with $O(nr)$ preprocessing, and quadratic form is $O(r^2\log_2(n/r))$ with $O(nr^2)$ preprocessing.

We informally say that the computational cost of the proposed work is $O(n)$, omitting the dependency on $r$. From the function point of view, the quality of $k_{\rm{h}}$ is independent of data. It is a valid covariance function for any positive integer $r$. Hence, one may use a fixed $r$ and apply $k_{\rm{h}}$ to increasingly more data (e.g., increasingly dense sampling within a fixed domain). It is in this sense that the matrix operations are linear in $n$, although we recognize that for some purposes, one may want to consider allowing $r$ to increase with $n$.

\section{Connections and Distinctions to Hierarchical Matrices}\label{sec:compare.hierarchical.matrix}
The proposed recursively low-rank matrix structure builds on a number of previous efforts. For decades, researchers in scientific computing have been keenly developing fast methods for multiplying a dense matrix with a vector, $K\bm{y}$, where the matrix $K$ is defined based on a kernel function (e.g., Green's function) that resembles a covariance function. Notable methods include the tree code~\citep{Barnes1986}, the fast multipole method (FMM)~\citep{Greengard1987,Sun2001}, hierarchical matrices~\citep{Hackbusch1999,Hackbusch2002,Boerm2003}, and various extensions~\citep{Gimbutas2002,Ying2004,Chandrasekaran2006,Martinsson2007,Fong2009,Ho2013,Ambikasaran2014,March2015}. These methods were either co-designed, or later generalized, for solving linear systems $K^{-1}\bm{y}$. They are all based on a hierarchical partitioning of the computation domain, or equivalently, a hierarchical block partitioning of the matrix. The diagonal blocks at the bottom level remain unchanged but (some of) the off-diagonal blocks are low-rank approximated. The differences, however, lie in the fine details, including whether all off-diagonal blocks are low-rank approximated or the ones immediately next to the diagonal blocks should remain unchanged;
whether the low-rank factors across levels share bases;
and how the low-rank approximations are computed.

The aim of this work is an approach applicable to as many computational components as possible of GRF. Hence, the aforementioned design details necessarily differ from those for other applications. Moreover, certain compromises may need to be made for a broad coverage; for example, a structure optimal for kriging is out of the question if not generalizable to likelihood calculation. The rationale of our design choice is best conveyed through comparing with related methods. Our work distinguishes from them in the following aspects.

\paragraph{Function versus matrix.} We explicitly define the covariance function on $\real^d\times\real^d$, which is shown to be (strictly) positive definite. Whereas the related methods are all understood as matrix approximations, to the best of our knowledge, none of these works considers the underlying kernel function that corresponds to the approximate matrix. The knowledge of the underlying function is important for out-of-sample extensions, because, for example in kriging~\eqref{eqn:kriging}, one should approximate also the vector $\bm{k}_0$ in addition to the matrix $K$.

One may argue that if $K$ is well approximated (e.g., accurate to many digits), then it suffices to use the nonapproximate $\bm{k}_0$ for computation. It is important to note, however, that the matrix approximations are elementwise, which does not guarantee good spectral approximations. As a consequence, numerical error may be disastrously amplified through inversion, especially when there is no or a small nugget effect. Moreover, using the nonapproximate $\bm{k}_0$ for computation will incur a computational bottleneck if one needs to krige a large number of sites, because constructing the vector $\bm{k}_0$ alone incurs an $O(n)$ cost.

On the other hand, we start from the covariance function and hence one needs not interpret the proposed approach as an approximation. \emph{All the linear algebra computations are exact in infinite precision, including inversion and factorization}. Additionally, positive definiteness is proved. Few methods under comparison hold such a guarantee.

\paragraph{Positive definiteness.} A substantial flexibility in the design of methods under comparison is the low-rank approximation of the off-diagonal blocks. If the approximation is algebraic, the common objective is to minimize the approximation error balanced with computational efficiency (otherwise the standard truncated singular value decomposition suffices). Unfortunately, rarely does such a method maintain the positive definiteness of the matrix, which poses difficulty for Cholesky-like factorization and log-determinant computation. A common workaround is some form of compensation, either to the original blocks of the matrix~\citep{Bebendorf2007} or to the Schur complements~\citep{Xia2010a}. Our approach requires no compensation because of the guaranteed positive definiteness.

\paragraph{Matrix structures and algorithms.} The fine distinctions in matrix structures lead to substantially different algorithms for matrix operations, if even possible. Our structure is almost the same as that of HSS matrices~\citep{Chandrasekaran2006,Xia2010} and of H$^2$ matrices with weak admissiblity~\citep{Hackbusch2002}, but distant from that of tree code~\citep{Barnes1986}, FMM~\citep{Greengard1987}, H matrices~\citep{Hackbusch1999}, and HODLR matrices~\citep{Ambikasaran2014}. Whereas fast matrix-vector multiplications are a common capability of different matrix structures, the picture starts to diverge for solving linear systems: some structures (e.g., HSS) are amenable for direct factorizations~\citep{Chandrasekaran2006a,Xia2010a,Li2012,Wang2013}, while the others must apply preconditioned iterative methods. An additional complication is that direct factorizations may only be approximate, and thus if the approximation is not sufficiently accurate, it can serve only as a preconditioner but cannot be used in a direct method~\citep{Iske2017}. Then, it will be nearly impossible for these matrix structures to perform Cholesky-like factorizations accurately.

In this regard, our matrix structure is the most clean. Thanks to the property that the matrix inverse and the Cholesky-like factor admit the same structure as that of the original matrix, all the matrix operations considered in this work are exact. Moreover, the explicit covariance function also allows for the development of $O(\log n)$ algorithms for computing inner products and quadratic forms, which, to the best of our knowledge, has not been discussed in the literature for other matrix structures.

\paragraph{Translation from function to matrix.} In the proposed approach, the factors are defined by exploiting the base covariance function, as opposed to HSS and H$^2$ approaches where the factors are generally computed through algebraic factorization and approximation. The delicate definition of the factors ensures positive definiteness, which is lacked by the algebraic methods and even by the methods that exploit the base kernel (e.g., \citet{Fong2009}). The guarantee of positive definiteness necessitates certain sacrifice in approximation accuracy. Thus, the proposed approach is well suited for GRF but for other applications, such as solving partial differential equations, where more specialized methods such as HSS and H$^2$ are preferred.
  
\paragraph{Computational costs.} Although most of the methods under this category enjoy an $O(n)$ or $O(n\log^p n)$ (for some small $p$) arithmetic cost, not every one does so. For example, the cost of skeletonization~\citep{Ho2013,Minden2016} is dimension dependent; in two dimensions it is approximately $O(n^{3/2})$ and in higher dimensions it will be even higher. In general, all these methods are considered matrix approximation methods, and hence there exists a likely tradeoff between approximation accuracy and computational cost. What confounds the approximation is that the low-rank phenomenon exhibited in the off-diagonal blocks fades as the dimension increases~\citep{Ambikasaran2016}. In this regard, it is beneficial to shift the focus from covariance matrices to covariance functions where approximation holds in a more meaningful sense. We conduct experiments to show that predictions and likelihoods are well preserved with the proposed approach.

\section{Practical Considerations}\label{sec:practical}
So far, we have presented a hierarchical framework for constructing valid covariance functions and revealed their appealing computational consequences. The framework is general but there remain instantiations for specific use. In this section, we discuss details tailored to GRF, a low dimensional use case as opposed to the more general (often high-dimensional) case of reproducing kernel Hilbert space.

\subsection{Partitioning of Domain}
For GRF, the sampling sites often reside on a regular grid or a structured (e.g., triangular) mesh. Large spatial datasets with irregular locations commonly occur in remote sensing, although even in this setting, there is usually substantial regularity in the locations due to, for example, the periodicity in a polar-orbiting satellite. When the sites are on a regular grid, a natural choice of the partitioning is axis aligned and balanced. We recommend the following bounding box approach: Begin with the bounding box of the grid, select the longest dimension, cut it into equal halves, and repeat. If the number of grid points along the partitioning dimension in each partitioning is even, the procedure results in a perfect binary tree, whose leaf nodes have exactly the same bounding box volume and the same number of sites. If the number of grid points is odd in some occasion, one shifts the cutting point by half the grid spacing, so that the sampling sites in the middle are not cut.

This bounding box approach straightforwardly generalizes to the mesh or random configuration: Each time the longest dimension of the bounding box is selected and the box is cut into two halves, each of which contains approximately the same number of sampling sites. For random points without exploitable structures, the resulting partitioning tree is known as the k-d tree~\citep{Bentley1975}.

\subsection{Landmark Points}
Assume that the partitioning tree is balanced. As explained in the cost analysis, we consolidate the two parameters, leaf size $n_0$ and the number of landmark points, $r$, into one for convenience. To achieve so, we set the tree height $h$ to be some integer such that the leaf size $n_0=n/2^h$ is greater than or equal to $r$ but less than $2r$. Even if the partitioning is not balanced, the same effect can still be achieved: the recursive partitioning is terminated when each leaf size is $\ge r$ but $<2r$.

The appropriate $r$ is case dependent. There exists a tradeoff between approximation accuracy and computational cost. The larger $r$, the closer $k_{\rm{h}}$ is to $k$ but the more expensive is the computation (the cost of matrix-vector multiplication is linear in $r$, whereas those for inversion, Cholesky, inner product, and quadratic forms are all quadratic in $r$). Although there exists analysis (see, e.g., \citet{Drineas2005}) on the approximation error of the covariance matrix under Nystr\"{o}m approximation (which is part of our one-level construction), extending it to the error analysis of kriging or likelihood is challenging, let alone to the analysis under the multilevel setting. For empirical evidence, we show later a computational example of the kriging error and the log-likelihood, as $r$ varies. We suggest that in practice, one sets $r$ through balancing the tolerable error (which may be estimated, for example, by using a hold out set) and the computational resources at hand.

The configuration of the landmark points is flexible. Because of the low dimension, a regular grid is feasible. One may set the number of grid points along each dimension to be approximately proportional to the size of the bounding box. An advantage of using regular grids is that the results are deterministic. An alternative is randomization. The landmark points may either be uniformly random within the bounding box, or uniformly sampled from the sampling sites. A later experiment indicates that the random choice yields a worse approximation on average, but the variance is nonnegligible such that sometimes a better approximation is obtained compared with the regular-grid choice.

\section{Numerical Experiments}
In this section, we show a comprehensive set of experiments to demonstrate the practical use of the proposed covariance function $k_{\rm{h}}$ for various GRF computations. These computations are centered around simulated data and data from test functions, based on a simple stationary covariance model $k$. In the next section we will demonstrate an application with real-life data and a more realistic nonstationary covariance model.

The base covariance function $k$ in this section is the Mat\'{e}rn model
\begin{equation}\label{eqn:matern}
k(\bm{x},\bm{x}')=
\frac{10^{\alpha}}{2^{\nu-1}\Gamma(\nu)}
\left(\frac{\sqrt{2\nu}\|\bm{r}\|}{\ell}\right)^{\nu}
\bessel_{\nu}\left(\frac{\sqrt{2\nu}\|\bm{r}\|}{\ell}\right)
+10^{\tau}\cdot\bm{1}(\bm{r}=\bm{0})
\quad\text{with}\quad
\bm{r}=\bm{x}-\bm{x}',
\end{equation}
where $10^{\alpha}$ is the sill, $\ell$ is the range, $\nu$ is the smoothness, and $10^{\tau}$ is the nugget. In each experiment, the vector $\bm{\theta}$ of parameters include some of them depending on appropriate setting. We have reparameterized the sill and the nugget through a power of ten, because often the plausible search range is rather wide or narrow. Note that for the extremely smooth case (i.e., $\nu=\infty$), \eqref{eqn:matern} becomes equivalently the squared-exponential model
\begin{equation}\label{eqn:sq.exp}
k(\bm{x},\bm{x}')=
10^{\alpha}\exp\left(-\frac{\|\bm{r}\|^2}{2\ell^2}\right)
+10^{\tau}\cdot\bm{1}(\bm{r}=\bm{0}).
\end{equation}
We will use this covariance function in one of the experiments. Throughout we assume zero mean for simplicity.

\subsection{Small-Scale Example}\label{sec:exp.closed.loop}
We first conduct a closed-loop experiment whereby data are simulated on a two-dimensional grid from some prescribed parameter vector $\bm{\theta}$. We discard (uniformly randomly) half the data and perform maximum likelihood estimation. The purpose is to verify that the estimated $\widehat{\bm{\theta}}$ is indeed close to the $\bm{\theta}$ that generates the data. Afterward, we perform kriging by using the estimated $\widehat{\bm{\theta}}$ to recover the discarded data. Because it is a closed-loop setting and there is no model misspecification, the kriging errors should align well with the square root of the variance of the conditional distribution (see~\eqref{eqn:kriging}). We do not use a large $n$, since we will compare the results of the proposed method with those from the standard method that requires $O(n^3)$ expensive linear algebra computations.

The prescribed parameter vector $\bm{\theta}$ consists of three elements: $\alpha$, $\ell$, and $\nu$. We choose to use a zero nugget because in some real-life settings, measurements can be quite precise and it is unclear one always needs a nugget effect. This experiment covers such a scenario. Further, note that numerically accurate codes for evaluating the derivatives with respect to $\nu$ are unknown. Such a limitation poses constraints when choosing optimization methods.

Further details are as follows. We simulate data on a grid of size $40\times50$ occupying a physical domain $[-0.8,0.8]\times[-1,1]$, by using prescribed parameters $\alpha=0$, $\ell=0.2$, and $\nu=2.5$. Half of the data are discarded, which results in $n=1000$ sites for estimation and $m=1000$ sites for kriging.


For the proposed method, we build the partitioning tree by using the bounding box approach elaborated in Section~\ref{sec:practical}. We specify the number of landmark points, $r$, to be $125$, and make the height of the partitioning tree $h=\lfloor\log_2(n/r)\rfloor$ such that the number of points in each leaf node is approximately $r$. The landmark points for each subdomain in the hierarchy are placed on a regular grid.

Figure~\ref{fig:exp.closed.loop.1}(a) illustrates the random field simulated by using $k$. With this data, maximum likelihood estimation is performed, by using separately $k$ and $k_{\rm{h}}$. The parameter estimates and their standard errors are given in Table~\ref{tab:exp.closed.loop}. The numbers between the two methods are both quite close to the truth. With the estimated parameters, kriging is performed, with the results shown in Figure~\ref{fig:exp.closed.loop.1}(b) and (c). The kriging errors are sorted in the increasing order of the prediction variance. The red curves in the plots are three times the square root of the variance; not surprisingly almost all the errors are below this curve.

\begin{figure}[ht]
\centering
\subfigure[Simulated random field]{
  \includegraphics[width=.31\linewidth]{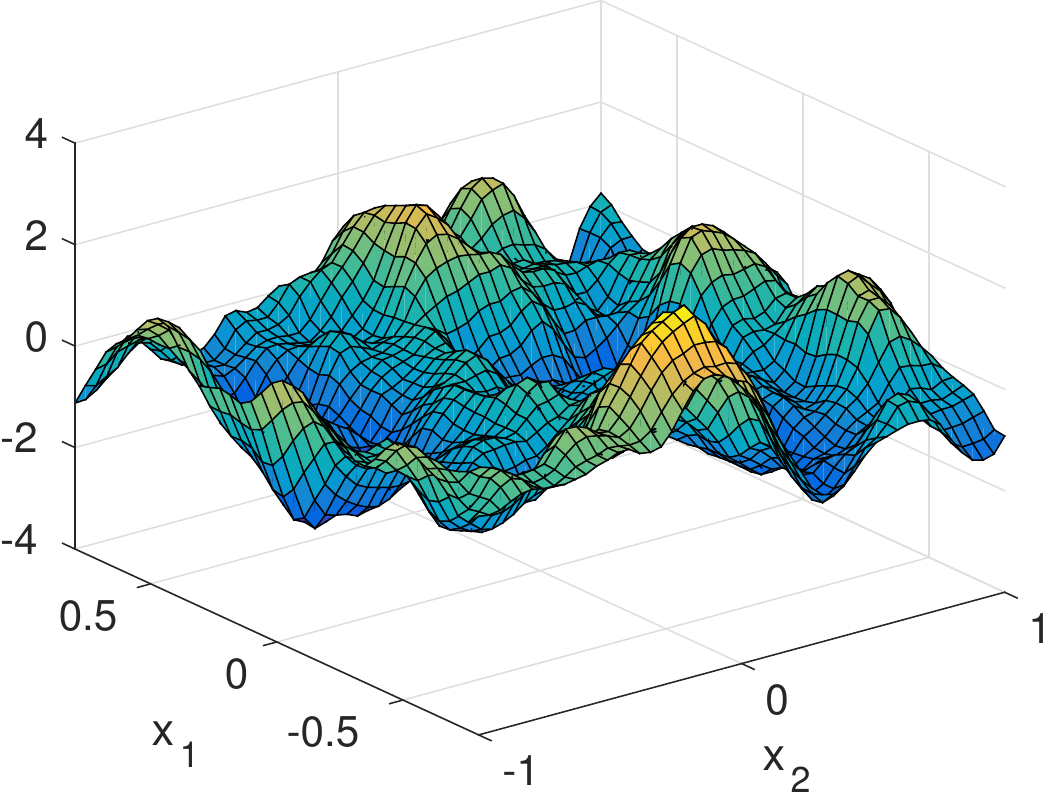}}
\subfigure[Kriging error using $k$]{
  \includegraphics[width=.31\linewidth]{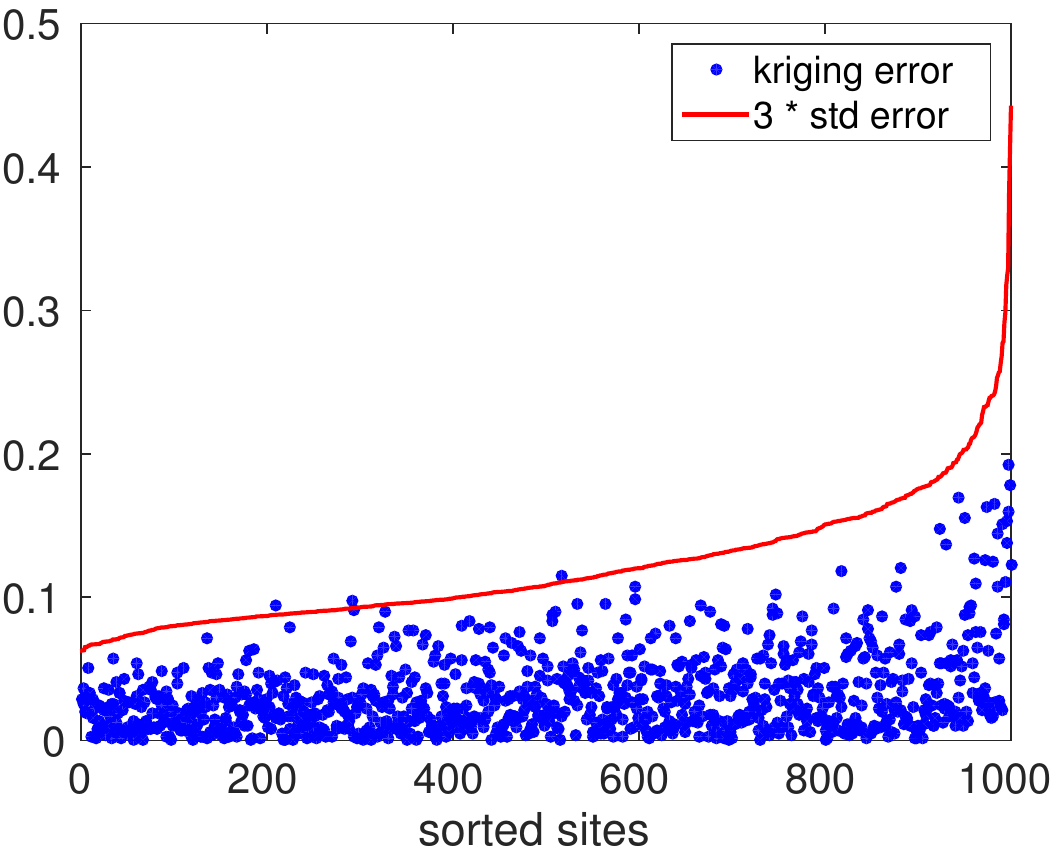}}
\subfigure[Kriging error using $k_{\rm{h}}$]{
  \includegraphics[width=.31\linewidth]{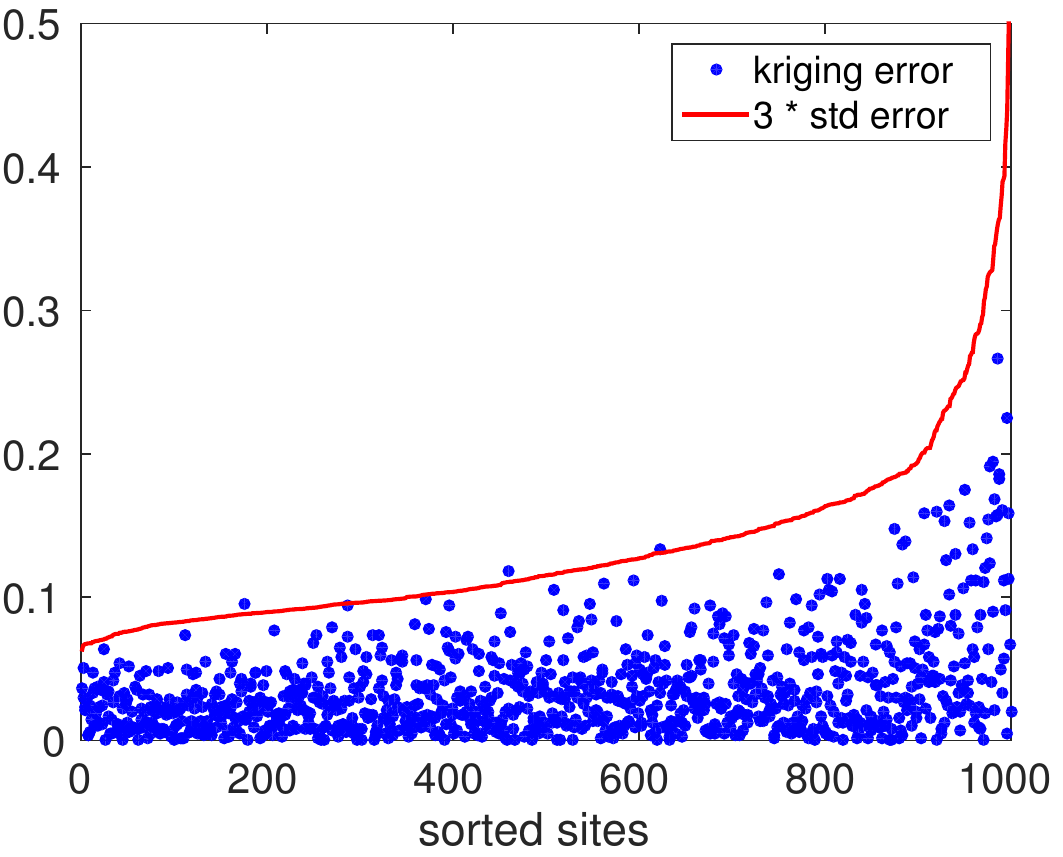}}
\caption{Simulated random field and kriging errors.}
\label{fig:exp.closed.loop.1}
\end{figure}

\begin{table}[ht]
\centering
\caption{True parameters and estimates.}
\label{tab:exp.closed.loop}
\begin{tabular}{lr@{$\,\,$}cc@{$\,\,$}cc@{$\,\,$}c}
\hline
& \multicolumn{2}{c}{$\alpha$}
& \multicolumn{2}{c}{$\ell$}
& \multicolumn{2}{c}{$\nu$}\\
\hline
Truth                     
& $0.000$  &           & $0.200$ &           & $2.50$ &          \\
Estimated with $k$          
& $-0.172$ & $(0.076)$ & $0.182$ & $(0.012)$ & $2.56$ & $(0.11)$ \\
Estimated with $k_{\rm{h}}$ 
& $-0.150$ & $(0.075)$ & $0.186$ & $(0.012)$ & $2.53$ & $(0.11)$ \\
\hline
\end{tabular}
\end{table}

\subsection{Comparison of Log-Likelihoods and Estimates}\label{sec:exp.loglik}
One should note that the base covariance function $k$ and the proposed $k_{\rm{h}}$ are not particularly close, because the number $r$ of landmarks for defining $k_{\rm{h}}$ is only $125$ (compare this number with the number of observed sites, $n=1000$). Hence, if one compares the covariance matrix $K$ with $K_{\rm{h}}$, they agree in only a limited number of digits. However, the reason why $k_{\rm{h}}$ is a good alternative of $k$ is that the shapes of the likelihoods are similar, as well as the locations of the optimum.

\begin{figure}[ht]
\centering
\subfigure[$\ell$-$\nu$ plane]{
  \includegraphics[width=.31\linewidth]{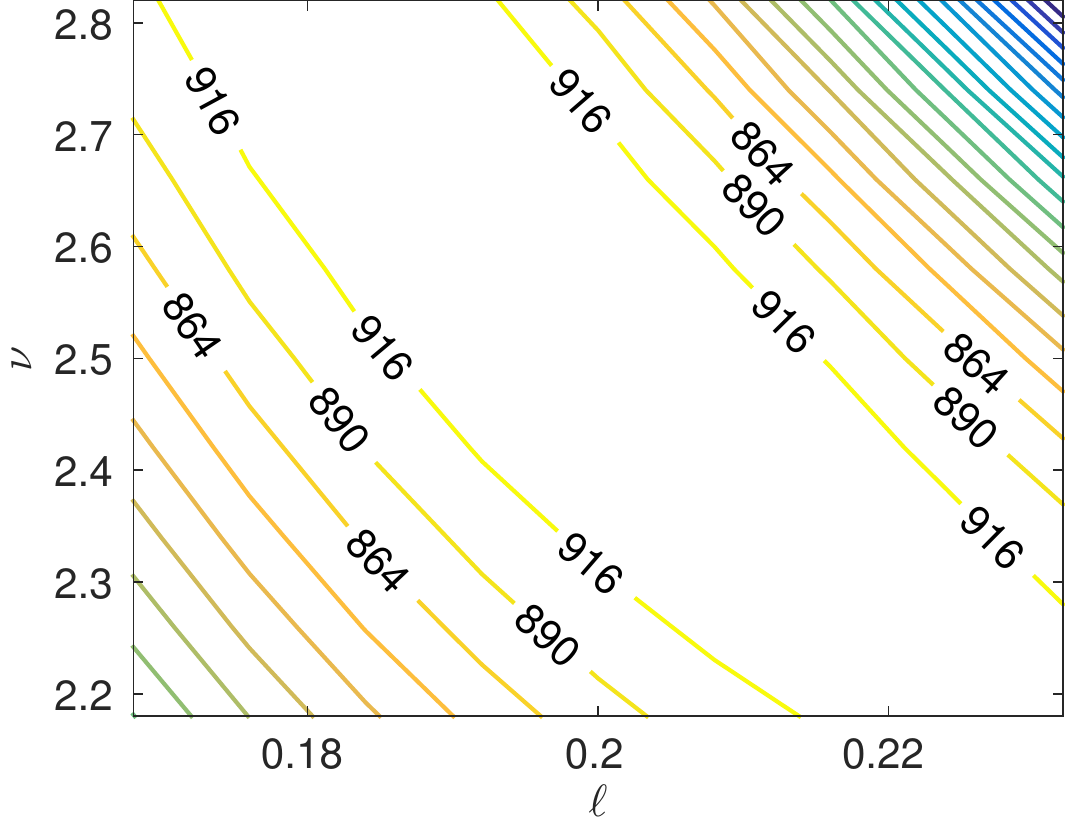}}
\subfigure[$\alpha$-$\nu$ plane]{
  \includegraphics[width=.31\linewidth]{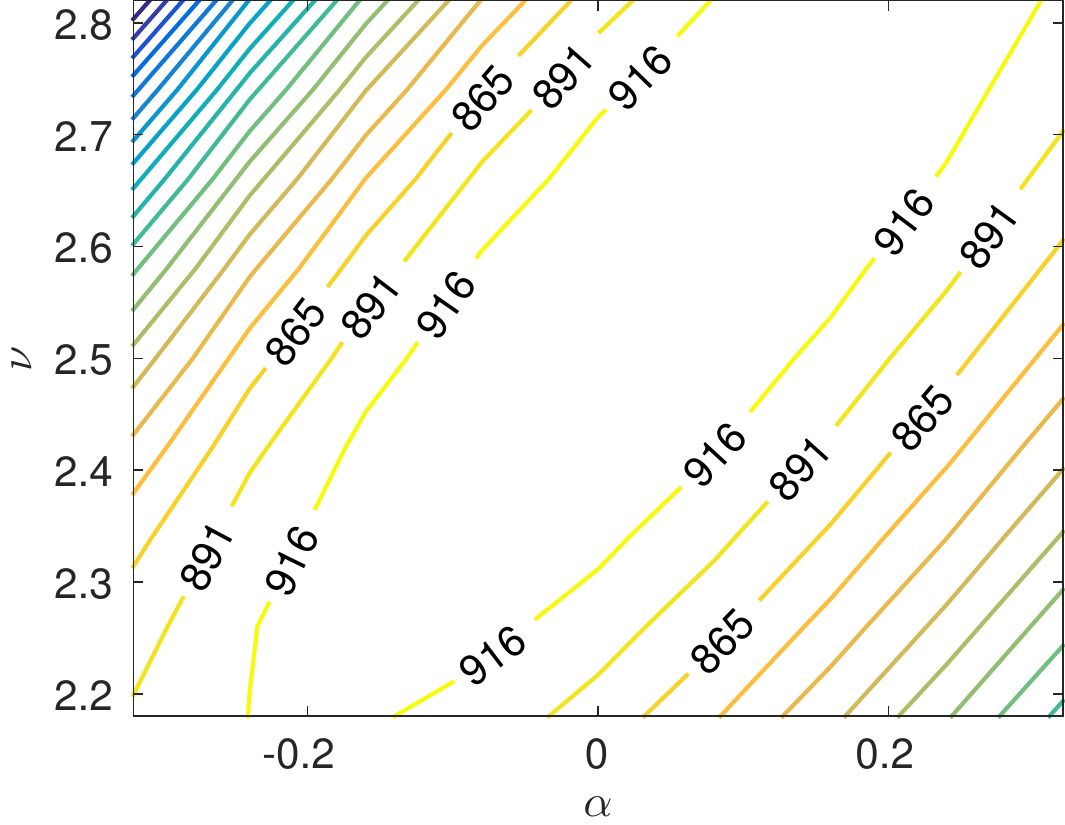}}
\subfigure[$\alpha$-$\ell$ plane]{
  \includegraphics[width=.31\linewidth]{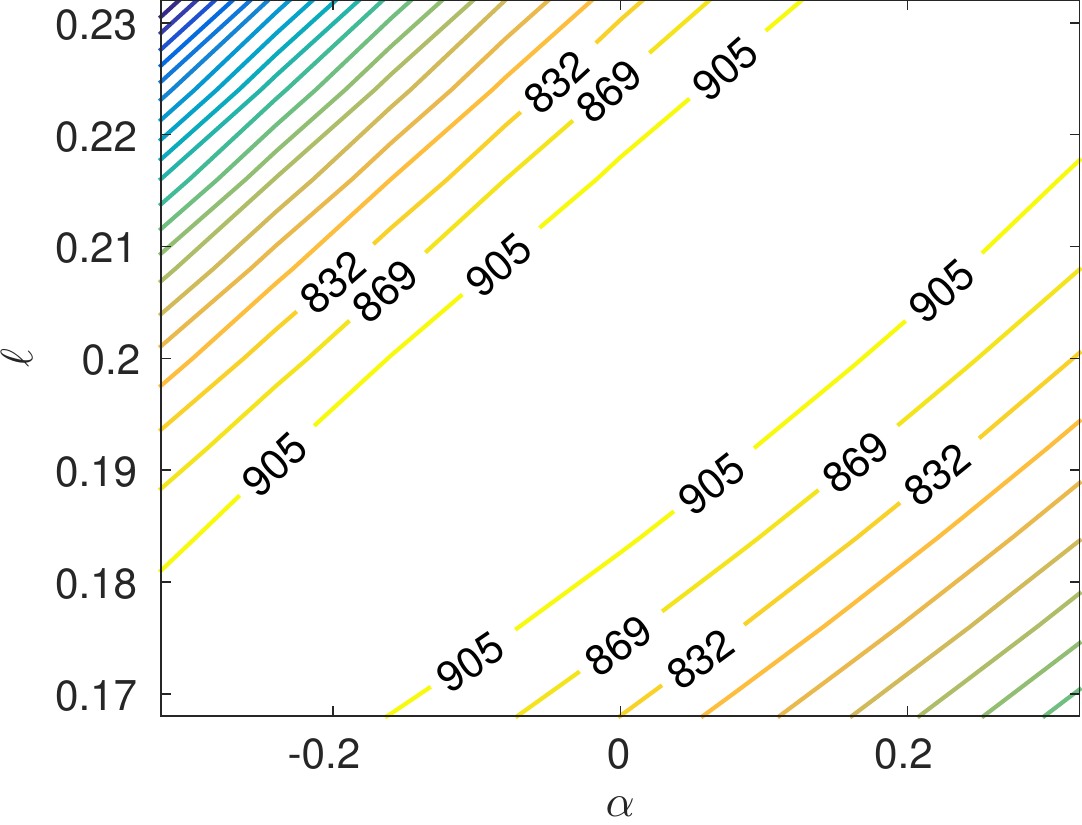}}\\
\subfigure[$\ell$-$\nu$ plane]{
  \includegraphics[width=.31\linewidth]{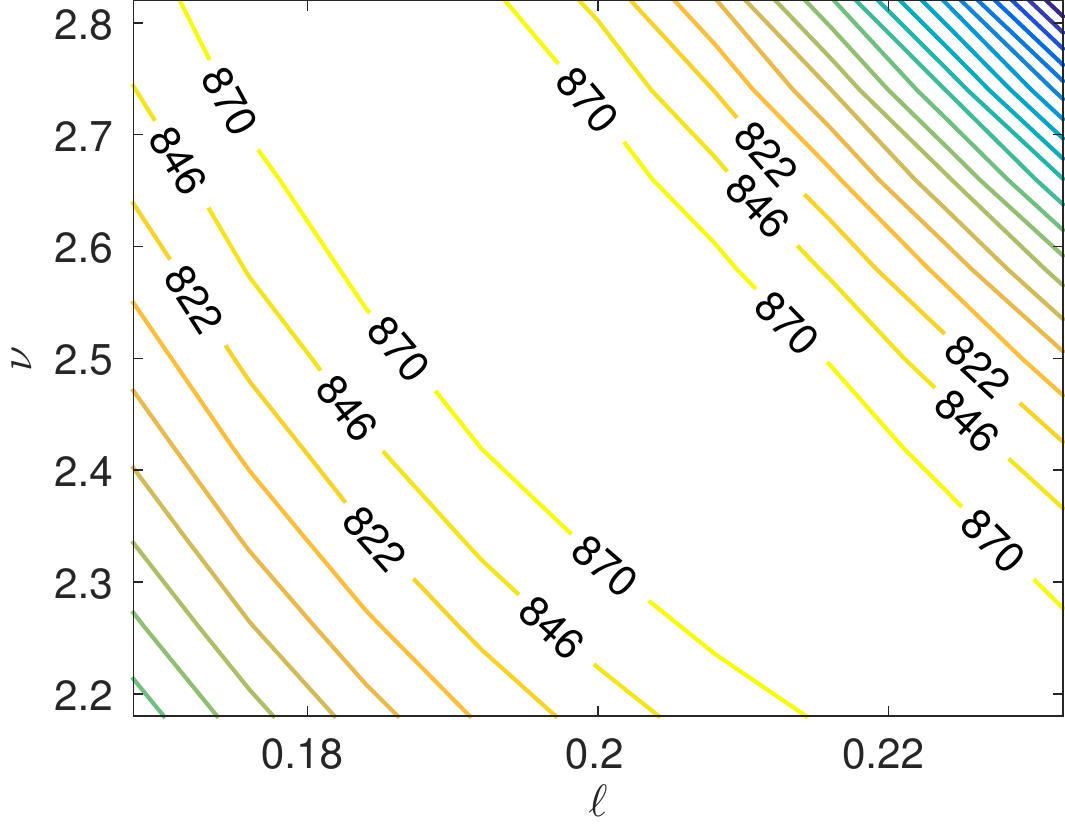}}
\subfigure[$\alpha$-$\nu$ plane]{
  \includegraphics[width=.31\linewidth]{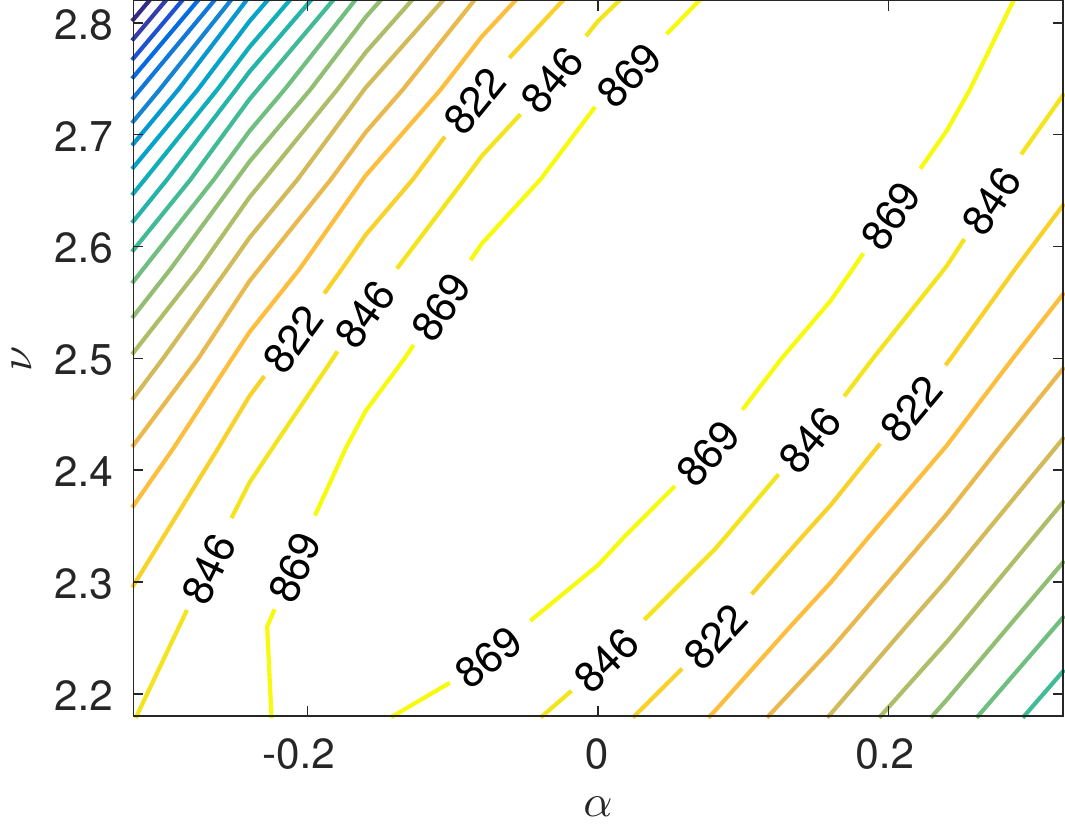}}
\subfigure[$\alpha$-$\ell$ plane]{
  \includegraphics[width=.31\linewidth]{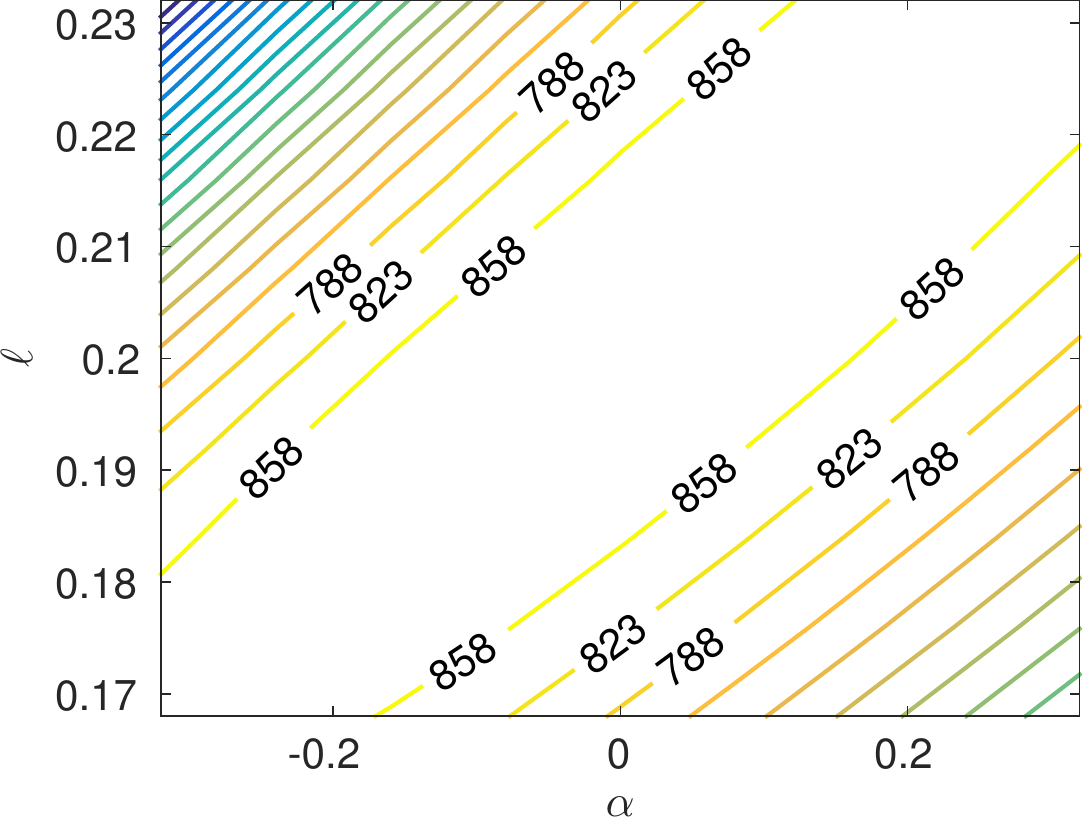}}
\caption{Cross sections of log-likelihood. Top row: base covariance function $k$; bottom row: proposed covariance function $k_{\rm{h}}$.}
\label{fig:exp.closed.loop.3}
\end{figure}

We graph in Figure~\ref{fig:exp.closed.loop.3} the cross sections of the log-likelihood centered at the truth $\bm{\theta}$. The top row corresponds to $k$ and the bottom row to $k_{\rm{h}}$. One sees that in both cases, the center (truth $\bm{\theta}$) is located within a reasonably concave neighborhood, whose contours are similar to each other.

\begin{table}[ht]
\centering
\caption{Difference of estimates and log-likelihoods under $k$ and $k_{\rm{h}}$. The unparenthesized number is the mean and the number with parenthesis is the standard deviation. For reference, the uncertainties (denoted as stderr) of the estimates are listed in the second part of the table.}
\label{tab:exp.loglik}
\begin{tabular}{cccc}
\hline
$|\widehat{\alpha}-\widehat{\alpha}_{\rm{h}}|$
& $|\widehat{\ell}-\widehat{\ell}_{\rm{h}}|$
& $|\widehat{\nu}-\widehat{\nu}_{\rm{h}}|$
& $|\mathcal{L}_k(\widehat{\bm{\theta}})-\mathcal{L}_k(\widehat{\bm{\theta}}_{\rm{h}})|$\\
$0.0120$ $(0.0098)$ & $0.0018$ $(0.0018)$ & $0.0240$ $(0.0211)$ & $0.1151$ $(0.0880)$\\
\hline
\hline
$\text{stderr}(\widehat{\alpha})$
& $\text{stderr}(\widehat{\ell})$
& $\text{stderr}(\widehat{\nu})$\\
$0.0841$ $(0.0050)$ & $0.0137$ $(0.0016)$ & $0.1002$ $(0.0074)$\\
\hline
\end{tabular}
\end{table}

In fact, the maxima of the log-likelihoods are rather close. We repeat the simulation ten times and report the statistics in Table~\ref{tab:exp.loglik}. The quantities with a subscript ``h'' correspond to the proposed covariance function $k_{\rm{h}}$. One sees that for each parameter, the differences of the estimates are generally about $20\%$ of the standard errors of the estimates. Furthermore, the difference of the true log-likelihoods at the two estimates is always substantially less than one unit. These results indicate that the proposed $k_{\rm{h}}$ produces highly comparable parameter estimates with the base covariance function $k$.

\subsection{Landmark Points}\label{sec:exp.landmark}
In the preceding two subsections, we fixed the number of landmark points, $r$, to be $125$ and placed them on a regular grid within each subdomain. Here, we study the effect of $r$ and the locations.

\begin{figure}[ht]
\centering
\subfigure[Using ground truth parameters]{
  \includegraphics[width=.4\linewidth]{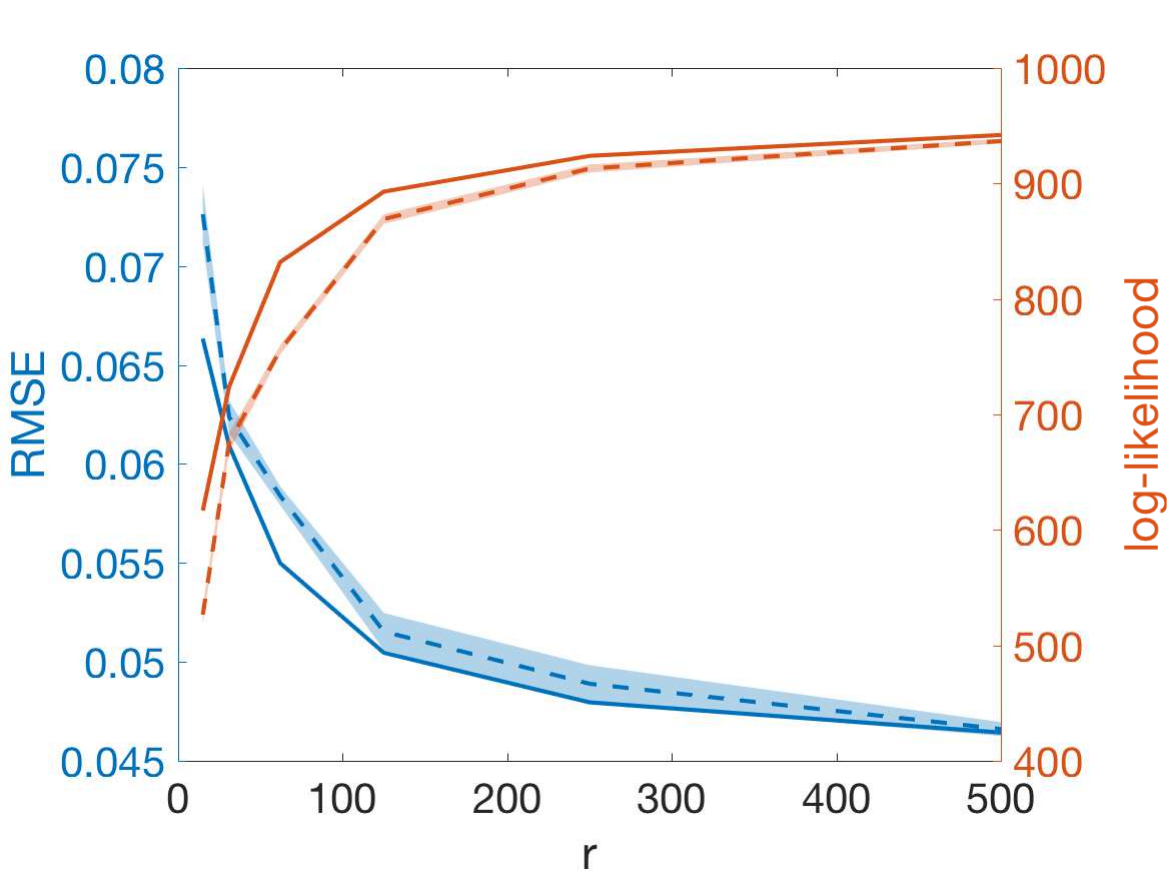}}
\subfigure[Using a different set of parameters]{
  \includegraphics[width=.4\linewidth]{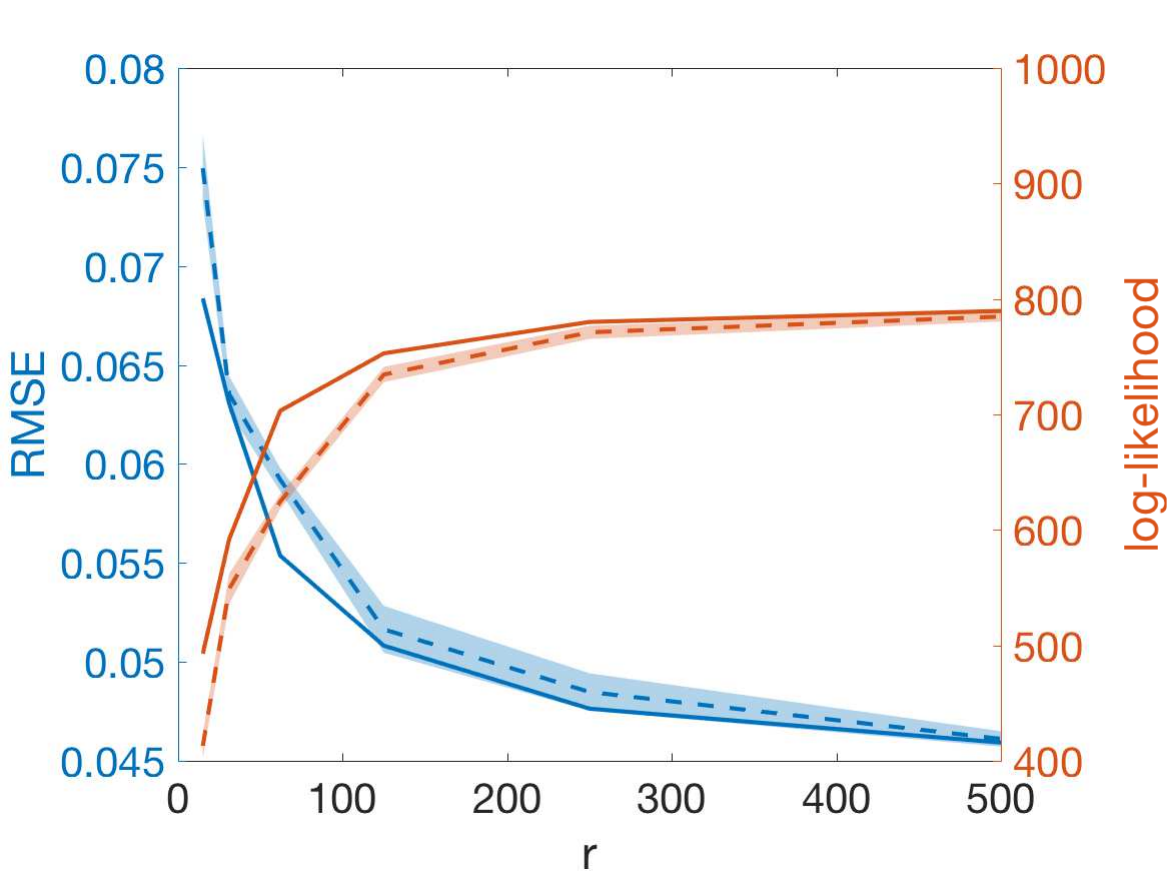}}
\caption{Kriging error and log-likelihood as $r$ varies. The solid curve corresponds to a regular grid configuration of landmark points, whereas the dashed curve with shaded region corresponds to randomized landmark points (repeated 30 times).}
\label{fig:exp.landmark}
\end{figure}

In Figure~\ref{fig:exp.landmark}, we show two plots on the kriging error and the log-likelihood, one obtained by using the ground truth parameters $[\alpha,\ell,\nu]=[0,0.2,2.5]$ and the other by using $[\alpha,\ell,\nu]=[0.2,0.24,2.7]$, which results in a noticeably different covariance function as judged from the likelihood surface exhibited in Figure~\ref{fig:exp.closed.loop.3}. The experimented values of $r$ are $7$, $15$, $31$, $62$, $125$, $250$, and $500$, geometrically progressing toward the number of observed sites, $n=1000$. The solid curve corresponds to a regular grid of landmark points, whereas the dashed curve corresponds to the randomized choice, with one times standard deviation shown as a shaded region. ``RMSE'' denotes root mean squared error.

One sees that the error decreases monotonically as $r$ increases. There thus forms a tradeoff between error and time, since the computational cost is quadratic in $r$. In this particular case, it appears that $125$ yields a significant decrease in RMSE while being reasonably small. The likelihood shows a similar trend of change as $r$ varies (except that it increases rather than decreases). Moreover, the randomized choice of landmark points is inferior to the regular-grid choice, considering the mean and standard deviation. However, one should note that if three times standard deviation is considered instead, the shaded region will cover the solid curve for large $r$, indicating that the advantage of regular grid diminishes as $r$ increases. Finally, an interesting observation is that the kriging error remains highly comparable when one uses less accurate covariance parameters, although in this case the reduction of likelihood is substantial.

\subsection{Comparison with Nystr\"{o}m and Block-Diagonal Approximation}\label{sec:exp.nystrom}
In this subsection, we compare with two methods: Nystr\"{o}m and block-diagonal approximation. The former is a part of our one-level construction, whereas the latter performs kriging in each fine-level subdomain independently (equivalent to applying a block-diagonal approximation of the covariance matrix $K$). The experiment setting is the same as that of the preceding subsections.

\begin{figure}[ht]
\centering
\subfigure[Compared with Nystr\"{o}m]{
  \includegraphics[width=.4\linewidth]{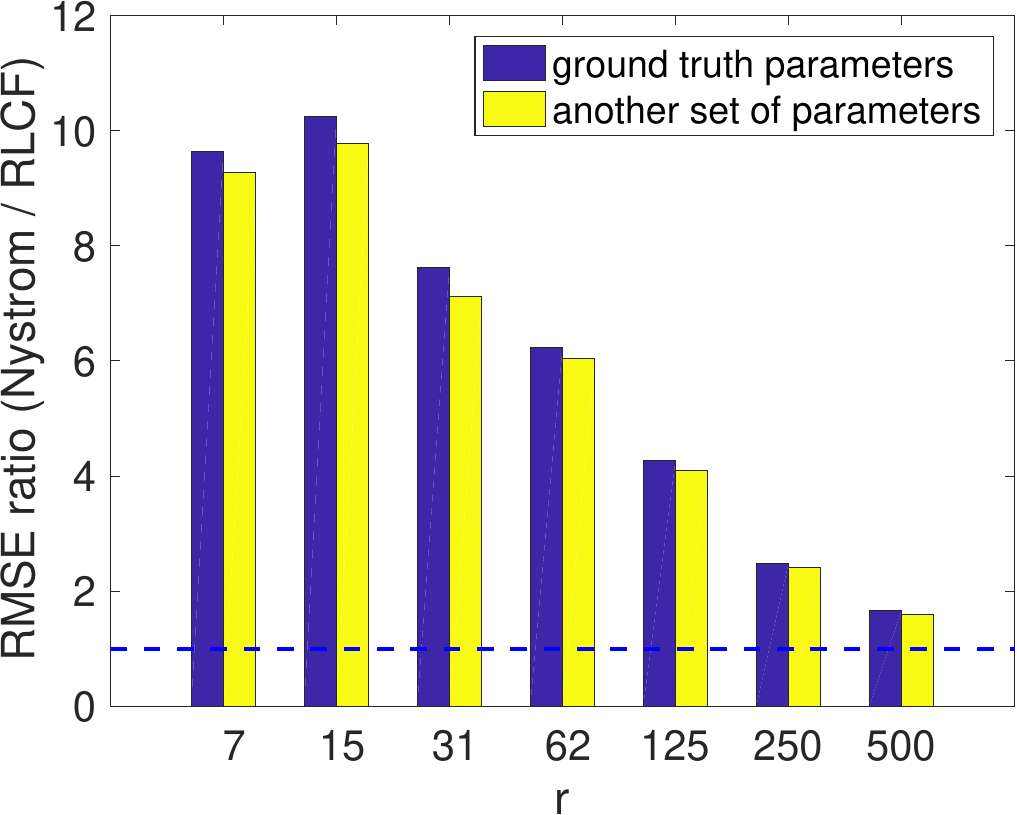}}
\subfigure[Compared with block-diagonal approx.]{
  \includegraphics[width=.4\linewidth]{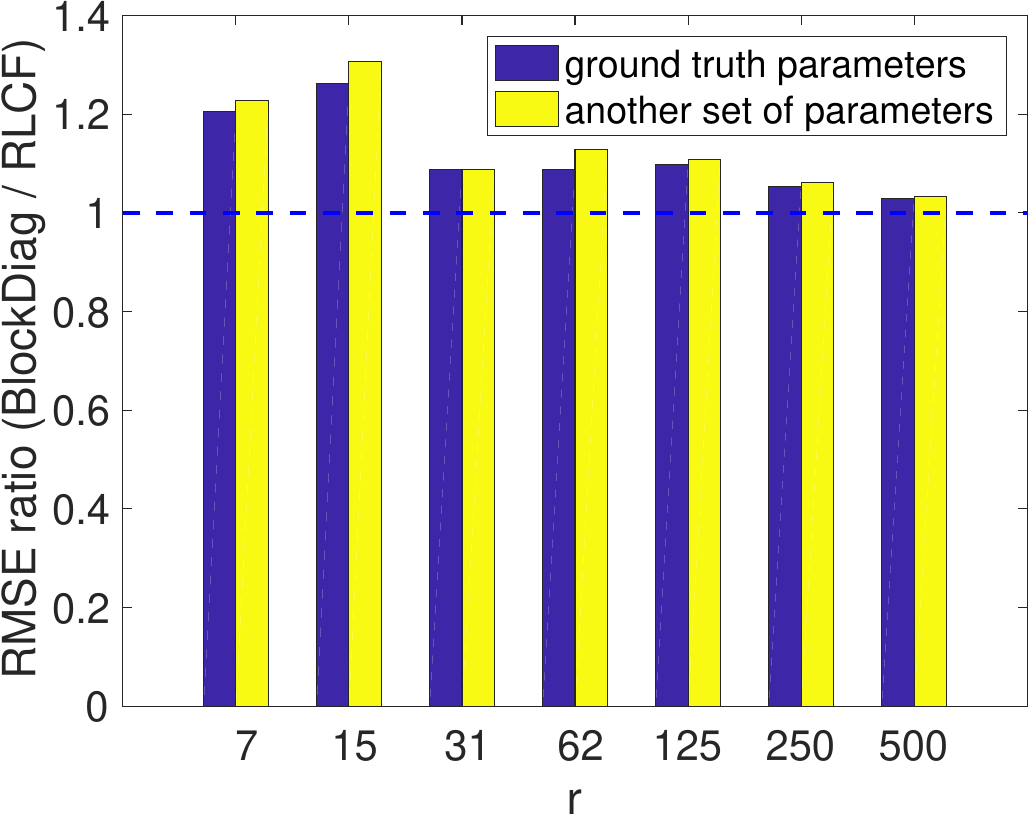}}
\caption{RMSE ratio between a compared method and the proposed method (RLCF). Ground truth parameters are $[\alpha,\ell,\nu]=[0,0.2,2.5]$ and the other set is $[\alpha,\ell,\nu]=[0.2,0.24,2.7]$.}
\label{fig:exp.nystrom}
\end{figure}

Figure~\ref{fig:exp.nystrom}(a) shows the kriging error of Nystr\"{o}m normalized by that of the proposed method. First, all error ratios are greater than one, indicating that the hierarchical approach clearly strengthens the approximation with only one level as in Nystr\"{o}m. Moreover, this observation is consistent regardless of what covariance parameters are used. Interestingly, the ratio is slightly smaller when the used parameters are less accurate, suggesting that one-level approximation appears to suffer less when the parameters are not close to the ground truth. Finally, as $r$ increases, the error ratio generally decreases, which is expected since the number of levels that strengthen the approximation becomes fewer. Nystr\"{o}m performs disastrously in light of the fact that the error ratio is greater than 2 when $r<500$.

Similarly, Figure~\ref{fig:exp.nystrom}(b) shows the kriging error of block-diagonal approximation, normalized. This method performs much better than Nystr\"{o}m, with the normalized errors only slightly greater than 1. Interestingly, contrary to Nystr\"{o}m, this method suffers more when the parameters are not close to the ground truth. Since the method performs essentially local kriging by ignoring the long-range correlation, this phenomenon is expected.

\subsection{Scaling}\label{sec:exp.scaling}
In this subsection, we verify that the linear algebra costs for the proposed method indeed agree with the theoretical analysis. Namely, random field simulation and log-likelihood evaluation are both $O(n)$, and the kriging of $m=n$ sites is $O(n\log n)$. Note that all these computations require the construction of the covariance matrix, which is $O(n\log n)$.

The experiment setting is the same as that of the preceding subsections, except that we restrict the number of log-likelihood evaluations to $125$ to avoid excessive computation. We vary the grid size from $40\times50$ to $640\times800$ to observe the scaling. The random removal of sites has a minimal effect on the partitioning and hence on the overall time. The computation is carried out on a laptop with eight Intel cores (CPU frequency 2.8GHz) and 32GB memory. Six parallel threads are used.

\begin{figure}[ht]
\centering
\subfigure[Random field simulation]{
  \includegraphics[width=.31\linewidth]{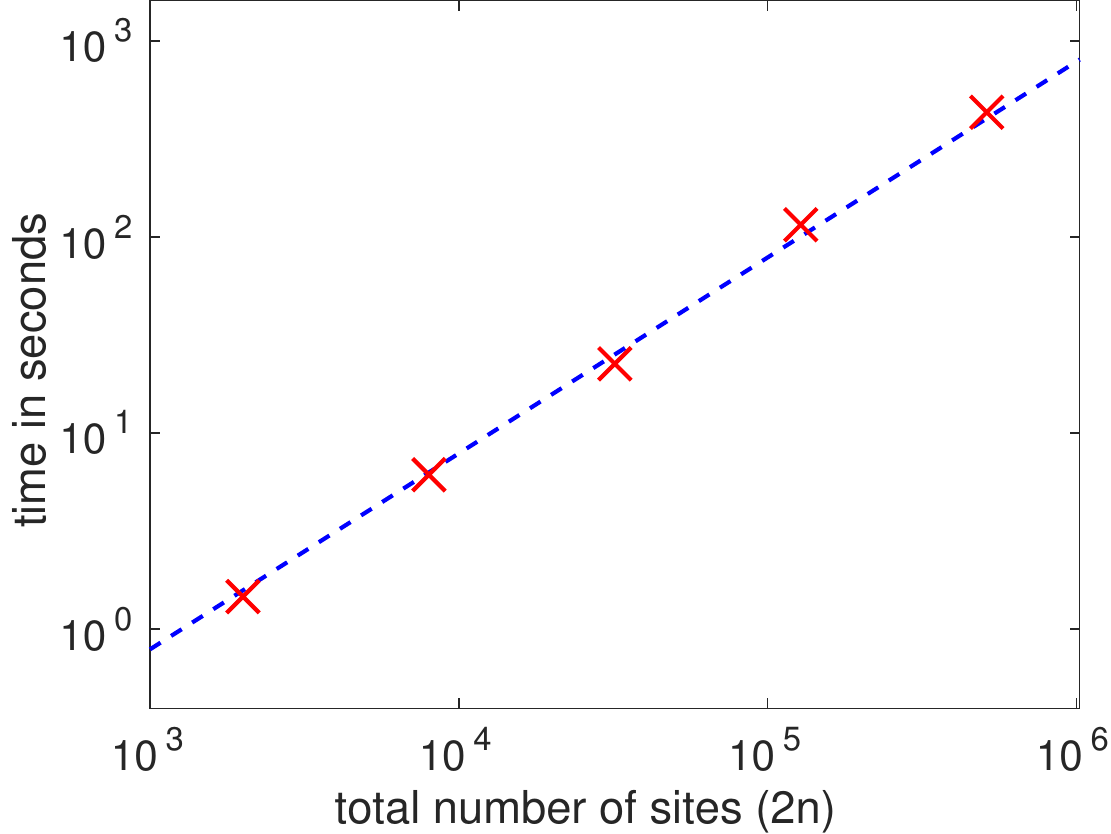}}
\subfigure[125 Log-likelihood evaluations]{
  \includegraphics[width=.31\linewidth]{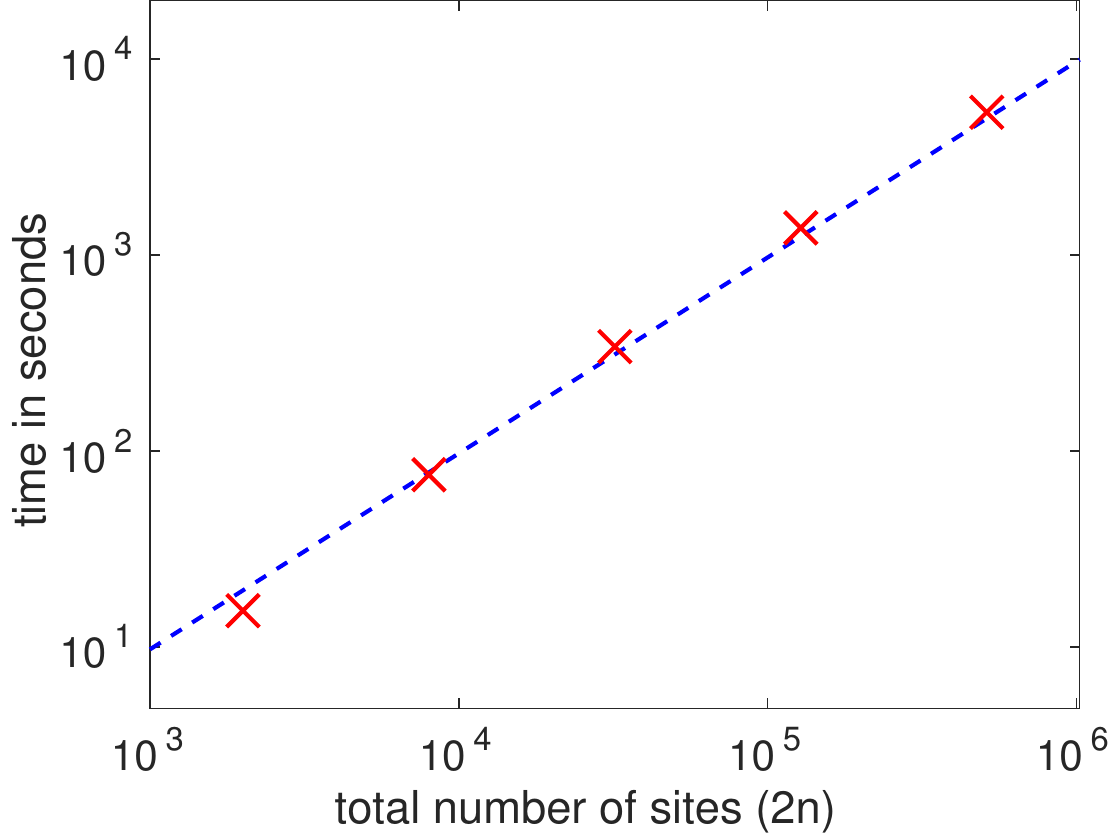}}
\subfigure[Kriging $n$ sites]{
  \includegraphics[width=.31\linewidth]{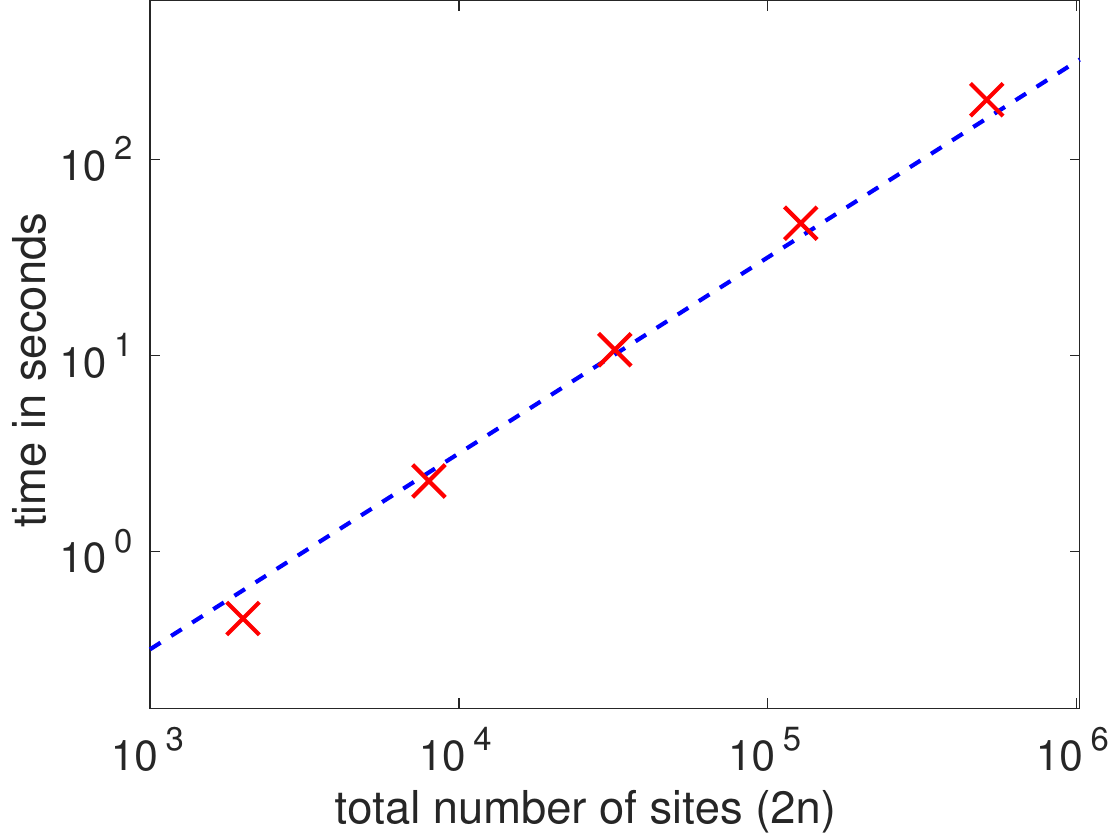}}
\caption{Computation time. The dashed blue line is an $O(n)$ scaling.}
\label{fig:exp.scaling}
\end{figure}

Figure~\ref{fig:exp.scaling} plots the computation times, which indeed well agree with the theoretical scaling. As expected, log-likelihood evaluations are the most expensive, particularly when many evaluations are needed for optimization. The simulation of a random field follows, with kriging being the least expensive, even when a large number of sites are kriged.

\subsection{Large-Scale Example Using Test Function}\label{sec:exp.large.scale}
The above scaling results confirm that handling a large $n$ is feasible on a laptop. In this subsection, we perform an experiment with up to one million data sites. Different from the closed-loop setting that uses a known covariance model, here we generate data by using a test function. We estimate the covariance parameters and krige with the estimated model.

The test function is
\begin{equation}\label{eqn:test.func}
Z(\bm{x}) = \exp(1.4x_1)\cos(3.5\pi x_1)[\sin(2\pi x_2)+0.2\sin(8\pi x_2)]
\end{equation}
on $[0,1]^2$. This function is rather smooth (see Figure~\ref{fig:exp.closed.loop}(a) for an illustration). Hence, we use the squared-exponential model~\eqref{eqn:sq.exp} for estimation. The high smoothness results in a too ill-conditioned matrix; therefore, a nugget is necessary. The vector of parameters is $\bm{\theta}=[\alpha, \ell, \tau]^T$. We inject independent Gaussian noise $\mathcal{N}(0,\,0.1^2)$ to the data so that the nugget will not vanish. As before, we randomly select half of the sites for parameter estimation and the other half for kriging. The number of landmark points, $r$, remains $125$.

Our strategy for large-scale estimation is to first perform a small-scale experiment with the base covariance function $k$ that quickly locates the optimum. The result serves as a reference for later use of the proposed $k_{\rm{h}}$ in the larger-scale setting.
The results are shown in Figure~\ref{fig:exp.closed.loop} (for the largest grid) and Table~\ref{tab:exp.large.scale}.

\begin{figure}[ht]
\centering
\subfigure[Test function]{
  \includegraphics[width=.31\linewidth]{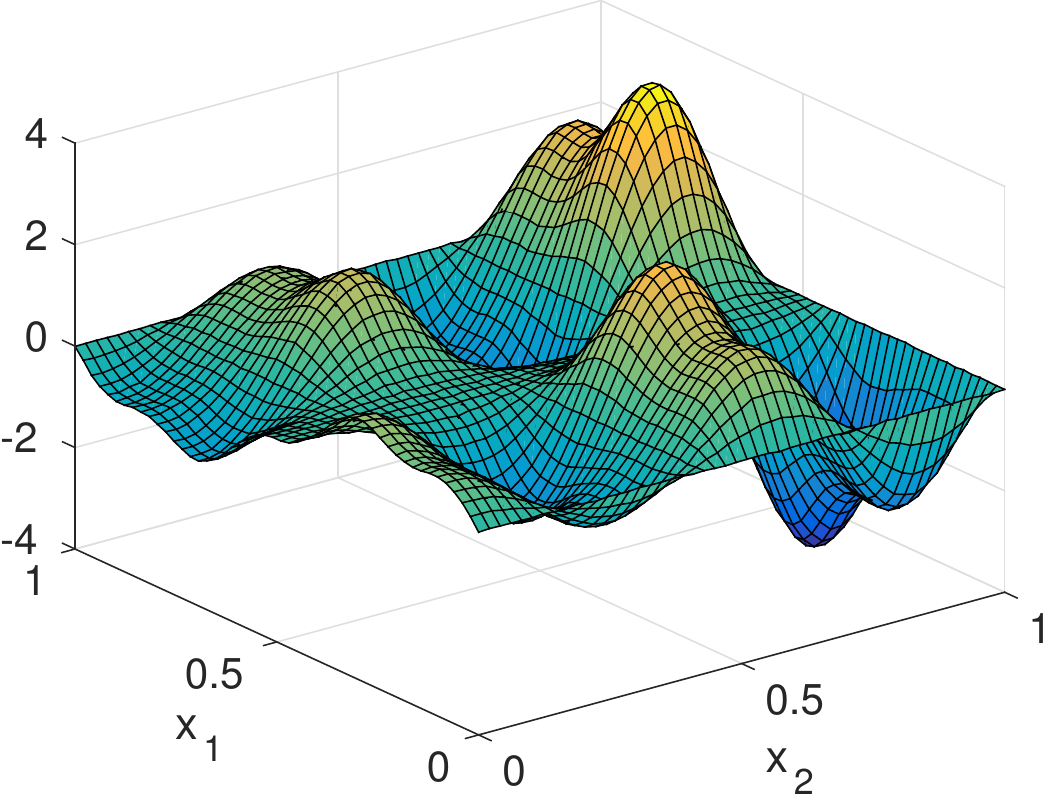}}\hspace{.3cm}
\subfigure[Kriged field]{
  \includegraphics[width=.26\linewidth]{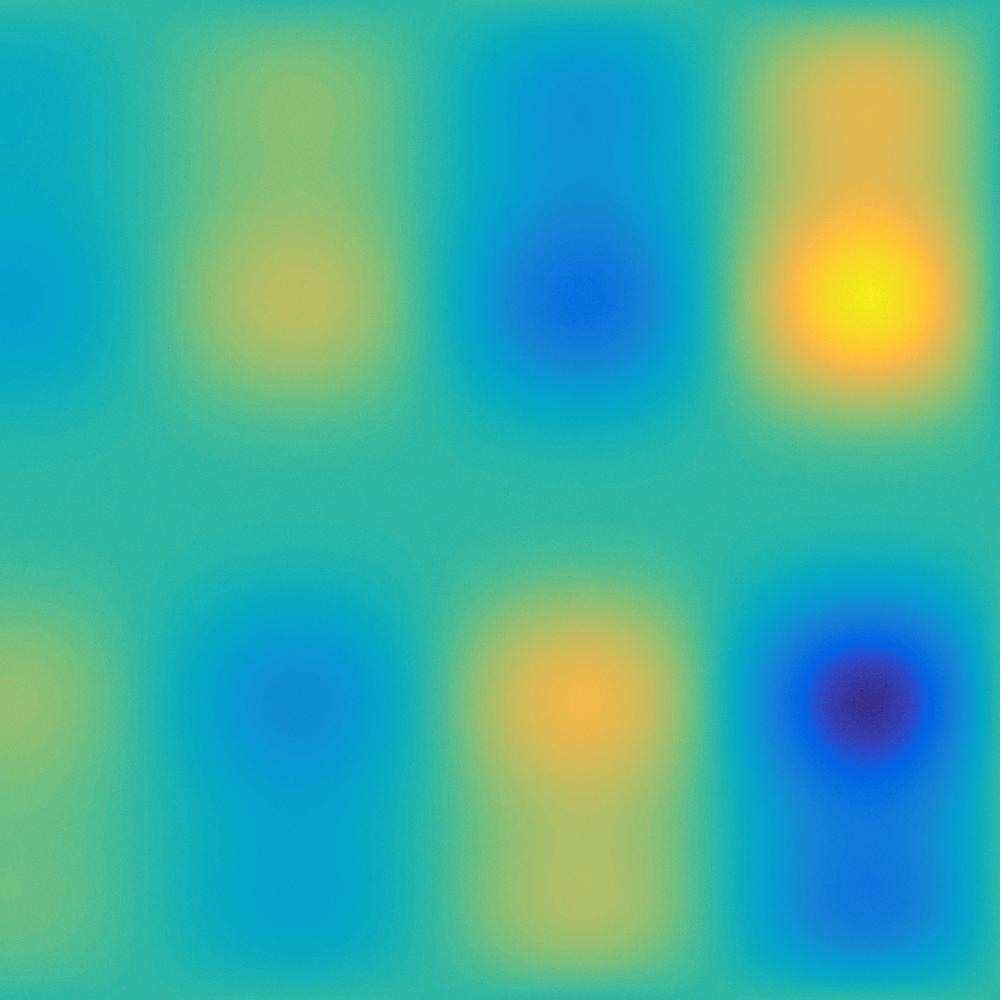}}\hspace{.3cm}
\subfigure[Kriging error]{
  \includegraphics[width=.31\linewidth]{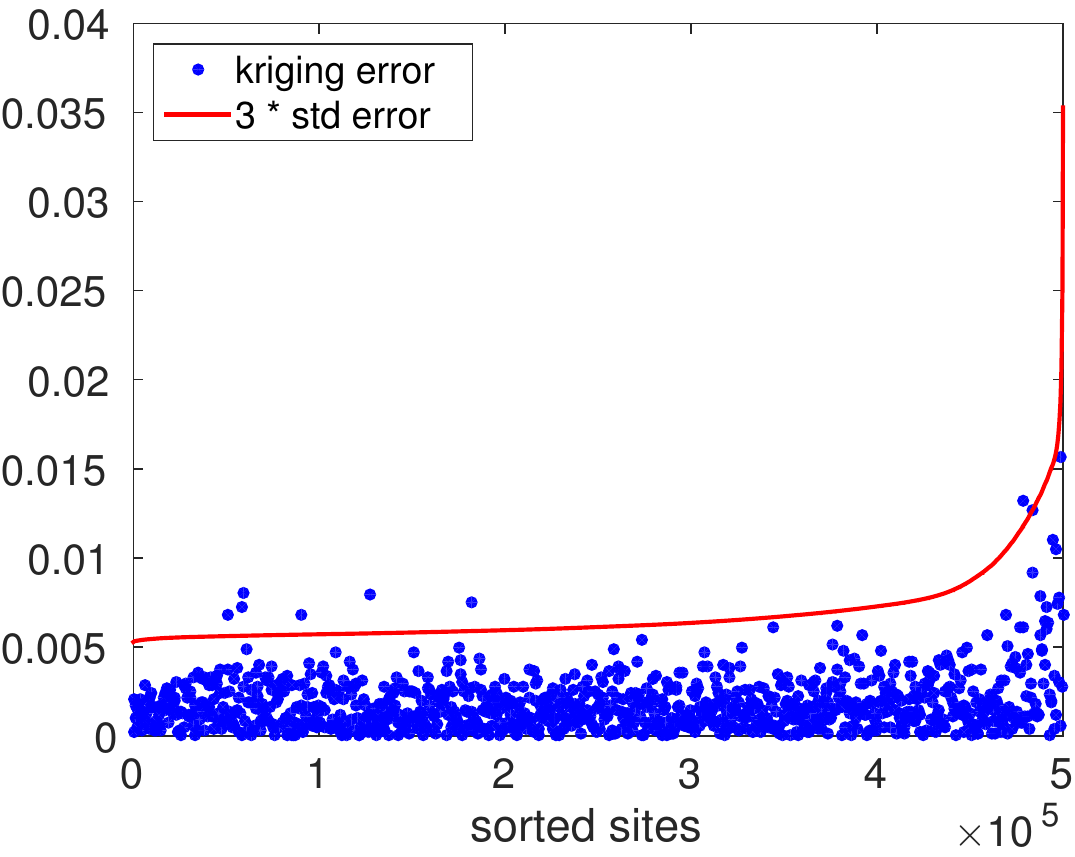}}
\subfigure[$\alpha$-$\ell$ plane]{
  \includegraphics[width=.31\linewidth]{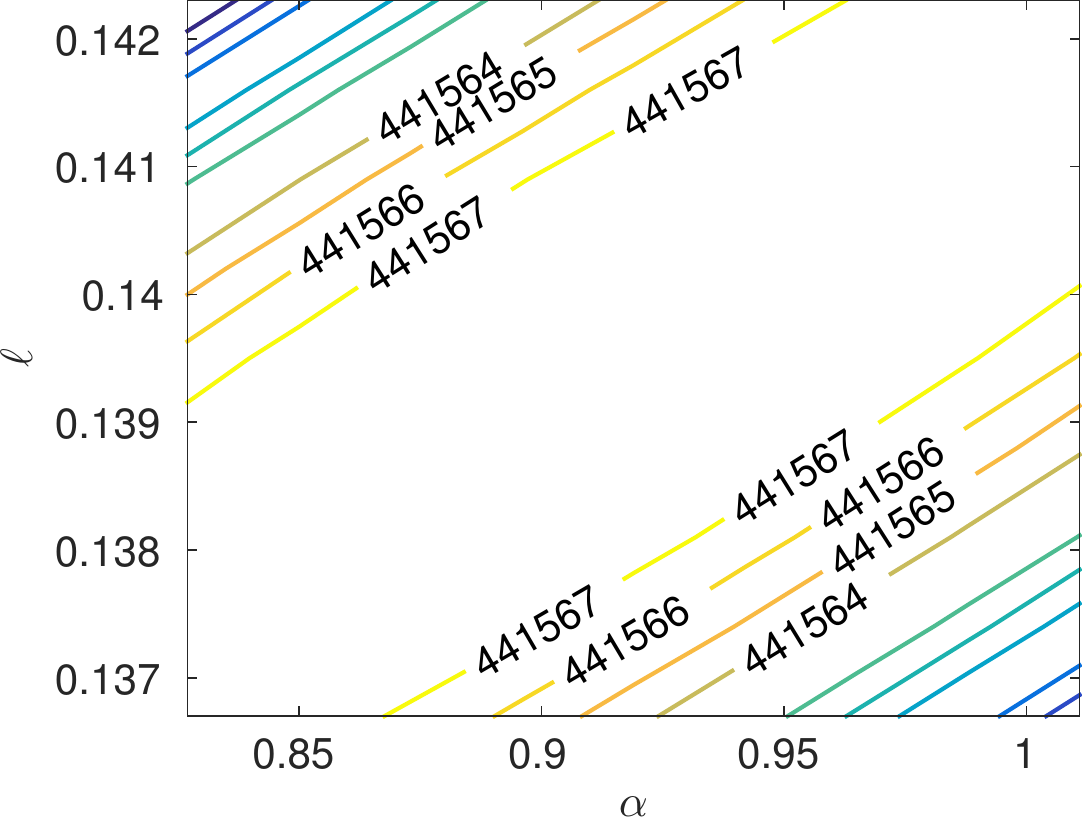}}
\subfigure[$\alpha$-$\tau$ plane]{
  \includegraphics[width=.31\linewidth]{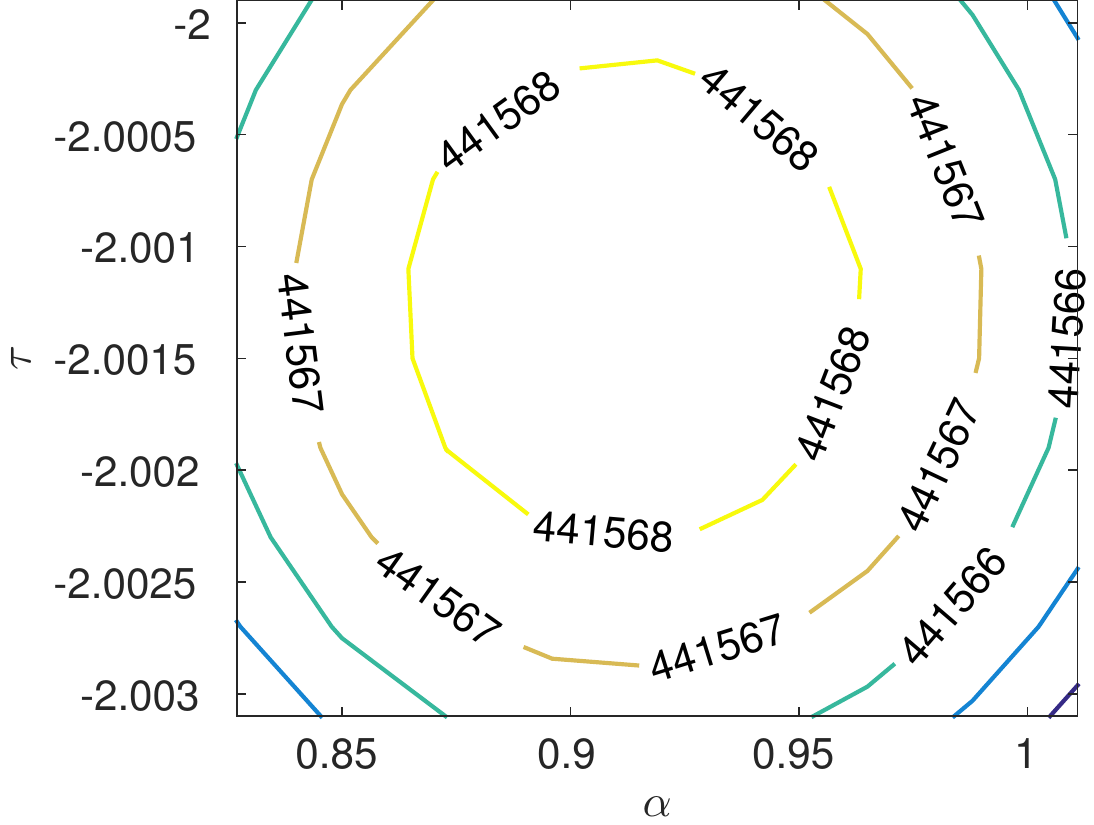}}
\subfigure[$\ell$-$\tau$ plane]{
  \includegraphics[width=.31\linewidth]{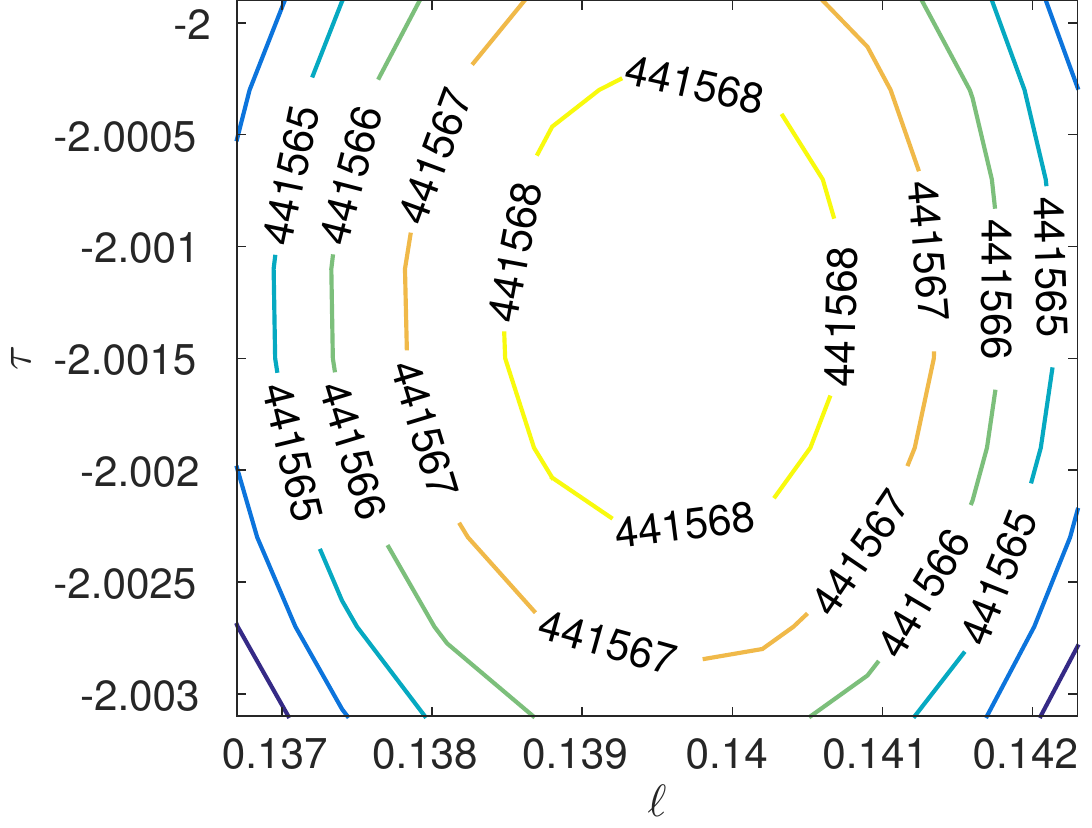}}
\caption{Top row: test function and kriging results; bottom row: log-likelihood. For plot (c), the blue dots are subsampled evenly so that they do not clutter the figure.}
\label{fig:exp.closed.loop}
\end{figure}

\begin{table}[ht]
\centering
\caption{Estimated parameters.}
\label{tab:exp.large.scale}
\begin{tabular}{ccc@{$\,\,$}cc@{$\,\,$}cc@{$\,\,$}c}
\hline
Grid & Est. w/
& \multicolumn{2}{c}{$\widehat{\alpha}$}
& \multicolumn{2}{c}{$\widehat{\ell}$}
& \multicolumn{2}{c}{$\widehat{\tau}$}\\
\hline
$50\times50$ & $k$          
& $0.313$ & $(0.098)$ & $0.1199$ & $(0.0035)$ & $-2.0109$ & $(0.0186)$\\
$100\times100$ & $k_{\rm{h}}$ 
& $0.389$ & $(0.095)$ & $0.1238$ & $(0.0029)$ & $-1.9923$ & $(0.0089)$\\
$1000\times1000$ & $k_{\rm{h}}$ 
& $0.919$ & $(0.134)$ & $0.1395$ & $(0.0031)$ & $-2.0011$ & $(0.0009)$\\
\hline
\end{tabular}
\end{table}

Each of the cross sections of the log-likelihood on the bottom row of Figure~\ref{fig:exp.closed.loop} is plotted by setting the unseen parameter at the estimated value. For example, the $\alpha$-$\ell$ plane is located at $\widehat{\tau}=-2.0011$. From these contour plots, we see that the estimated parameters are located at a local maximum with nicely concave contours in a neighborhood of this maximum. The estimated nugget ($\approx -2$) well agrees with $\log_{10}$ of the actual noise variance. The kriged field (plot (b)) is visually as smooth as the test function. The kriging errors for predicting the test function $Z(\cdot)$, again sorted by their estimated standard errors, are plotted in (c). As one would expect, nearly all of the errors are less than three times their estimated standard errors. Note that the kriging errors are counted without the perturbed noise; they are substantially lower than the noise level.

\subsection{Comparison with MRA}\label{sec:exp.mra}
We use the test function~\eqref{eqn:test.func} in the preceding subsection to further compare the proposed method with MRA~\citep{Katzfuss2017} on kriging and maximum likelihood. Both methods perform a hierarchical decomposition. Our method defines the covariance structure in a bottom-up manner across the partitioning tree and translates it to a recursive low-rank compressed matrix that admits $O(n)$ complexity, while MRA decomposes the random field in a top-down fashion along the tree and yields $O(n\log^pn)$ computational costs for certain $p$'s, suppressing dependency on $r$. The terminology \emph{knots} in MRA plays a similar role as \emph{landmark points} in our method, but the resulting covariance structure is quite different.

\begin{figure}[ht]
\centering
\includegraphics[width=.4\linewidth]{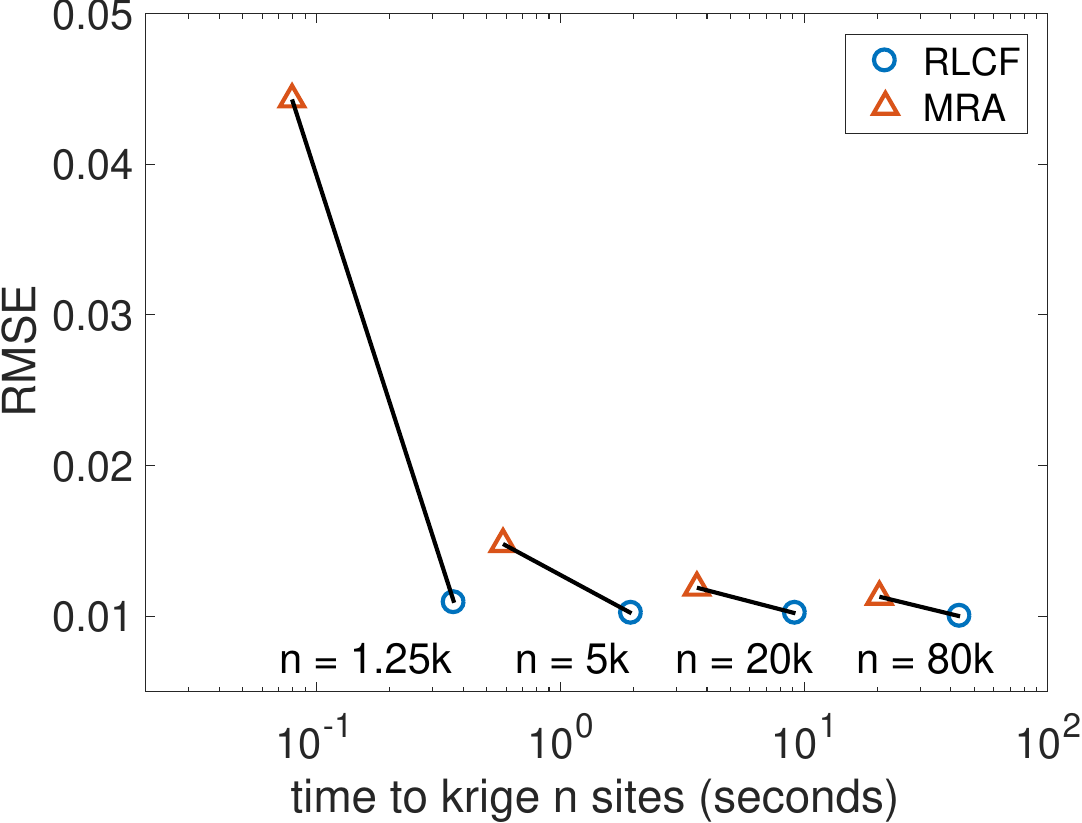} \quad
\includegraphics[width=.4\linewidth]{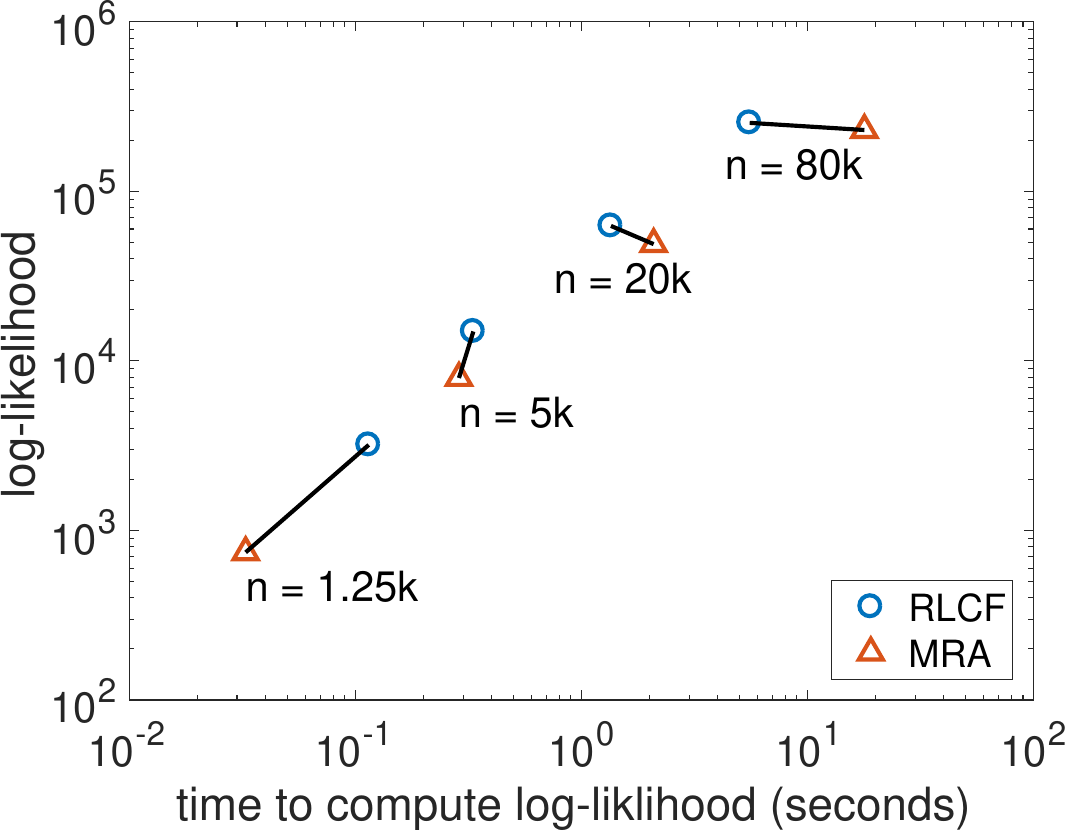} \\
\vskip 10pt
\begin{tabular}{ccccc}
\hline
$n$ & 1.25k & 5k & 20k & 80k \\
\hline
RMSE$_{\rm{MRA}}$ $-$ RMSE$_{\rm{RLCF}}$ &
0.0334 & 0.0046 & 0.0018 & 0.0013\\
(RMSE$_{\rm{MRA}}$ $-0.01$) $/$ (RMSE$_{\rm{RLCF}}$ $-0.01$) &
37.04 & 19.40 & 11.37 & 108.21\\
\hline
log-likelihood$_{\rm{RLCF}}$ $-$ log-likelihood$_{\rm{MRA}}$ &
2455 & 6983 & 13703 & 23478\\
log-likelihood$_{\rm{RLCF}}$ $/$ log-likelihood$_{\rm{MRA}}$ &
4.3046 & 1.8800 & 1.2804 & 1.1020\\
\hline
\end{tabular}
\caption{Comparison with MRA. The quantity $n$ is both the number of observations and the number of kriging sites. For each $n$, the hierarchical partitioning yields the same tree and we use the same $r$. Covariance parameters for each case are individually estimated.}
\label{fig:exp.mra}
\end{figure}

We follow almost the same setting as in the preceding subsection, except to inject $\mathcal{N}(0, 0.01^2)$ noise for a smaller RMSE. The MRA code is the C++ implementation suggested by \url{https://github.com/katzfuss-group/MRA_JASA}, for fair comparison. The computing platform is the same as that in Section~\ref{sec:exp.scaling}. We experiment with a few grid sizes, for each of which we first optimize the log-likelihood on half of the randomly sampled data to estimate covariance parameters, and then perform kriging on the rest of the data. In both methods, we fix $r=125$ and let the tree height be $h=\lfloor\log_2(n/r)\rfloor$.

Figure~\ref{fig:exp.mra} plots RMSE/log-likelihood versus time. A few observations follow. First, with the same hierarchical partitioning and $r$, our method yields lower RMSE (nearly the noise level) and higher log-likelihood. More appealingly, when $n$ grows, the absolute log-likelihood difference increases. Second, both methods obey the $O(n)$ trend (ignoring the logarithmic factor), because the time approximately follows an arithmetic progression under the logarithmic scale. Third, our method calculates log-likelihood faster at large $n$, whereas slower in other cases compared with MRA. For kriging, the consistently slower time is probably caused by a larger constant factor in the big-O complexity. For log-likelihood, one observes that the spacing in elapsed time is different across the two methods. MRA has a bigger spacing, due to an additional $r$ factor in the big-O complexity.

\section{Analysis of Climate Data}\label{sec:exp.noaa}
In this section, we apply the proposed method to analyze a climate data product developed by the National Centers for Environmental Prediction (NCEP) of the National Oceanic and Atmospheric Administration (NOAA).\footnote{\url{https://www.ncdc.noaa.gov/data-access/model-data/model-datasets/climate-forecast-system-version2-cfsv2}} The Climate Forecast System Reanalysis (CFSR) data product~\citep{Saha2010} offers hourly time series as well as monthly means data with a resolution down to one-half of a degree (approximately 56 km) around the Earth, over a period of 32 years from 1979 to 2011. For illustration purpose, we extract the temperature variable at 500 mb height from the monthly means data and show a snapshot on the top of Figure~\ref{fig:noaa}.   Temperatures at this pressure (generally around a height of 5 km) provide a good summary of large-scale weather patterns and should be more nearly stationary than surface temperatures. We will estimate a covariance model for every July over the 32-year period.

\begin{figure}[ht]
\centering
\includegraphics[width=0.8\linewidth]{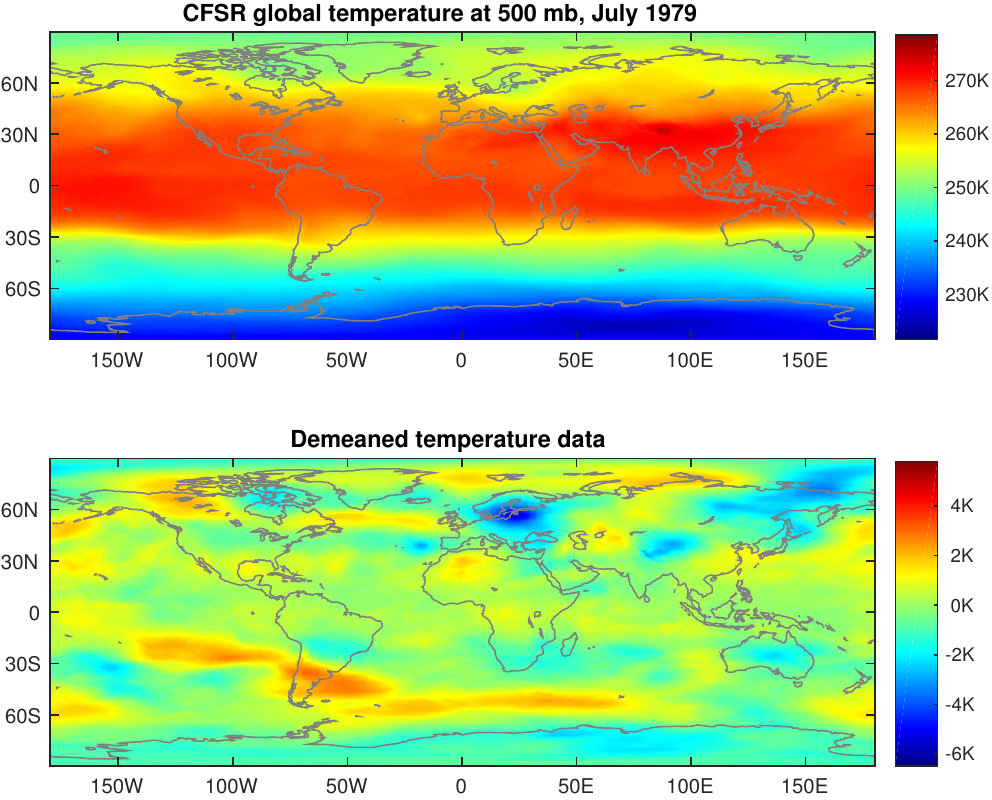}
\caption{Snapshot of CFSR global temperature at 500 mb and the resulting data after subtraction of pixelwise mean for the same month over 32 years.}
\label{fig:noaa}
\end{figure}

Through preliminary investigations, we find that the data appears fairly Gaussian after a subtraction of pixelwise mean across time. An illustration of the demeaned data for the same snapshot is given at the bottom of Figure~\ref{fig:noaa}. Moreover, the correlation between the different snapshots are so weak that we shall treat them as independent anomalies.  Although temperatures have warmed during this period, the warming is modest compared to the interannual variation in temperatures at this spatial resolution, so we assume the anomalies have mean 0. We use $\bm{z}_i$ to denote the anomaly at time $i$. Then, the log-likelihood with $N=32$ zero-mean independent anomalies $\bm{z}_i$ is
\[
\mathcal{L}=-\sum_{i=1}^N\frac{1}{2}\bm{z}_i^TK^{-1}\bm{z}_i-\frac{N}{2}\log\det K-\frac{Nn}{2}\log2\pi.
\]

For random fields on a sphere, a reasonable covariance function for a pair of sites $\bm{x}$ and $\bm{x}'$ may be based on their great-circle distance, or equivalently the chordal distance, because of their monotone relationship. Specifically, let a site $\bm{x}$ be represented by latitude $\phi$ and longitude $\psi$. Then, the chordal distance between two sites $\bm{x}$ and $\bm{x}'$ is
\begin{equation}\label{eqn:chordal}
r=2\left[\sin^2\left(\frac{\phi-\phi'}{2}\right)+\cos\phi\cos\phi'\sin^2\left(\frac{\psi-\psi'}{2}\right)\right]^{1/2}.
\end{equation}
Here, we assume that the radius of the sphere is $1$ for simplicity, because it can always be absorbed into a range parameter later. We still use the Mat\'{e}rn model
\begin{equation}\label{eqn:matern2}
k(\bm{x},\bm{x}')=
\frac{10^{\alpha}}{2^{\nu-1}\Gamma(\nu)}
\left(\frac{\sqrt{2\nu}r}{\ell}\right)^{\nu}
\bessel_{\nu}\left(\frac{\sqrt{2\nu}r}{\ell}\right)
+10^{\tau}\cdot\bm{1}(r=0)
\end{equation}
to define the covariance function, where $r$ is the chordal distance~\eqref{eqn:chordal}, so that the model is isotropic on the sphere. More sophisticated models based on the same Mat\'{e}rn function and the chordal distance $r$ are proposed in~\citep{Jun2008}. Note that this model depends on the longitudes for $\bm{x}$ and $\bm{x}'$ only through their differences modulo $2\pi$. Such a model is called \emph{axially symmetric}~\citep{Jones1963}.

A computational benefit of an axially symmetric model and gridded observations is that one may afford computations with $k$ even when the latitude-longitude grid is dense. The reason is that for any two fixed latitudes, the cross-covariance matrix between the observations is circulant and diagonalizing it requires only one discrete Fourier transform (DFT), which is efficient. Thus, diagonalizing the whole covariance matrix amounts to diagonalizing only the blocks with respect to each longitude, apart from the DFT's for each latitude.

Hence, we will perform computations with both the base covariance function $k$ and the proposed function $k_{\rm{h}}$ and compare the results. We subsample the grid with every other latitude and longitude for parameter estimation. We also remove the two grid lines 90N and 90S due to their degeneracy at the pole. Because of the half-degree resolution, this results in a coarse grid of size $180\times360$ for parameter estimation, for a total of $180\times 360\times 32 = 2{,}073{,}600$ observations. The rest of the grid points are used for kriging. As before, we set the number $r$ of landmark points to be $125$.

\begin{table}[ht]
\centering
\caption{Optimization results for different $\nu$'s using the base covariance function $k$.}
\label{tab:exp.noaa0}
\begin{tabular}{crrrrrrc}
\hline
$\nu$
& \multicolumn{3}{c}{Initial guess $\bm{\theta}_0$}
& \multicolumn{3}{c}{Terminate at $\widehat{\bm{\theta}}$}
& Log-likelihood\\
\hline
$0.5$
& ($-0.285$ & $ 0.156$ & $-4.935$) & ($-0.794$ & $ 1.446$ & $-7.165$) & $ 3.938\times10^6$ \\
& ($-0.794$ & $ 1.446$ & $-7.165$) & \multicolumn{3}{c}{diverge} \\
\hline
$1.0$
& ($-0.285$ & $ 0.156$ & $-4.935$) & ($-0.279$ & $ 0.411$ & $-5.133$) & $ 4.696\times10^6$ \\
& ($-0.279$ & $ 0.411$ & $-5.133$) & ($\phantom{-}0.838$ & $ 1.494$ & $-5.125$) & $ 4.700\times10^6$ \\
\hline
$1.5$
& ($\phantom{-}0.124$ & $ 0.215$ & $-4.933$) & ($-0.285$ & $ 0.156$ & $-4.935$) & $ 4.757\times10^6$ \\
& ($-0.285$ & $ 0.156$ & $-4.935$) & ($-0.285$ & $ 0.156$ & $-4.935$) & $ 4.757\times10^6$ \\
\hline
$2.0$
& ($-0.285$ & $ 0.156$ & $-4.935$) & ($-0.279$ & $ 0.094$ & $-4.933$) & $ 4.643\times10^6$ \\
& ($-0.279$ & $ 0.094$ & $-4.933$) & ($-0.545$ & $ 0.081$ & $-4.821$) & $ 4.653\times10^6$ \\
\hline
\end{tabular}
\end{table}

We set the parameter vector $\bm{\theta}=[\alpha,\ell,\tau]^T$, considering only several values for the smoothness parameter $\nu$ because of the difficulties of numerical optimization of the loglikelihood over $\nu$. To our experience, blackbox optimization solvers do not always find accurate optima. We show in Table~\ref{tab:exp.noaa0} several results of the Matlab solver \texttt{fminunc} when one varies $\nu$. For each $\nu$, we start the solver at some initial guess $\bm{\theta}_0$ until it claims a local optimum $\widehat{\bm{\theta}}$. Then, we use this optimum as the initial guess to run the solver again. Ideally, the solver should terminate at $\widehat{\bm{\theta}}$ if it indeed is an optimum. However, reading Table~\ref{tab:exp.noaa0}, one finds that this is not always the case.

When $\nu=0.5$, the second search diverges from the initial $\widehat{\bm{\theta}}$. The cross-section plots of the log-likelihood (not shown) indicate that $\widehat{\bm{\theta}}$ is far from the center of the contours. The solver terminates merely because the gradient is smaller than a threshold and the Hessian is positive-definite (recall that we \emph{minimize} the negative log-likelihood). The diverging search starting from $\widehat{\bm{\theta}}$ (with $\alpha$ and $\ell$ continuously increasing) implies that the infimum of the negative log-likelihood may occur at infinity, as can sometimes happen in our experience.

When $\nu=1.0$, although the search starting at $\widehat{\bm{\theta}}$ does not diverge, it terminates at a location quite different from $\widehat{\bm{\theta}}$, with the log-likelihood increased by about 4000, which is arguably a small amount given the number of observations. Such a phenomenon is often caused by the fact that the peak of the log-likelihood is flat (at least along some directions); hence, the exact optimizer is hard to locate. This phenomenon similarly occurs in the case $\nu=2.0$. Only when $\nu=1.5$ does restarting the optimization yield  $\widehat{\bm{\theta}}$ that is essentially the same as the initial estimate. Incidentally, the log-likelihood in this case is also the largest. Hence, all subsequent results are produced for only $\nu=1.5$.

\begin{table}[ht]
\centering
\caption{Estimation results ($\nu=1.5$).}
\label{tab:exp.noaa}
\begin{tabular}{cc@{$\,\,$}cc@{$\,\,$}cc@{$\,\,$}c}
\hline
Est. w/
& \multicolumn{2}{c}{$\widehat{\alpha}$}
& \multicolumn{2}{c}{$\widehat{\ell}$}
& \multicolumn{2}{c}{$\widehat{\tau}$}\\
\hline
$k$          
& $-0.2875$ & $(0.0047)$ & $0.15620$ & $(0.00058)$ & $-4.9360$ & $(0.0014)$\\
$k_{\rm{h}}$ 
& $-0.2275$ & $(0.0044)$ & $0.16640$ & $(0.00058)$ & $-4.9300$ & $(0.0015)$\\
\hline
\end{tabular}
\end{table}

\begin{table}[ht]
\centering
\caption{Log-likelihood (left) and root mean squared prediction error (right).}
\label{tab:exp.noaa2}
\begin{minipage}{0.45\linewidth}
\centering
\begin{tabular}{ccc}
\hline
& at $\widehat{\bm{\theta}}$ & at $\widehat{\bm{\theta}}_{\rm{h}}$ \\
\hline
Using $k$          & $4757982$ & $4756981$ \\
Using $k_{\rm{h}}$ & $4557568$ & $4558731$ \\
\hline
\end{tabular}
\end{minipage}%
\begin{minipage}{0.43\linewidth}
\centering
\begin{tabular}{ccc}
\hline
& at $\widehat{\bm{\theta}}$ & at $\widehat{\bm{\theta}}_{\rm{h}}$ \\
\hline
Using $k$          & $0.01394$ & $0.01394$ \\
Using $k_{\rm{h}}$ & $0.01556$ & $0.01556$ \\
\hline
\end{tabular}
\end{minipage}
\end{table}

Near $\widehat{\bm{\theta}}$, we further perform a local grid search and obtain finer estimates, as shown in Table~\ref{tab:exp.noaa}. One sees that the estimated parameters produced by $k$ and $k_{\rm{h}}$ are qualitatively similar, although their differences exceed the tiny standard errors. To distinguish the two estimates, we use $\widehat{\bm{\theta}}$ to denote the one resulting from $k$ and $\widehat{\bm{\theta}}_{\rm{h}}$ from $k_{\rm{h}}$. In Table~\ref{tab:exp.noaa2}, we list the log-likelihood values and the kriging errors when the covariance function is evaluated at both locations. One sees that the estimate $\widehat{\bm{\theta}}_{\rm{h}}$ is quite close to $\widehat{\bm{\theta}}$ in two important regards: first, the root mean squared prediction errors using $k$ are the same to four significant figures, and the log-likelihood under $k$ differs by 1000, which we would argue is a very small difference for more than 2 million observations.  On the other hand, $k_{\rm{h}}$ does not provide a great approximation to the loglikelihood itself and the predictions using $k_{\rm{h}}$ are slightly inferior to those using $k$ no matter which estimate is used. Figure~\ref{fig:exp.noaa} plots the log-likelihoods centered around the respectively optimal estimates. The shapes are visually identical, which supports the use of $k_{\rm{h}}$ for parameter estimation. Since kriging with $k$ is often much easier than maximizing the log-likelihood, in this data example one could use $k_{\rm{h}}$ to estimate $\bm{\theta}$ and then $k$ to krige.

\begin{figure}[ht]
\centering
\subfigure[$\ell$-$\tau$ plane]{
  \includegraphics[width=.31\linewidth]{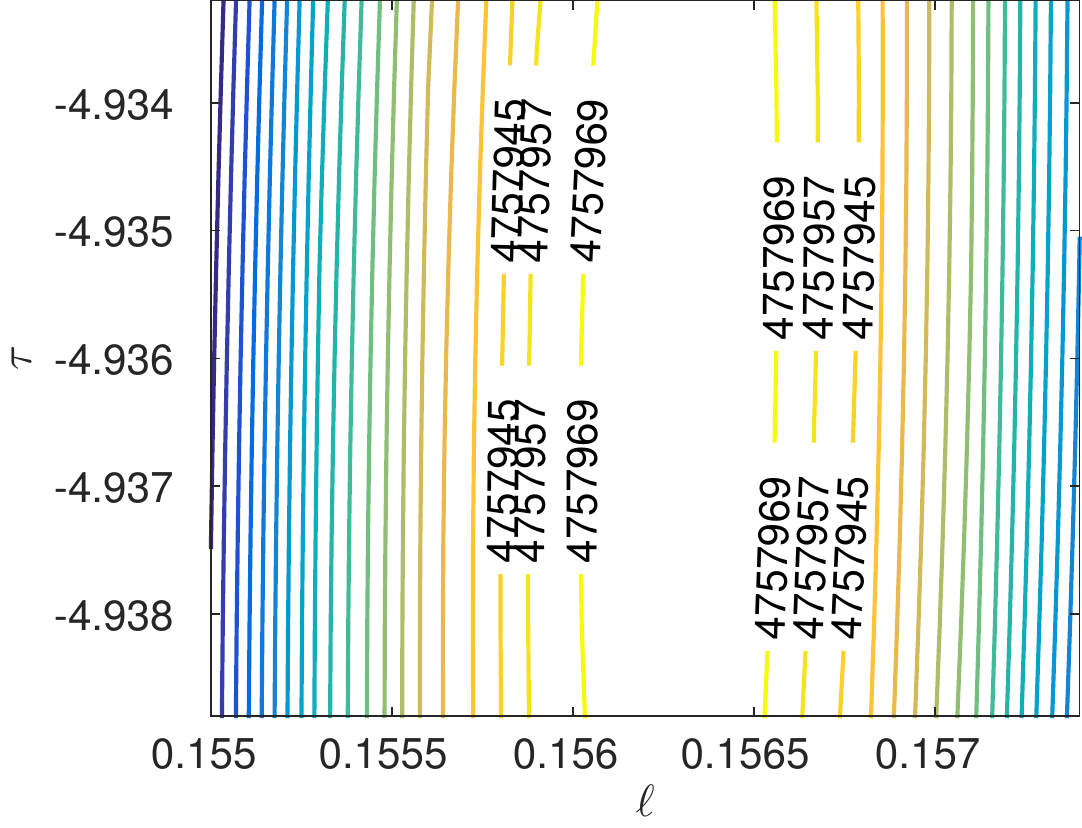}}
\subfigure[$\alpha$-$\tau$ plane]{
  \includegraphics[width=.31\linewidth]{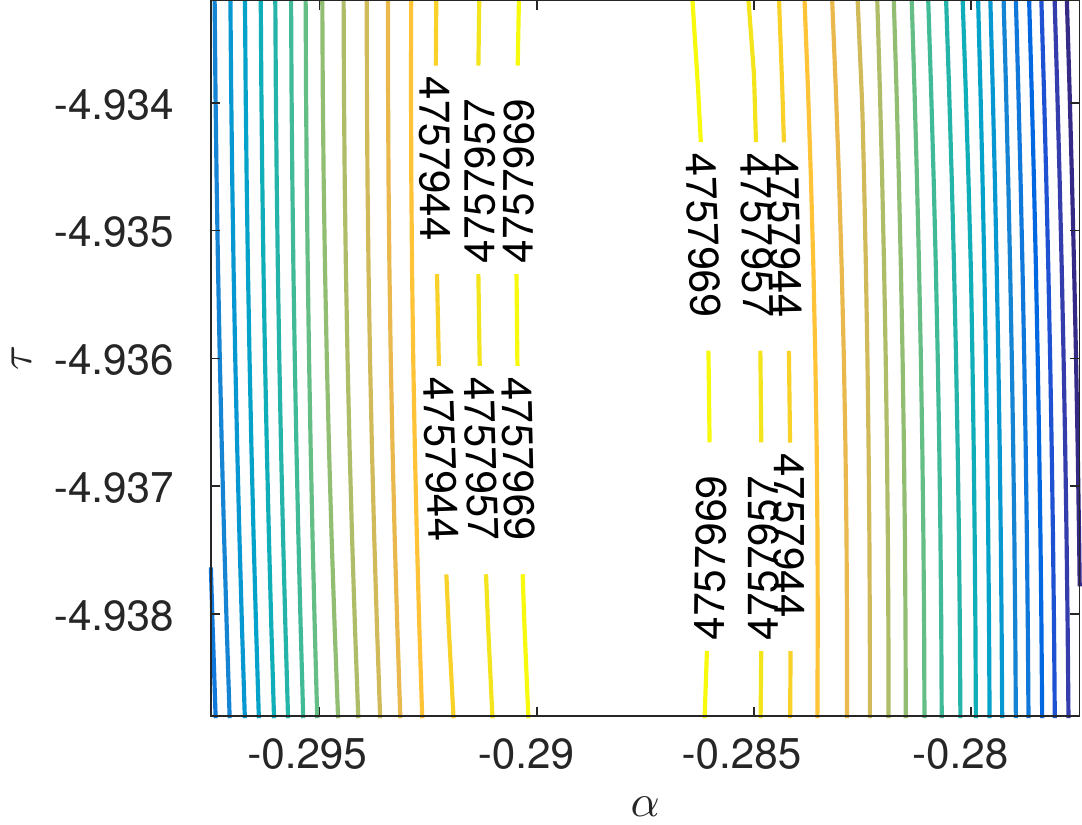}}
\subfigure[$\alpha$-$\ell$ plane]{
  \includegraphics[width=.31\linewidth]{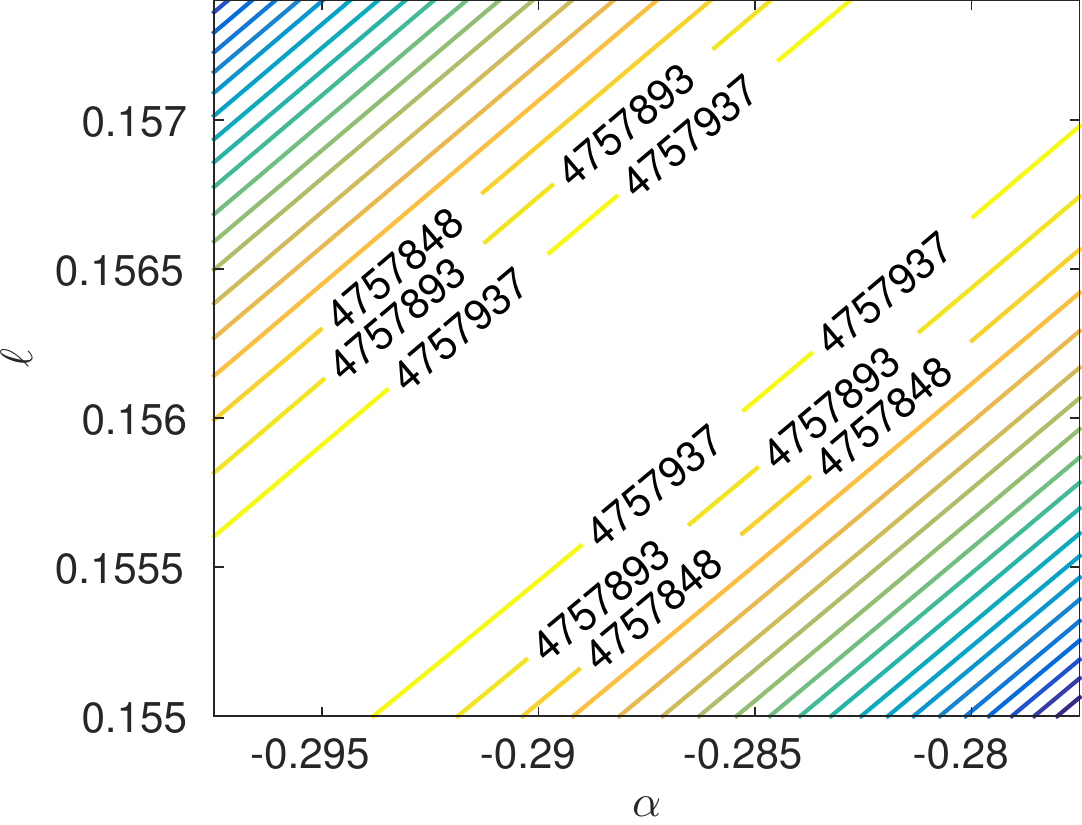}}\\

\subfigure[$\ell$-$\tau$ plane]{
  \includegraphics[width=.31\linewidth]{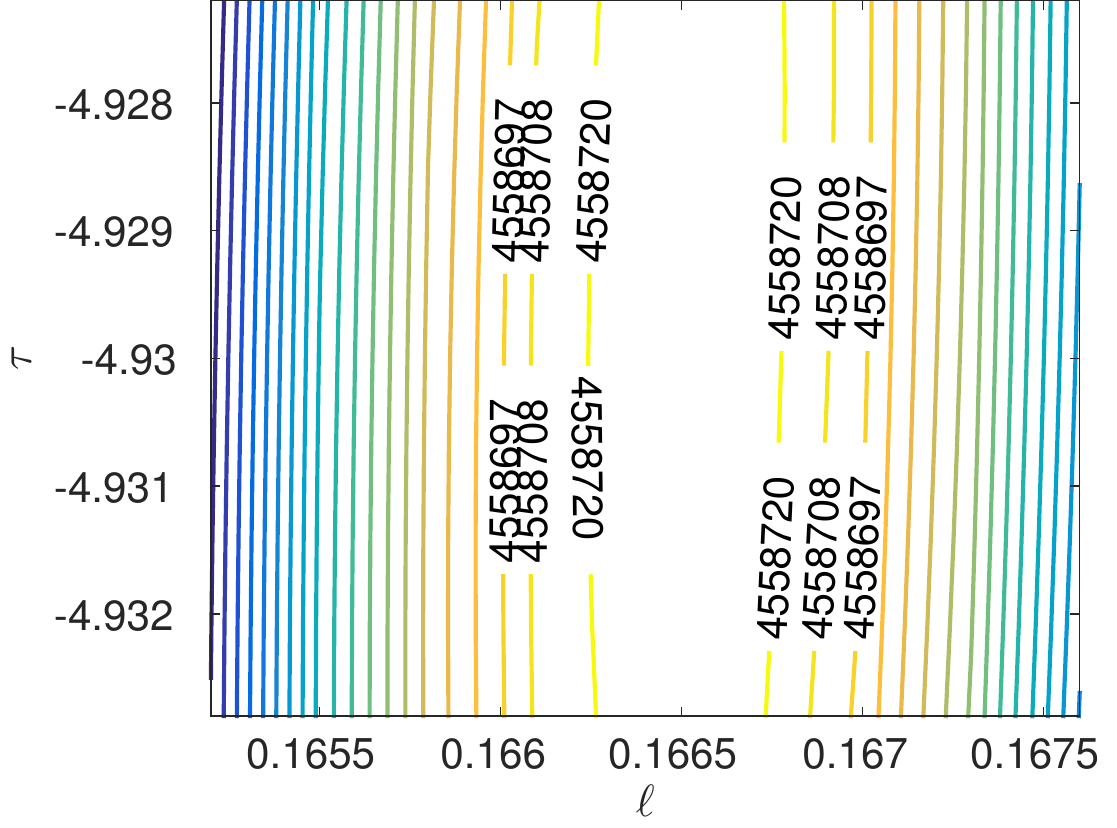}}
\subfigure[$\alpha$-$\tau$ plane]{
  \includegraphics[width=.31\linewidth]{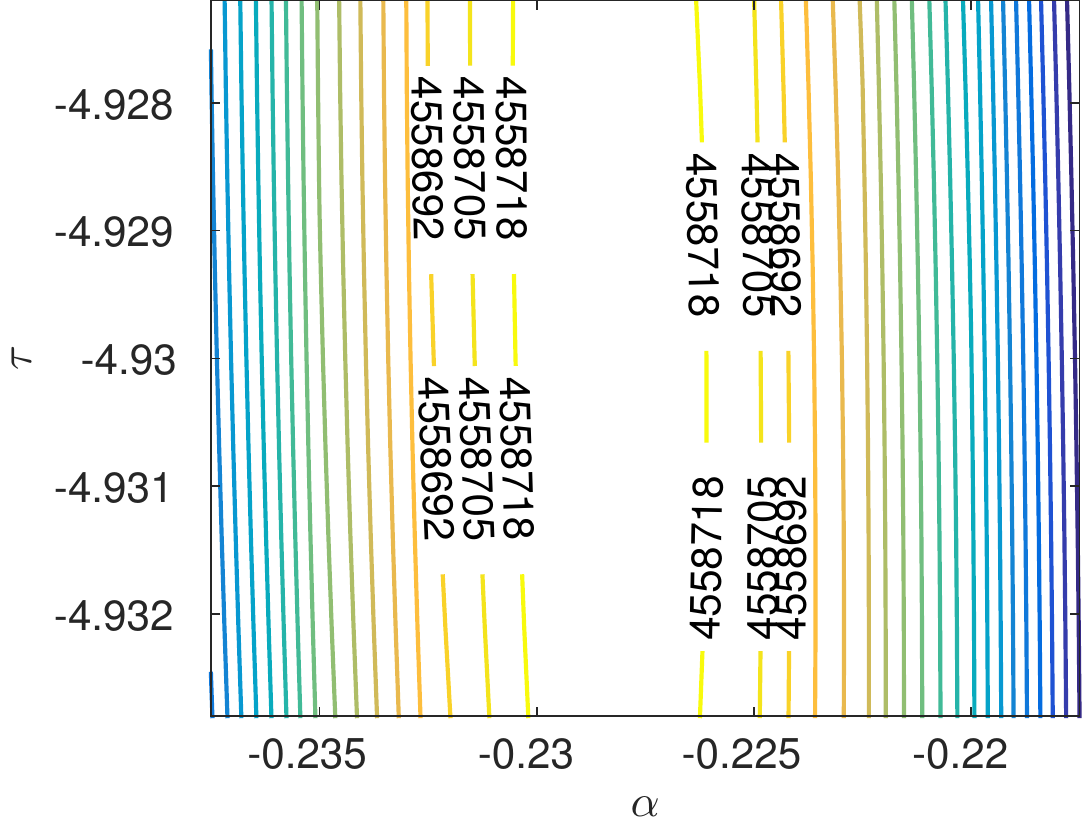}}
\subfigure[$\alpha$-$\ell$ plane]{
  \includegraphics[width=.31\linewidth]{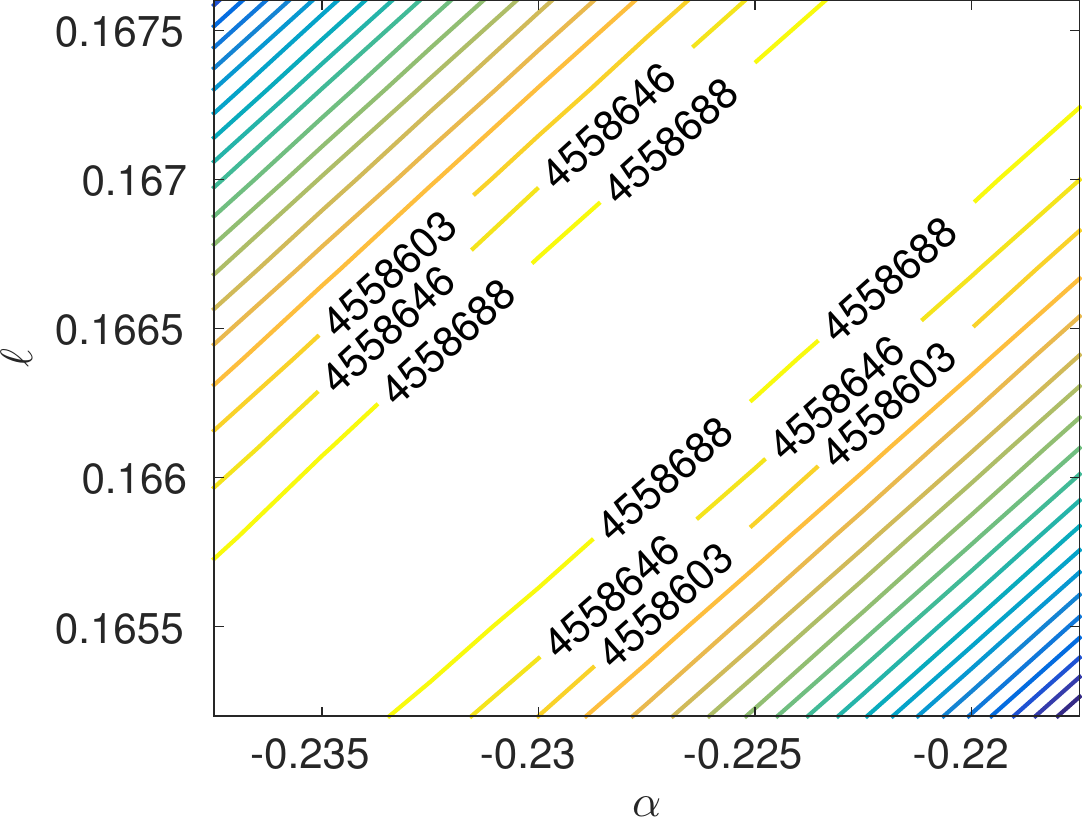}}
\caption{Log-likelihood centered around optimum. Top row: base covariance function $k$; bottom row: proposed covariance function $k_{\rm{h}}$.}
\label{fig:exp.noaa}
\end{figure}

\section{Conclusions}
We have presented a computationally friendly approach that addresses the challenge of formidably expensive computations of Gaussian random fields in the large scale. Unlike many methods that focus on the approximation of the covariance matrix or of the likelihood, the proposed approach operates on the covariance function such that positive definiteness is maintained. The hierarchical structure and the nested bases in the proposed construction allow for organizing various computations in a tree format, achieving costs proportional to the tree size and hence to the data size $n$. These computations range from the simulation of random fields to kriging and likelihood evaluations. More desirably, kriging has an amortized cost of $O(\log n)$ and hence one may perform predictions for as many as $O(n)$ sites easily. Moreover, the efficient evaluation of the log-likelihoods paves the way for maximum likelihood estimation as well as Markov Chain Monte Carlo. Numerical experiments show that the proposed construction yields comparable prediction results and likelihood surfaces with those of the base covariance function, while being scalable to data of ever increasing size.

\bibliographystyle{plainnat}
\bibliography{reference}

\clearpage
\appendix
\section{Proof of Theorem~\ref{thm:k.one.level}}
For a proof of positive definiteness, we write $k_{\rm{h}}$ as a sum of two functions $\xi^{(1)}$ and $\xi^{(2)}$, where
\[
\xi^{(1)}(\bm{x},\bm{x}')=k(\bm{x},\ud{X})k(\ud{X},\ud{X})^{-1}k(\ud{X},\bm{x}')
\]
is the Nystr\"{o}m approximation in the whole domain $S$ and hence positive definite, and
\[
\xi^{(2)}(\bm{x},\bm{x}')=
\begin{cases}
k(\bm{x},\bm{x}')-k(\bm{x},\ud{X})k(\ud{X},\ud{X})^{-1}k(\ud{X},\bm{x}'), & \text{if $\bm{x},\bm{x}'\in S_j$ for some $j$},\\
0, & \text{otherwise},
\end{cases}
\]
is a Schur complement in each subdomain $S_j$ and hence also positive definite. Then, the constructed $k_{\rm{h}}$ is positive definite.

To prove the strict positive definiteness, we need the following lemma.

\begin{lemma}\label{lem:strict}
Let $k$ be strictly positive definite. For any set of points $X=\{\bm{x}_1,\ldots,\bm{x}_n\}$ such that $X\cap\ud{X}=\emptyset$ and for any set of coefficients $\alpha_1,\ldots,\alpha_n\in\real$ that are not all zero, we have
\[
\sum_{i,j=1}^n\alpha_i\alpha_j\left[k(\bm{x}_i,\bm{x}_j)-k(\bm{x}_i,\ud{X})k(\ud{X},\ud{X})^{-1}k(\ud{X},\bm{x}_j)\right]>0.
\]
\end{lemma}

\begin{proof}
The result is equivalent to saying that the matrix $k(X,X)-k(X,\ud{X})k(\ud{X},\ud{X})^{-1}k(\ud{X},X)$ is positive definite. To see so, consider
\[
k(X\cup\ud{X},X\cup\ud{X})=
\begin{bmatrix}
k(X,X) & k(X,\ud{X}) \\
k(\ud{X},X) & k(\ud{X},\ud{X})
\end{bmatrix}.
\]
Because of the strict positive definiteness of the function $k$, the matrix $k(X\cup\ud{X},X\cup\ud{X})$ is positive definite. Then, the Schur complement matrix $k(X,X)-k(X,\ud{X})k(\ud{X},\ud{X})^{-1}k(\ud{X},X)$ is also positive definite.
\end{proof}

We now continue the proof of Theorem~\ref{thm:k.one.level}. For a set of coefficients $\alpha_1,\ldots,\alpha_n\in\real$,
\begin{equation}\label{eqn:sum}
\sum_{i,j=1}^n\alpha_i\alpha_jk_{\rm{h}}(\bm{x}_i,\bm{x}_j)
=\underbrace{\sum_{i,j=1}^n\alpha_i\alpha_j \xi^{(1)}(\bm{x}_i,\bm{x}_j)}_{B_1}+
\underbrace{\sum_{i,j=1}^n\alpha_i\alpha_j \xi^{(2)}(\bm{x}_i,\bm{x}_j)}_{B_2}.
\end{equation}
If we want the left-hand side to be zero, $B_1$ and $B_2$ must be simultaneously zero. Because $\xi^{(2)}(\bm{x},\bm{x}')$ is zero whenever $\bm{x}$ or $\bm{x}'$ belongs to $\ud{X}$, based on Lemma~\ref{lem:strict}, $B_2=0$ implies that $\alpha_i=0$ for all $\bm{x}_i\notin\ud{X}$. In such a case, $B_1$ is simplified to
\[
B_1=\sum_{\bm{x}_i,\bm{x}_j\in\ud{X}}\alpha_i\alpha_j \xi^{(1)}(\bm{x}_i,\bm{x}_j)
=\ud{\bm{\alpha}}^Tk(\ud{X},\ud{X})\ud{\bm{\alpha}},
\]
where $\ud{\bm{\alpha}}$ is the column vector of $\alpha_i$'s for all $\bm{x}_i\in\ud{X}$. Then, because of the strict positive definiteness of $k$, $B_1=0$ implies that $\alpha_i=0$ for all $\bm{x}_i\in\ud{X}$. Thus, all coefficients $\alpha_i$ must be zero for the left-hand side of~\eqref{eqn:sum} to be zero. This concludes that $k_{\rm{h}}$ is strictly positive definite.

\clearpage
\section{Proof of Theorem~\ref{thm:k.multilevel}}
The positive definiteness of $k_{\rm{h}}$ straightforwardly follows from the discussion in the main text: $k_{\rm{h}}$ is the sum of all $\xi^{(i)}$'s, each of which is positive definite.

To prove strict positive definiteness, we first simplify notations. We write for the covariance function $k$:
\[
k_{\bm{x},\bm{x}'}\equiv k(\bm{x},\bm{x}'),\quad
k_{\bm{x},\ud{i}}\equiv k(\bm{x},\ud{X}_i),\quad
k_{\ud{i},\ud{j}}\equiv k(\ud{X}_i,\ud{X}_j),
\]
and similarly for the auxiliary function $\psi^{(i)}$. Then, \eqref{eqn:xi} is simplified to
\begin{equation*}
\begin{split}
\xi^{(i)}(\bm{x},\bm{x}') &= 0 \text{ if either $\bm{x}$ or $\bm{x}'\notin S_i$; otherwise: } \\
\xi^{(i)}(\bm{x},\bm{x}') &=
\begin{cases}
k_{\bm{x},\bm{x}'}-k_{\bm{x},\ud{p}}k_{\ud{p},\ud{p}}^{-1}k_{\ud{p},\bm{x}'}, & \text{if $i$ is leaf},\\
\psi^{(i)}_{\bm{x},\ud{i}}k_{\ud{i},\ud{i}}^{-1}
\Big(k_{\ud{i},\ud{i}}-k_{\ud{i},\ud{p}}k_{\ud{p},\ud{p}}^{-1}k_{\ud{p},\ud{i}}\Big)
k_{\ud{i},\ud{i}}^{-1}\psi^{(i)}_{\ud{i},\bm{x}'}, & \text{if $i$ is neither leaf nor root},\\
\psi^{(i)}_{\bm{x},\ud{i}}k_{\ud{i},\ud{i}}^{-1}\psi^{(i)}_{\ud{i},\bm{x}'}, & \text{if $i$ is root}.
\end{cases}
\end{split}
\end{equation*}

We need the following lemma.

\begin{lemma}\label{lem:psik}
Let $l$ be a leaf descendant of some nonleaf node $i$ and let $(l,l_1,l_2,\ldots,l_s,i)$ be the path connecting $l$ and $i$. Then,
\[
\psi^{(i)}_{\bm{x},\ud{i}}=k_{\bm{x},\ud{i}}
\]
if $\bm{x}\in \ud{X}_{l_1}\cap\ud{X}_{l_2}\cap\cdots\cap\ud{X}_{l_s}$.
\end{lemma}

\begin{proof}
The result is a straightforward verification. For an array of distinct points which contains some point $\bm{x}$ at the $j$-th location, we use the notation $\bm{e}_{\bm{x}}$ to denote a column vector whose $j$-th element is $1$ and otherwise $0$. Then, for $\bm{x}\in S_l$ and also $\in \ud{X}_{l_1}$,
\[
\psi^{(i)}_{\bm{x},\ud{i}}
=k_{\bm{x},\ud{l_1}}k_{\ud{l_1},\ud{l_1}}^{-1}k_{\ud{l_1},\ud{l_2}}k_{\ud{l_2},\ud{l_2}}^{-1}\cdots k_{\ud{l_s},\ud{l_s}}^{-1}k_{\ud{l_s},\ud{i}}
=\bm{e}_{\bm{x}}^Tk_{\ud{l_1},\ud{l_2}}k_{\ud{l_2},\ud{l_2}}^{-1}\cdots k_{\ud{l_s},\ud{l_s}}^{-1}k_{\ud{l_s},\ud{i}}
=k_{\bm{x},\ud{l_2}}k_{\ud{l_2},\ud{l_2}}^{-1}\cdots k_{\ud{l_s},\ud{l_s}}^{-1}k_{\ud{l_s},\ud{i}}.
\]
Iteratively simplifying by noting that $\bm{x}$ also belongs to $\ud{X}_{l_1},\ldots,\ud{X}_{l_s}$, we eventually reach
\[
\psi^{(i)}_{\bm{x},\ud{i}}
=k_{\bm{x},\ud{l_s}}k_{\ud{l_s},\ud{l_s}}^{-1}k_{\ud{l_s},\ud{i}}
=\bm{e}_{\bm{x}}^Tk_{\ud{l_s},\ud{i}}
=k_{\bm{x},\ud{i}}.
\]
\end{proof}

We now continue the proof of Theorem~\ref{thm:k.multilevel}. The strategy resembles induction. For a set of coefficients $\alpha_1,\ldots,\alpha_n\in\real$, write
\begin{equation}\label{eqn:sum2}
\sum_{j,l=1}^n\alpha_j\alpha_lk_{\rm{h}}(\bm{x}_j,\bm{x}_l)
=\sum_i\underbrace{\sum_{j,l=1}^n\alpha_j\alpha_l\xi^{(i)}(\bm{x}_j,\bm{x}_l)}_{B_i}.
\end{equation}
If we want the left-hand side to be zero, all the $B_i$'s on the right must be simultaneously zero. When $i$ is a leaf node, $\xi^{(i)}(\bm{x},\bm{x}')$ is zero whenever $\bm{x}$ or $\bm{x}'$ belongs to $\ud{X}_p$, where $p$ is the parent of $i$. Then, $B_i=0$ implies that $\alpha_j=0$ for all $\bm{x}_j\in S_i\backslash\ud{X}_p$.

For any nonleaf node $p$, we use $Q_p$ to denote the union of the intersections of landmark points:
\[
Q_p\equiv\bigcup_{l \text{ is leaf descendant of } p}
\{ \ud{X}_{l_1}\cap\cdots\cap\ud{X}_{l_s}\cap\ud{X}_p \mid
\text{$(l,l_1,\ldots,l_s,p)$ is a path connecting $l$ and $p$} \}.
\]
Clearly, $Q_p\subset S_p$. As a special case, if all the children of $p$ are leaf nodes, $Q_p=\ud{X}_p$. We now have an induction hypothesis: for a nonroot node $i$ with parent $p$, there holds $\alpha_j=0$ for all $\bm{x}_j\in S_i\backslash (Q_i\cap\ud{X}_p)$. Assume that the hypothesis is true for all child nodes of some node $p$, who has a parent $q$. Then, summarizing the results for all these child nodes, we have $\alpha_j=0$ for all $\bm{x}_j\in S_p\backslash Q_p$. Furthermore, based on Lemma~\ref{lem:psik}, $\xi^{(p)}(\bm{x},\bm{x}')$ is zero whenever $\bm{x}$ or $\bm{x}'$ belongs to $Q_p\cap\ud{X}_q$. Then, $B_p=0$ implies that $\alpha_j=0$ for all $\bm{x}_j\in (S_p\backslash Q_p)\cup(Q_p\backslash\ud{X}_q)=S_p\backslash (Q_p\cap\ud{X}_q)$. This finishes the induction step.

At the end of the induction, we reach the root node $p$. Summarizing the results for all the child nodes of the root, we have $\alpha_j=0$ for all $\bm{x}_j\in S_p\backslash Q_p$. Invoking Lemma~\ref{lem:psik} again, we have $\xi^{(p)}(\bm{x},\bm{x}')=k_{\bm{x},\bm{x}'}$ whenever $\bm{x}$ or $\bm{x}'$ belongs to $Q_p$. Then, by the strict positive definiteness of $k$, $B_p=0$ implies that $\alpha_j=0$ for all $\bm{x}_j\in Q_p$. Thus, all coefficients $\alpha_i$ must be zero for the left-hand side of~\eqref{eqn:sum2} to be zero. This concludes that $k_{\rm{h}}$ is strictly positive definite.

\clearpage
\section{Algorithm for Matrix-vector Multiplication}\label{sec:matvec}
The objective is to compute $\bm{y}=A\bm{b}$. We will use a shorthand notation $\bm{b}_i$ to denote a subvector of $\bm{b}$ that corresponds to the index set $I_i$; and similarly for the vector $\bm{y}$. In computer implementation, only the subvectors corresponding to leaf nodes are stored therein. On the other hand, we need auxiliary vectors $\bm{c}_j$ and $\bm{d}_j$, all of length $r$, to be stored in each nonroot node $j$. These auxiliary vectors are defined in the following context.

The vector $\bm{y}$ is the sum of two parts: the first part comes from $A_{ll}\bm{b}_l$ for every leaf node $l$ and the second part comes from $A_{ij}\bm{b}_j$ for every pair of sibling nodes $i$ and $j$. The first part is straightforward to calculate. The second part, however, needs an expansion through change of basis according to Definition~\ref{def:matrix}. In particular, let $l$ be a leaf descendant of $i$. Then, the subvector of $A_{ij}\bm{b}_j$ corresponding to the index set $I_l$ is
\[
U_lW_{l_1}W_{l_2}\cdots W_{l_s}W_i\Sigma_pZ_j^T\left(
\sum_{\substack{q \text{ is leaf} \\ (q,q_1,q_2,\ldots,q_t,j) \text{ is path}}}
Z_{q_t}^T\cdots Z_{q_2}^TZ_{q_1}^TV_q^T\bm{b}_q\right),
\]
where $p$ is the parent of $i$ and $j$, $(l,l_1,l_2,\ldots,l_s,i)$ is the path connecting $l$ and $i$, and the bracketed expression to the right of $Z_j^T$ sums over all the contributions from any descendant leaf $q$ of $j$.

Many computations in the above summation are duplicated. For example, the term $V_q^T\bm{b}_q$ at a leaf node $q$ appears in all $A_{ij}\bm{b}_j$ whenever $q$ is a leaf descendant of $j$. Hence, we define two sets of auxiliary vectors
\[
\bm{c}_i=\begin{dcases}
V_i^T\bm{b}_i, & \text{if $i$ is leaf},\\
Z_i^T\sum_{j\in \text{Ch}(i)}\bm{c}_j, & \text{otherwise},\\
\end{dcases}
\]
and
\[
\bm{d}_j=W_i\bm{d}_i+\sum_{j'\in \text{Ch}(i)\backslash\{j\}}\Sigma_i\bm{c}_{j'},
\quad\text{for $j$ being a child of $i$; \quad $W_i\bm{d}_i=0$ if $i$ is root},
\]
as temporary storage to avoid duplicate computation. It is not hard to see that for any leaf node $l$, the final output subvector is $\bm{y}_l=A_{ll}\bm{b}_l+U_l\bm{d}_l$.

By definition, the set of auxiliary vectors $\bm{c}_i$ may be recursively computed from children to parent, whereas the other set $\{\bm{d}_j\}$ may be computed in a reverse order, from parent to children. Then, the overall computation consists of two tree walks, one upward and the other downward. This computation is summarized in Algorithm~\ref{algo:Ab}. The blue texts highlight the modification of the algorithm when $A$ is symmetric. All subsequent algorithms similarly use blue texts to indicate modifications for symmetry.

\begin{algorithm}[ht]
\caption{Computing $\bm{y}=A\bm{b}$}
\label{algo:Ab}
\begin{algorithmic}[1]
\State Initialize $\bm{d}_i\gets\bm{0}$ for each nonroot node $i$ of the tree
\State \Call{Upward}{\texttt{root}}
\State \Call{Downward}{\texttt{root}}
\Statex

\Function{Upward}{$i$}
\If{$i$ is leaf}
\State $\bm{c}_i\gets V_i^T\bm{b}_i$; \,\, $\bm{y}_i\gets A_{ii}\bm{b}_i$
\Comment{\symln{if $A$ is symmetric, replace $V_i$ by $U_i$}}
\Else
\State \textbf{for all} children $j$ of $i$ \textbf{do} \Call{Upward}{$j$} \textbf{end for}
\State $\bm{c}_i\gets Z_i^T\left(\sum_{j\in \text{Ch}(i)}\bm{c}_j\right)$ \textbf{if} $i$ is not root
\Comment{\symln{if $A$ is symmetric, replace $Z_i$ by $W_i$}}
\EndIf
\If{$i$ is not root}
\State \textbf{for all} siblings $l$ of $i$ \textbf{do} $\bm{d}_l\gets \bm{d}_l+\Sigma_p \bm{c}_i$ \textbf{end for}
\Comment{$p$ is parent of $i$}
\EndIf
\EndFunction
\Statex

\Function{Downward}{$i$}
\State \textbf{if} $i$ is leaf \textbf{then} $\bm{y}_i\gets \bm{y}_i+U_i\bm{d}_i$ and return \textbf{end if}
\ForAll{children $j$ of $i$}
\State $\bm{d}_j\gets \bm{d}_j+W_i\bm{d}_i$, \textbf{if} $i$ is not root
\State \Call{Downward}{$j$}
\EndFor
\EndFunction
\end{algorithmic}
\end{algorithm}

\clearpage
\section{Algorithm for Matrix Inversion}
The objective is to compute $A^{-1}$. We first note that $A^{-1}$ has exactly the same structure as that of $A$. We repeat this observation mentioned in the main paper:

\begin{theorem}
Let $A$ be recursively low-rank with a partitioning tree $T$ and a positive integer $r$. If $A$ is invertible and additionally, $A_{ii}-U_i\Sigma_pV_i^T$ is also invertible for all pairs of nonroot node $i$ and parent $p$, then there exists a recursively low-rank matrix $\widetilde{A}$ with the same partitioning tree $T$ and integer $r$, such that $\widetilde{A}=A^{-1}$. We denote the corresponding factors of $\widetilde{A}$ to be
\begin{equation}\label{eqn:factors.invA}
\{\widetilde{A}_{ii}, \widetilde{U}_i, \widetilde{V}_i, \widetilde{\Sigma}_p, \widetilde{W}_q, \widetilde{Z}_q \mid i \text{ is leaf, } p \text{ is nonleaf, } q \text{ is neither leaf nor root}\}.
\end{equation}
\end{theorem}

This theorem may be proved by construction, which simultaneously gives all the factors in~\eqref{eqn:factors.invA}. Consider a pair of child node $p$ and parent $q$ and let $p$ have children such as $i$ and $j$. By noting that a diagonal block of $A_{pp}$ is $A_{ii}$ and an off-diagonal block is $A_{ij}=U_i\Sigma_pV_j^T$, we may write $A_{pp}-U_p\Sigma_qV_p^T$ as a block diagonal matrix (with diagonal blocks equal to $A_{ii}-U_i\Sigma_pV_i^T$) plus a rank-$r$ term:
\begin{equation}\label{eqn:AA}
A_{pp}-U_p\Sigma_qV_p^T=
\diag\Big[A_{ii}-U_i\Sigma_pV_i^T\Big]_{i\in\text{Ch}(p)}
+\begin{bmatrix}\vdots \\ U_i \\ \vdots\end{bmatrix}
(\Sigma_p-W_p\Sigma_qZ_p^T)
\begin{bmatrix}\cdots & V_i^T & \cdots\end{bmatrix}.
\end{equation}
In fact, this equation also applies to $p=$ root, in which case one treats $\Sigma_q,W_p,Z_p=0$. Then, the Sherman--Morrison--Woodbury formula gives the inverse
\begin{equation}\label{eqn:iA}
(A_{pp}-U_p\Sigma_qV_p^T)^{-1}=
\diag\Big[(A_{ii}-U_i\Sigma_pV_i^T)^{-1}\Big]_{i\in\text{Ch}(p)}
+\begin{bmatrix}\vdots \\ \widetilde{U}_i \\ \vdots\end{bmatrix}
\widetilde{\Pi}_p
\begin{bmatrix}\cdots & \widetilde{V}_i^T & \cdots\end{bmatrix},
\end{equation}
where the tilded factors are related to the non-tilded factors through
\begin{equation}\label{eqn:UV}
\widetilde{U}_i=(A_{ii}-U_i\Sigma_pV_i^T)^{-1}U_i, \qquad
\widetilde{V}_i=(A_{ii}-U_i\Sigma_pV_i^T)^{-T}V_i,
\end{equation}
and
\begin{equation}\label{eqn:Pi}
\widetilde{\Pi}_p=
-(I+\widetilde{\Lambda}_p\widetilde{\Xi}_p)^{-1}\widetilde{\Lambda}_p
\quad\text{with}\quad
\widetilde{\Lambda}_p=\Sigma_p-W_p\Sigma_qZ_p^T
\quad\text{and}\quad
\widetilde{\Xi}_p=\sum_{i\in \text{Ch}(p)}V_i^T\widetilde{U}_i.
\end{equation}
Equation~\eqref{eqn:UV} immediately gives the $\widetilde{U}_i$ and $\widetilde{V}_i$ factors of $\widetilde{A}$ for all leaf nodes $i$. Further, right-multiplying $U_p$ to both sides of~\eqref{eqn:iA} and similarly left-multiplying $V_p^T$ to both sides, we obtain 
\[
\widetilde{W}_p=(I+\widetilde{\Pi}_p\widetilde{\Xi}_p)W_p \quad\text{and}\quad
\widetilde{Z}_p=(I+\widetilde{\Pi}_p^T\widetilde{\Xi}_p^T)Z_p,
\]
which give the $\widetilde{W}_p$ and $\widetilde{Z}_p$ factors of $\widetilde{A}$ for all nonleaf and nonroot nodes $p$.

Additionally, \eqref{eqn:iA} may be interpreted as relating the inverse of $A_{pp}-U_p\Sigma_rV_p^T$ at some parent level $p$, to that of $A_{ii}-U_i\Sigma_pV_i^T$ at the child level $i$ with a rank-$r$ correction. Then, let $i$ be a leaf node and $(i,i_1,i_2,\ldots,i_s,1)$ be the path connecting $i$ and the root $=1$. We expand the chain of corrections and obtain
\begin{equation}\label{eqn:Aii}
\widetilde{A}(I_i,I_i)=(A_{ii}-U_i\Sigma_{i_1}V_i^T)^{-1}
+\widetilde{U}_i\widetilde{\Pi}_{i_1}\widetilde{V}_i^T
+\widetilde{U}_i\widetilde{W}_{i_1}\widetilde{\Pi}_{i_2}\widetilde{Z}_{i_1}^T\widetilde{V}_i^T
+\cdots
+(\widetilde{U}_i\widetilde{W}_{i_1}\cdots\widetilde{W}_{i_s}\widetilde{\Pi}_{1}\widetilde{Z}_{i_s}^T\cdots\widetilde{Z}_{i_1}^T\widetilde{V}_i^T).
\end{equation}
Meanwhile, for any nonleaf node $p$, the factor $\widetilde{\Sigma}_p$ admits a similar chain of corrections:
\begin{equation}\label{eqn:Sigmap}
\widetilde{\Sigma}_p=\widetilde{\Pi}_p
+\widetilde{W}_p\widetilde{\Pi}_{p_1}\widetilde{Z}_p^T
+\widetilde{W}_p\widetilde{W}_{p_1}\widetilde{\Pi}_{p_2}\widetilde{Z}_{p_1}^T\widetilde{Z}_p^T
+\cdots
+(\widetilde{W}_p\widetilde{W}_{p_1}\cdots\widetilde{W}_{p_t}\widetilde{\Pi}_{1}\widetilde{Z}_{p_t}^T\cdots\widetilde{Z}_{p_1}^T\widetilde{Z}_p^T),
\end{equation}
where $(p,p_1,p_2,\ldots,p_t,1)$ is the path connecting $p$ and the root $=1$. The above two formulas give the $\widetilde{A}_{ii}$ and $\widetilde{\Sigma}_p$ factors of $\widetilde{A}$ for all leaf nodes $i$ and nonleaf nodes $p$.

Hence, the computation of $\widetilde{A}$ consists of two tree walks, one upward and the other downward. In the upward phase, $\widetilde{U}_i$, $\widetilde{V}_i$, $\widetilde{W}_p$, and $\widetilde{Z}_p$ are computed. This phase also computes $(A_{ii}-U_i\Sigma_{i_1}V_i^T)^{-1}$ and $\widetilde{\Pi}_p$ as the starting point of corrections. Then, in the downward phase, a chain of corrections as detailed by~\eqref{eqn:Aii} and~\eqref{eqn:Sigmap} is performed from parent to children, which eventually yields the correct $\widetilde{A}_{ii}$ and $\widetilde{\Sigma}_p$. The overall computation is summarized in Algorithm~\ref{algo:invA}. The algorithm also includes straightforward modifications for the case of symmetric $A$.

{%
\renewcommand{\baselinestretch}{1.1}
\begin{algorithm}[!ht]
\caption{Computing $\widetilde{A}=A^{-1}$}
\label{algo:invA}
\begin{algorithmic}[1]
\State \Call{Upward}{\texttt{root}}
\State \Call{Downward}{\texttt{root}}
\Statex

\Function{Upward}{$i$}
\If{$i$ is leaf}
\State\label{algo.ln:patch1} $\widetilde{A}_{ii}\gets (A_{ii}-U_i\Sigma_pV_i^T)^{-1}$
\Comment{$p$ is parent of $i$}
\Statex \Comment{\symln{if $A$ is symmetric, replace $V_i$ by $U_i$}}
\State $\widetilde{U}_i\gets \widetilde{A}_{ii}U_i$
\State $\widetilde{V}_i\gets \widetilde{A}_{ii}^TV_i$
\Comment{\symln{if $A$ is symmetric, no need for this step}}
\State $\widetilde{\Theta}_i\gets V_i^T\widetilde{U}_i$
\Comment{\symln{if $A$ is symmetric, replace $V_i$ by $U_i$}}
\State return
\EndIf
\ForAll{children $j$ of $i$}
\State \Call{Upward}{$j$}
\State $\widetilde{W}_j\gets (I+\widetilde{\Sigma}_j\widetilde{\Xi}_j)W_j$ \textbf{if} $j$ is not leaf
\State $\widetilde{Z}_j\gets (I+\widetilde{\Sigma}_j^T\widetilde{\Xi}_j^T)Z_j$ \textbf{if} $j$ is not leaf
\Comment{\symln{if $A$ is symmetric, no need for this step}}
\State $\widetilde{\Theta}_j\gets Z_j^T\widetilde{\Xi}_j\widetilde{W}_j$ \textbf{if} $j$ is not leaf
\Comment{\symln{if $A$ is symmetric, replace $Z_j$ by $W_j$}}
\EndFor
\State $\widetilde{\Xi}_i\gets\sum_{j\in \text{Ch}(i)}\widetilde{\Theta}_j$
\State \textbf{if} $i$ is not root \textbf{then}
$\widetilde{\Lambda}_i\gets \Sigma_i-W_i\Sigma_pZ_i^T$ \textbf{else} $\widetilde{\Lambda}_i\gets \Sigma_i$
\textbf{end if}
\Comment{$p$ is parent of $i$}
\Statex \Comment{\symln{if $A$ is symmetric, replace $Z_i$ by $W_i$}}
\State\label{algo.ln:patch2} $\widetilde{\Sigma}_i\gets -(I+\widetilde{\Lambda}_i\widetilde{\Xi}_i)^{-1}\widetilde{\Lambda}_i$
\ForAll{children $j$ of $i$}
\State $\widetilde{E}_j\gets\widetilde{W}_j\widetilde{\Sigma}_i\widetilde{Z}_j^T$ \textbf{if} $j$ is not leaf
\Comment{\symln{if $A$ is symmetric, replace $\widetilde{Z}_j$ by $\widetilde{W}_j$}}
\EndFor
\State $\widetilde{E}_i\gets0$ \textbf{if} $i$ is root
\EndFunction
\Statex
%
\Function{Downward}{$i$}
\If{$i$ is leaf}
\State $\widetilde{A}_{ii}\gets\widetilde{A}_{ii}+\widetilde{U}_i\widetilde{\Sigma}_p\widetilde{V}_i^T$ \textbf{if} $i$ is not root
\Comment{$p$ is parent of $i$}
\Statex \Comment{\symln{if $A$ is symmetric, replace $\widetilde{V}_i$ by $\widetilde{U}_i$}}
\Else
\State $\widetilde{E}_i\gets \widetilde{E}_i+\widetilde{W}_i\widetilde{E}_p\widetilde{Z}_i^T$ \textbf{if} $i$ is not root
\Comment{$p$ is parent of $i$}
\Statex \Comment{\symln{if $A$ is symmetric, replace $\widetilde{Z}_i$ by $\widetilde{W}_i$}}
\State $\widetilde{\Sigma}_i\gets\widetilde{\Sigma}_i+\widetilde{E}_i$
\State \textbf{for all} children $j$ of $i$ \textbf{do} \Call{Downward}{$j$} \textbf{end for}
\EndIf
\EndFunction
\end{algorithmic}
\end{algorithm}
}

\clearpage
\section{Algorithm for Determinant Computation}
The computation of the determinant $\delta=\det(A)$ is rather simple if done simultaneously with the inversion of $A$. The key idea is that one may apply Sylvester's determinant theorem on~\eqref{eqn:AA} to obtain
\begin{equation}\label{eqn:det}
\det(A_{pp}-U_p\Sigma_qV_p^T)=\det(I+\widetilde{\Lambda}_p\widetilde{\Xi}_p)
\prod_{i\in \text{Ch}(p)}\det(A_{ii}-U_i\Sigma_pV_i^T),
\end{equation}
where $\widetilde{\Lambda}_p$ and $\widetilde{\Xi}_p$ are given in~\eqref{eqn:Pi}. In fact, $I+\widetilde{\Lambda}_p\widetilde{\Xi}_p$ must have been factorized in order to compute $\widetilde{\Pi}_p$ in~\eqref{eqn:Pi}; hence its determinant is trivial to obtain. Then, the determinant of $A_{pp}-U_p\Sigma_qV_p^T$ at the parent $p$ is the product of those at the children  $i$, multiplied by $\det(I+\widetilde{\Lambda}_p\widetilde{\Xi}_p)$. A simple recursion suffices for obtaining the determinant at the root. The procedure is summarized as Algorithm~\ref{algo:detA}. It is organized as an upward tree walk.

Note that the determinant easily overflows or underflows in finite precision arithmetics. A common treatment is to compute the log-determinant instead, in which case the multiplications in~\eqref{eqn:det} becomes summation. However, the log-determinant may be complex if $\det(A)$ is negative. Hence, if one wants to avoid complex arithmetic, as we do in Algorithm~\ref{algo:detA}, one may use two quantities, the log-absolute-determinant $\log|\delta|$ and the sign $\sgn(\delta)$, to uniquely represent $\delta$.

\begin{algorithm}[ht]
\caption{Computing $\delta=\det(A)$}
\label{algo:detA}
\begin{algorithmic}[1]
\State Patch Algorithm~\ref{algo:invA}:
\Statex\hspace{.5cm} Line~\ref{algo.ln:patch1}: Store $\log|\delta_i|$ and $\sgn(\delta_i)$, where $\delta_i=\det(A_{ii}-U_i\Sigma_pV_i^T)$
\Statex\hspace{.5cm} Line~\ref{algo.ln:patch2}: Store $\log|\delta_i|$ and $\sgn(\delta_i)$, where $\delta_i=\det(I+\widetilde{\Lambda}_i\widetilde{\Xi}_i)$
\State \Call{Upward}{\texttt{root}}
\Statex

\Function{Upward}{$i$}
\State $\log|\delta|\gets\log|\delta_i|$;\,\, $\sgn(\delta)\gets\sgn(\delta_i)$
\If{$i$ is not leaf}
\ForAll{children $j$ of $i$}
\State \Call{Upward}{$j$}
\State $\log|\delta|\gets\log|\delta|+\log|\delta_j|$;\,\, $\sgn(\delta)\gets\sgn(\delta)\cdot\sgn(\delta_j)$
\EndFor
\EndIf
\State \Return $\log|\delta|$ and $\sgn(\delta)$
\EndFunction
\end{algorithmic}
\end{algorithm}

\clearpage
\section{Algorithm for Cholesky-like Factorization}
The objective is to compute a factorization $A=GG^T$ when $A$ is symmetric positive definite. This factorization is not Cholesky in the traditional sense, because $G$ is not triangular. Rather, we would like to compute a $G$ that has the same structure as $A$, so that we can reuse the matrix-vector multiplication developed in Section~\ref{sec:matvec} on $G$. We repeat the existence theorem of $G$ mentioned in the main paper:

\begin{theorem}
Let $A$ be recursively low-rank with a partitioning tree $T$ and a positive integer $r$. If $A$ is symmetric, by convention let $A$ be represented by the factors
\[
\{A_{ii}, U_i, U_i, \Sigma_p, W_q, W_q \mid i \text{ is leaf, } p \text{ is nonleaf, } q \text{ is neither leaf nor root}\}.
\]
Furthermore, if $A$ is positive definite and additionally, $A_{ii}-U_i\Sigma_pU_i^T$ is also positive definite for all pairs of nonroot node $i$ and parent $p$, then there exists a recursively low-rank matrix $G$ with the same partitioning tree $T$ and integer $r$, and with factors
\[
\{G_{ii}, U_i, V_i, \Omega_p, W_q, Z_q \mid i \text{ is leaf, } p \text{ is nonleaf, } q \text{ is neither leaf nor root}\},
\]
such that $A=GG^T$.
\end{theorem}

Note that in the theorem, $G$ and $A$ share factors $U_i$ and $W_q$. In other words, only the factors $G_{ii}$, $V_i$, $\Omega_p$, and $Z_q$ are to be determined. Similar to matrix inversion, we will prove this theorem through constructing these factors. Consider a pair of child node $p$ and parent $q$ and let $p$ have children such as $i$ and $j$. We repeat~\eqref{eqn:AA} for the symmetric case in the following
\begin{equation}\label{eqn:AA2}
\underbrace{A_{pp}-U_p\Sigma_qU_p^T}_{B_{pp}}=
\diag\Big[\underbrace{A_{ii}-U_i\Sigma_pU_i^T}_{B_{ii}}\Big]_{i\in\text{Ch}(p)}
+\begin{bmatrix}\vdots \\ U_i \\ \vdots\end{bmatrix}
(\underbrace{\Sigma_p-W_p\Sigma_rW_p^T}_{\Lambda_p})
\begin{bmatrix}\cdots & U_i^T & \cdots\end{bmatrix},
\end{equation}
and also write
\begin{equation}\label{eqn:GG2}
\underbrace{G_{pp}-U_p\Omega_qV_p^T}_{C_{pp}}=
\diag\Big[\underbrace{G_{ii}-U_i\Omega_pV_i^T}_{C_{ii}}\Big]_{i\in\text{Ch}(p)}
+\begin{bmatrix}\vdots \\ U_i \\ \vdots\end{bmatrix}
D_p
\begin{bmatrix}\cdots & V_i^T & \cdots\end{bmatrix}
\end{equation}
for some $D_p$. Suppose we have computed $B_{ii}=C_{ii}C_{ii}^T$ for all $i\in\text{Ch}(p)$, then equating $B_{pp}=C_{pp}C_{pp}^T$ we obtain
\begin{equation}\label{eqn:V}
C_{ii}V_i=U_i
\end{equation}
and
\begin{equation}\label{eqn:D}
\Lambda_p=D_p^T+D_p+D_p\Xi_pD_p^T
\quad\text{where}\quad
\Xi_p=\sum_{i\in \text{Ch}(p)}V_i^TV_i.
\end{equation}
When $i$ is a leaf node, we let $C_{ii}$ be the Cholesky factor of $B_{ii}=A_{ii}-U_i\Sigma_pU_i^T$. Then, \eqref{eqn:V} gives the factors $V_i$ of $G$ for all leaf nodes $i$: $V_i=C_{ii}^{-1}U_i$. Further, right-multiplying $V_p$ to both sides of~\eqref{eqn:GG2} and substituting~\eqref{eqn:V}, we have $W_p=(I+D_p\Xi_p)Z_p$, which gives the factors $Z_p$ of $G$ for all nonleaf and nonroot nodes $p$, provided that $D_p$ and $\Xi_p$ are known. The term $\Xi_p$ enjoys a simple recurrence relation that we omit here to avoid tediousness. On the other hand, the term $D_p$ is solved from~\eqref{eqn:D}. Equation~\eqref{eqn:D} is a continuous-time algebraic Riccati equation and it admits a symmetric solution $D_p$ when all the eigenvalues of $I+\Xi_p\Lambda_p$ are positive. It is not hard to see that the eigenvalues of $I+\Xi_p\Lambda_p$ are positive if and only if $B_{pp}$ is symmetric positive definite, which is satisfied based on the assumptions of the theorem. The solution $D_p$ may be computed by using the well-known Schur method~\citep{Laub1979,Arnold1984}.

Additionally, \eqref{eqn:GG2} may be interpreted as relating the Cholesky-like factor of $B_{pp}$ at some parent level $p$, to that of $B_{ii}$ at the child level $i$ with a rank-$r$ correction. Then, let $i$ be a leaf node and $(i,i_1,i_2,\ldots,i_s,1)$ be the path connecting $i$ and the root $=1$. We expand the chain of corrections and obtain
\begin{equation}\label{eqn:Gii}
G_{ii}=C_{ii}
+U_iD_{i_1}V_i^T
+U_iW_{i_1}D_{i_2}Z_{i_1}^TV_i^T
+\cdots
+(U_iW_{i_1}\cdots W_{i_s}D_{1}Z_{i_s}^T\cdots Z_{i_1}^TV_i^T).
\end{equation}
Meanwhile, for any nonleaf node $p$, the factor $\Omega_p$ admits a similar chain of corrections:
\begin{equation}\label{eqn:Omegap}
\Omega_p=D_p
+W_pD_{p_1}Z_p^T
+W_pW_{p_1}D_{p_2}Z_{p_1}^TZ_p^T
+\cdots
+(W_pW_{p_1}\cdots W_{p_t}D_{1}Z_{p_t}^T\cdots Z_{p_1}^TZ_p^T),
\end{equation}
where $(p,p_1,p_2,\ldots,p_t,1)$ is the path connecting $p$ and the root $=1$. The above two formulas give the $G_{ii}$ and $\Omega_p$ factors of $G$ for all leaf nodes $i$ and nonleaf nodes $p$.

Hence, the computation of $G$ consists of two tree walks, one upward and the other downward. In the upward phase, $V_i$ and $Z_p$ are computed. This phase also computes $C_{ii}$ and $D_p$ as the starting point of corrections. Then, in the downward phase, a chain of corrections as detailed by~\eqref{eqn:Gii} and~\eqref{eqn:Omegap} are performed from parent to children, which eventually yields the correct $G_{ii}$ and $\Omega_p$. The overall computation is summarized in Algorithm~\ref{algo:cholA}.

{%
\renewcommand{\baselinestretch}{1.1}
\begin{algorithm}[!ht]
\caption{Cholesky-like factorization $A=GG^T$ (for symmetric positive definite $A$)}
\label{algo:cholA}
\begin{algorithmic}[1]
\State Copy all factors $U_i$ and $W_i$ from $A$ to $G$
\State \Call{Upward}{\texttt{root}}
\State \Call{Downward}{\texttt{root}}
\Statex

\Function{Upward}{$i$}
\If{$i$ is leaf}
\State Factorize $G_{ii}G_{ii}^T\gets A_{ii}-U_i\Sigma_pU_i^T$;\quad
$V_i\gets G_{ii}^{-1}U_i$;\quad
$\Theta_i\gets V_i^TV_i$
\Comment{$p$ is parent of $i$}
\State return
\EndIf
\ForAll{children $j$ of $i$}
\State \Call{Upward}{$j$}
\State $Z_j\gets (I+\Omega_j\Xi_j)^{-1}W_j$ \textbf{if} $j$ is not leaf
\State $\Theta_j\gets Z_j^T\Xi_jZ_j$ \textbf{if} $j$ is not leaf
\EndFor
\State $\Xi_i\gets\sum_{j\in \text{Ch}(i)}\Theta_j$
\State \textbf{if} $i$ is not root \textbf{then}
$\Lambda_i\gets \Sigma_i-W_i\Sigma_pW_i^T$ \textbf{else} $\Lambda_i\gets \Sigma_i$
\textbf{end if}
\Comment{$p$ is parent of $i$}
\State Solve $\Lambda_i=\Omega_i^T+\Omega_i+\Omega_i\Xi_i \Omega_i^T$ for $\Omega_i$
\ForAll{children $j$ of $i$}
\State $E_j\gets W_j\Omega_iZ_j^T$ \textbf{if} $j$ is not leaf
\EndFor
\State $E_i\gets0$ \textbf{if} $i$ is root
\EndFunction
\Statex
\Function{Downward}{$i$}
\If{$i$ is leaf}
\State $G_{ii}\gets G_{ii}+U_i\Omega_pV_i^T$ \textbf{if} $i$ is not root
\Comment{$p$ is parent of $i$}
\Else
\State $E_i\gets E_i+W_iE_pZ_i^T$ \textbf{if} $i$ is not root
\Comment{$p$ is parent of $i$}
\State $\Omega_i\gets\Omega_i+E_i$
\State \textbf{for all} children $j$ of $i$ \textbf{do} \Call{Downward}{$j$} \textbf{end for}
\EndIf
\EndFunction
\end{algorithmic}
\end{algorithm}
}

\clearpage
\section{Algorithm for Constructing $K_{\rm{h}}$}
The computation is summarized in Algorithm~\ref{algo:constructing.A}. See Section~\ref{sec:out.of.sample} of the main paper.

\begin{algorithm}[ht]
\caption{Constructing $A=k_{\rm{h}}(X,X)$}
\label{algo:constructing.A}
\begin{algorithmic}[1]
\State Construct a partitioning tree and for every nonleaf node $i$, find landmark points $\ud{X}_i$
\State \Call{Downward}{\texttt{root}}
\Statex

\Function{Downward}{$i$}
\If{$i$ is leaf}
\State $A_{ii}\gets k(X_i,X_i)$;\quad $U_i\gets k(X_i,\ud{X}_p)k(\ud{X}_p,\ud{X}_p)^{-1}$
\Comment{$p$ is parent of $i$}
\State $V_i\gets\text{empty matrix}$
\State return
\EndIf
\State $\Sigma_i\gets k(\ud{X}_i,\ud{X}_i)$;
\State $W_i\gets k(\ud{X}_i,\ud{X}_p)k(\ud{X}_p,\ud{X}_p)^{-1}$ \textbf{if} $i$ is not root
\Comment{$p$ is parent of $i$}
\State $Z_i\gets\text{empty matrix}$
\State \textbf{for all} children $j$ of $i$ \textbf{do} \Call{Downward}{$j$} \textbf{end for}
\EndFunction
\end{algorithmic}
\end{algorithm}

\clearpage
\section{Algorithm for Computing $\bm{w}^T\bm{v}$ with $\bm{v}=k_{\rm{h}}(X,\bm{x})$}
To begin with, note that $\bm{x}$ must lie in one of the subdomains $S_j$ for some leaf node $j$. We will abuse language and say that ``$\bm{x}$ lies in the leaf node $j$'' for simplicity. In such a case, the subvector $\bm{v}_j=k(X_j,\bm{x})$ and for any leaf node $l\ne j$, the subvector
\[
\bm{v}_l=U_lW_{l_1}W_{l_2}\cdots W_{l_s}\Sigma_pW_{j_t}^T\cdots W_{j_2}^TW_{j_1}^Tk(\ud{X}_{j_1},\ud{X}_{j_1})^{-1}k(\ud{X}_{j_1},\bm{x}),
\]
where $p$ is the least common ancestor of $j$ and $l$, $(l,l_1,l_2,\ldots,l_s,p)$ is the path connecting $l$ and $p$, and $(j,j_1,j_2,\ldots,j_t,p)$ is the path connecting $j$ and $p$. Then, the inner product
\[
\bm{w}^T\bm{v}=\bm{w}_j^Tk(X_j,\bm{x})+\sum_{l\ne j,\,\, l \text{ is leaf}}
\bm{w}_l^TU_lW_{l_1}W_{l_2}\cdots W_{l_s}\Sigma_pW_{j_t}^T\cdots W_{j_2}^TW_{j_1}^Tk(\ud{X}_{j_1},\ud{X}_{j_1})^{-1}k(\ud{X}_{j_1},\bm{x}).
\]

Similar to matrix-vector multiplications, we may define a few sets of auxiliary vectors to avoid duplicate computations. Specifically, define $\bm{x}$-independent vectors
\[
\bm{e}_i=\begin{dcases}
U_i^T\bm{w}_i, & \text{if $i$ is leaf},\\
W_i^T\sum_{j\in \text{Ch}(i)}\bm{e}_j, & \text{otherwise},\\
\end{dcases}
\]
and
\[
\bm{c}_l=\Sigma_p^T\bm{e}_i\quad\text{for $i$ and $l$ being siblings with parent $p$},
\]
and $\bm{x}$-dependent vectors
\[
\bm{d}_p=W_p^T\bm{d}_i\quad\text{for $p$ being the parent of $i$}; \qquad
\bm{d}_j=k(\ud{X}_{j_1},\ud{X}_{j_1})^{-1}k(\ud{X}_{j_1},\bm{x})\quad\text{for $\bm{x}$ lying in $j$}.
\]
Then, the inner product is simplified as
\[
\bm{w}^T\bm{v}=\bm{w}_j^Tk(X_j,\bm{x})+\sum_{j_t\,\in\,\text{path connecting } j \text{ and root}}\bm{c}_{j_t}^T\bm{d}_{j_t}.
\]

Hence, the computation of $\bm{w}^T\bm{v}$ consists of a full tree walk and a partial one, both upward. The first upward phase computes $\bm{e}_i$ from children to parent and simultaneously $\bm{c}_l$ by crossing sibling nodes from $i$ to $l$. This computation is independent of $\bm{x}$ and hence is considered preprocessing. The second upward phase computes $\bm{d}_{j_t}$ for all $j_t$ along the path connecting $j$ and the root. This phase visits only one path but not the whole tree, which is the reason why it costs less than $O(n)$. We summarize the detailed procedure in Algorithm~\ref{algo:inprod}.

\begin{algorithm}[ht]
\caption{Computing $z=\bm{w}^T\bm{v}$, where $\bm{v}=k_{\rm{h}}(X,\bm{x})$, for $\bm{x}\notin X$}
\label{algo:inprod}
\begin{algorithmic}[1]
\State \Call{Common-Upward}{\texttt{root}}
\Statex $\triangleright$ The above step is independent of $\bm{x}$ and is treated as preprocessing. In computer implementation, the intermediate results $\bm{c}_i$ are carried over to the next step \Call{Second-Upward}{}, whereas the contents in $\bm{d}_i$ are discarded and the allocated memory is reused.
\State \Call{Second-Upward}{\texttt{root}}
\Statex
\Function{Common-Upward}{$i$}
\If{$i$ is leaf}
\State $\bm{d}_i\gets U_i^T\bm{w}_i$
\Else
\State \textbf{for all} children $j$ of $i$ \textbf{do} \Call{Common-Upward}{$j$} \textbf{end for}
\State $\bm{d}_i\gets W_i^T\left(\sum_{j\in \text{Ch}(i)}\bm{d}_j\right)$ \textbf{if} $i$ is not root
\EndIf
\If{$i$ is not root}
\State \textbf{for all} siblings $l$ of $i$ \textbf{do} $\bm{c}_l\gets\Sigma_p^T \bm{d}_i$ \textbf{end for}
\Comment{$p$ is parent of $i$}
\EndIf
\EndFunction
\Statex
\Function{Second-Upward}{$i$}
\If{$i$ is leaf}
\State $\bm{d}_i\gets k(\ud{X}_p,\ud{X}_p)^{-1}k(\ud{X}_p,\bm{x})$
\Comment{$p$ is parent of $i$}
\State $z\gets \bm{w}_i^Tk(X_i,\bm{x})$
\Else
\State Find the child $j$ (among all children of $i$) where $\bm{x}$ lies in
\State \Call{Second-Upward}{$j$}
\State $\bm{d}_i\gets W_i^T\bm{d}_j$ \textbf{if} $i$ is not root
\EndIf
\State $z\gets z+\bm{c}_i^T\bm{d}_i$ \textbf{if} $i$ is not root
\EndFunction
\end{algorithmic}
\end{algorithm}

\clearpage
\section{Algorithm for Computing $\bm{v}^T\widetilde{A}\bm{v}$ with $\bm{v}=k_{\rm{h}}(X,\bm{x})$ for Symmetric $\widetilde{A}$}
We consider the general case where $\widetilde{A}$ is not necessarily related to the covariance function $k_{\rm{h}}$; what is assumed is only symmetry. We recall that $\widetilde{A}$ is represented by the factors
\[
\{\widetilde{A}_{ii}, \widetilde{U}_i, \widetilde{U}_i, \widetilde{\Sigma}_p, \widetilde{W}_q, \widetilde{W}_q \mid i \text{ is leaf, } p \text{ is nonleaf, } q \text{ is neither leaf nor root}\}.
\]

The derivation of the algorithm is more involved than that of the previous ones; hence, we need to introduce further notations. Let $\text{p}(i)$ denote the parent of a node $i$ and similarly $\text{p}(i,j)$ denote the common parent of $i$ and $j$. Let $(l,l_1,l_2,\ldots,l_t,p)$ be a path connecting nodes $l$ and $p$, where $l$ is a descendant of $p$. Denote this path as $\text{path}(l,p)$ for short. We will use subscripts $l\to p$ and $p\gets l$ to simplify the notation of the product chain of the $W$ factors:
\[
W_{l\to p}\equiv W_{l_1}W_{l_2}\cdots W_{l_t}
\quad\text{and}\quad
W_{p\gets l}^T\equiv W_{l_t}^T\cdots W_{l_2}^TW_{l_1}^T.
\]
Note that the two ends of the path (i.e., $l$ and $p$) are not included in the product chain. If $l$ is a leaf and $p$ is the root, then every node $i\in\text{path}(l,p)$, except the root, has the parent also in the path, but its siblings are not. We collect all these sibling nodes to form a set $\text{B}(l)$. It is not hard to see that $\text{B}(l)\cup\{l\}$ is a disjoint partitioning of whole index set. Moreover, any two nodes from the set $\text{B}(l)\cup\{l\}$ must have a least common ancestor belonging to $\text{path}(l,\text{root})$; and this ancestor is the parent of (at least) one of the two nodes. If $\bm{x}$ lies in a leaf node $l$, $i$ is some node $\in\text{path}(l,\text{root})$, and $j\in \text{B}(l)$ is a sibling of $i$, then by reusing the $\bm{d}$ vectors defined in the preceding subsection, we have
\begin{equation}\label{eqn:vlj}
\bm{v}_l=k(X_l,\bm{x}) \quad\text{and}\quad
\bm{v}_j=U_j\Sigma_{\text{p}(j)}W_{\text{p}(j)\gets l}^Tk(\ud{X}_{\text{p}(l)},\ud{X}_{\text{p}(l)})^{-1}k(\ud{X}_{\text{p}(l)},\bm{x})
=U_j\Sigma_{\text{p}(j,i)}\bm{d}_i.
\end{equation}

Because $\text{B}(l)\cup\{l\}$ forms a disjoint partitioning of whole index set, the quadratic form $\bm{v}^T\widetilde{A}\bm{v}$ consists of three parts:
\[
\bm{v}^T\widetilde{A}\bm{v}
=\bm{v}_l^T\widetilde{A}_{ll}\bm{v}_l
+\sum_{i\in \text{B}(l)} \bm{v}_i^T\widetilde{A}_{ii}\bm{v}_i
+\sum_{\substack{i,j\in \text{B}(l)\\i\ne j}} \bm{v}_i^T\widetilde{A}_{ij}\bm{v}_j.
\]
The first part involving the leaf node $l$ is straightforward. For the second part, we expand $\bm{v}_i$ by using~\eqref{eqn:vlj} and define two quantities therein:
\[
\bm{v}_i^T\widetilde{A}_{ii}\bm{v}_i
=\Big( \bm{d}_t^T\underbrace{\Sigma_{\text{p}(i,t)}^T\overbrace{U_i^T \Big) \widetilde{A}_{ii} \Big( U_i}^{\Xi_i}\Sigma_{\text{p}(i,t)}}_{\widetilde{\Xi}_i}\bm{d}_t \Big),
\]
where $t$ as a sibling of $i$ belongs to $\text{path}(l,\text{root})$. For the third part, we similarly expand each individual term and define additionally two quantities:
\[
\bm{v}_i^T\widetilde{A}_{ij}\bm{v}_j
=\Big( \bm{d}_s^T
\underbrace{\Sigma_{\text{p}(i,s)}^T\overbrace{U_i^T \Big) \Big( \widetilde{U}_i}^{\Theta_i^T}}_{\widetilde{\Theta}_i^T}
\widetilde{W}_{i\rightarrow q}\widetilde{\Sigma}_q\widetilde{W}_{q\leftarrow j}^T
\underbrace{\overbrace{\widetilde{U}_j^T \Big) \Big( U_j}^{\Theta_j}\Sigma_{\text{p}(j,t)}}_{\widetilde{\Theta}_j}
\bm{d}_t \Big),
\]
where $s$ as a sibling of $i$ belongs to $\text{path}(l,\text{root})$, $t$ as a sibling of $j$ belongs to the same path, and $q$ is the least common ancestor of $i$ and $j$. The four newly introduced quantities $\Xi_i$, $\widetilde{\Xi}_i$, $\Theta_i$, and $\widetilde{\Theta}_i$ are independent of $\bm{x}$ and may be computed in preprocessing, in a recursive manner from children to parent. We omit the simple recurrence relation here to avoid tediousness. Then, the quadratic form $\bm{v}^T\widetilde{A}\bm{v}$ admits the following expression:
\[
\bm{v}^T\widetilde{A}\bm{v}
=\bm{v}_l^T\widetilde{A}_{ll}\bm{v}_l
+\sum_{i\in \text{B}(l)} \bm{d}_t^T\widetilde{\Xi}_i\bm{d}_t
+\sum_{\substack{i,j\in \text{B}(l)\\i\ne j}} \bm{d}_s^T\widetilde{\Theta}_i^T\widetilde{W}_{i\rightarrow q}\widetilde{\Sigma}_q\widetilde{W}_{q\leftarrow j}^T\widetilde{\Theta}_j\bm{d}_t.
\]

We may further simplify the summation in the last term of this equation to avoid duplicate computation. As mentioned, any two nodes in $\text{B}(l)$ have a least common ancestor that happens to be the parent of one of them. Assume that this node is $i$. Then, we write
\[
\sum_{\substack{i,j\in \text{B}(l)\\i\ne j}} \bm{d}_s^T\widetilde{\Theta}_i^T\widetilde{W}_{i\rightarrow q}\widetilde{\Sigma}_q\widetilde{W}_{q\leftarrow j}^T\widetilde{\Theta}_j\bm{d}_t
=2\sum_{i\in \text{B}(l)} \bm{d}_s^T\widetilde{\Theta}_i^T\widetilde{\Sigma}_{\text{p}(i)}\sum_{\substack{j\in\text{B}(l),\,\,j\ne i \\ j \text{ is descendant of p}(i)}}\widetilde{W}_{\text{p}(i)\leftarrow j}^T\widetilde{\Theta}_j\bm{d}_t.
\]
Note the inner summation on the right-hand side of this equality. This quantity iteratively accumulates as $i$ moves up the tree. Therefore, we define
\[
\bm{c}_i=\begin{dcases}
\widetilde{\Theta}_i\bm{d}_s, & \text{if $i\in \text{B}(l)$},\\
\widetilde{W}_i^T\sum_{j\in \text{Ch}(i)}\bm{c}_j, & \text{if $i\in \text{path}(l,\text{root})$},
\end{dcases}
\]
where recall that $s$ as a sibling of $i$ belongs to $\text{path}(l,\text{root})$. Then, the inner summation becomes $\bm{c}_s$. In other words,
\[
\sum_{\substack{i,j\in \text{B}(l)\\i\ne j}} \bm{d}_s^T\widetilde{\Theta}_i^T\widetilde{W}_{i\rightarrow q}\widetilde{\Sigma}_q\widetilde{W}_{q\leftarrow j}^T\widetilde{\Theta}_j\bm{d}_t
=2\sum_{i\in \text{B}(l)}\bm{c}_i^T\widetilde{\Sigma}_{\text{p}(i,s)}\bm{c}_s.
\]

To summarize, the computation of $\bm{v}^T\widetilde{A}\bm{v}$ consists of a full tree walk and a partial one, both upward. The first upward phase computes $\Xi_i$, $\widetilde{\Xi}_i$, $\Theta_i$, and $\widetilde{\Theta}_i$ recursively from children to parent. This computation is independent of $\bm{x}$ and hence is considered preprocessing. The second upward phase computes $\bm{d}_s$ and $\bm{c}_s$ for all $s$ along the path connecting $l$ and the root (assuming $\bm{x}\in S_l$), as well as all $\bm{c}_i$ for $i$ being sibling nodes of $s$. This phase visits only one path but not the whole tree, which is the reason why it costs less than $O(n)$. The detailed procedure is given in Algorithm~\ref{algo:quadratic.form}.

{%
\renewcommand{\baselinestretch}{1.1}
\begin{algorithm}[!ht]
\caption{Computing $z=\bm{v}^T\widetilde{A}\bm{v}$, where $\widetilde{A}$ is symmetric and $\bm{v}=k_{\rm{h}}(X,\bm{x})$, for $\bm{x}\notin X$}
\label{algo:quadratic.form}
\begin{algorithmic}[1]
\State \Call{Common-Upward}{\texttt{root}}
\Statex $\triangleright$ The above step is independent of $\bm{x}$ and is treated as preprocessing.
\State \Call{Second-Upward}{\texttt{root}}
\Statex
\Function{Common-Upward}{$i$}
\If{$i$ is leaf}
\State $\Theta_i\gets\widetilde{U}_i^TU_i$;\quad
$\widetilde{\Theta}_i\gets\Theta_i\Sigma_p$
\Comment{$p$ is parent of $i$}
\State $\Xi_i\gets U_i^T\widetilde{A}_iU_i$;\quad
$\widetilde{\Xi}_i\gets\Sigma_p^T\Xi_i\Sigma_p$
\Comment{$p$ is parent of $i$}
\State return
\EndIf
\State \textbf{for all} children $j$ of $i$ \textbf{do} \Call{Common-Upward}{$j$} \textbf{end for}
\If{$i$ is not root}
\State $\Theta_i\gets\widetilde{W}_i^T\left(\sum_{j\in \text{Ch}(i)}\Theta_j\right)W_i$;\quad
$\widetilde{\Theta}_i\gets\Theta_i\Sigma_p$
\Comment{$p$ is parent of $i$}
\State $\Xi_i\gets W_i^T\left(\sum_{j\in \text{Ch}(i)}\Xi_j+\sum_{\substack{j,k\in \text{Ch}(i)\\j\ne k}}\Theta_j^T\widetilde{\Sigma}_i\Theta_k\right)W_i$;\quad
$\widetilde{\Xi}_i\gets\Sigma_p^T\Xi_i\Sigma_p$
\Comment{$p$ is parent of $i$}
\EndIf
\EndFunction
\Statex
\Function{Second-Upward}{$i$}
\If{$i$ is leaf}
\State $\bm{d}_i\gets k(\ud{X}_p,\ud{X}_p)^{-1}k(\ud{X}_p,\bm{x})$
\Comment{$p$ is parent of $i$}
\State $\bm{c}_i\gets\widetilde{U}_i^Tk(X_i,\bm{x})$
\State $z\gets k(\bm{x},X_i)\widetilde{A}_ik(X_i,\bm{x})$
\Else
\State Find the child $j$ (among all children of $i$) where $\bm{x}$ lies in
\State \Call{Second-Upward}{$j$}
\State $\bm{d}_i\gets W_i^T\bm{d}_j$ \textbf{if} $i$ is not root
\EndIf
\If{$i$ is not root}
\ForAll{siblings $l$ of $i$}
\State $\bm{c}_l\gets\widetilde{\Theta}_l\bm{d}_i$
\State $z\gets z+\bm{d}_i^T\widetilde{\Xi}_l\bm{d}_i+2\bm{c}_l^T\widetilde{\Sigma}_p\bm{c}_i$
\Comment{$p$ is parent of $i$}
\EndFor
\State $\bm{c}_p\gets \widetilde{W}_p^T\left(\sum_{j\in \text{Ch}(p)}\bm{c}_j\right)$ \textbf{if} $p$ is not root
\Comment{$p$ is parent of $i$}
\EndIf
\EndFunction
\end{algorithmic}
\end{algorithm}
}

\clearpage
\section{Cost Analysis}
The storage cost has been analyzed in the main paper. In what follows is the analysis of arithmetic costs.

\subsection{Arithmetic Cost of Matrix-Vector Multiplication (Algorithm~\ref{algo:Ab})}
The algorithm consists of two tree walks, each of which visits all the $O(n/r)$ nodes. Inside each tree node, the computation is dominated by $O(1)$ matrix-vector multiplications with $r\times r$ matrices; hence the per-node cost is $O(r^2)$. Then, the overall cost is $O(n/r\times r^2)=O(nr)$.

\subsection{Arithmetic Cost of Matrix Inversion (Algorithm~\ref{algo:invA})}
The algorithm consists of two tree walks, each of which visits all the $O(n/r)$ nodes. Inside each tree node, the computation is dominated by $O(1)$ matrix operations (matrix-matrix multiplications and inversions) with $r\times r$ matrices; hence the per-node cost is $O(r^3)$. Then, the overall cost is $O(n/r\times r^3)=O(nr^2)$.

\subsection{Arithmetic Cost of Determinant Computation (Algorithm~\ref{algo:detA})}
The algorithm requires patching Algorithm~\ref{algo:invA} with additional computations that do not affect the $O(nr^2)$ cost of Algorithm~\ref{algo:invA}. Omitting the patching, Algorithm~\ref{algo:detA} visits every tree node once and the computation per node is $O(1)$. Hence, the cost of this algorithm is only $O(n/r)$.

In practice, we indeed implement the patching inside Algorithm~\ref{algo:invA}.

\subsection{Arithmetic Cost of Cholesky-like Factorization (Algorithm~\ref{algo:cholA})}
The cost analysis of this algorithm is almost the same as that of Algorithm~\ref{algo:invA}, except that the dominating per-node computation also includes Cholesky factorization of $r\times r$ matrices and the solving of continuous-time algebraic Riccati equation of size $r\times r$. Both costs are $O(r^3)$, the same as that of matrix-matrix multiplications and inversions. Hence, the overall cost of this algorithm is $O(nr^2)$.

\subsection{Arithmetic Cost of Constructing $K_{\rm{h}}$ (Algorithm~\ref{algo:constructing.A})}
The algorithm consists of three parts: (i) hierarchical partitioning of the domain; (ii) finding landmark points; and (iii) instantiating the factors of a symmetric recursively low-rank matrix.

For part (i), much flexibility exists. In practice, partitioning is data driven, which ensures that the number of points is balanced in all leaf nodes. If we assume that the cost of partitioning a set of $n$ points is $O(n)$, then the overall partitioning cost counting recursion is $O(n\log n)$.

Similarly, part (ii) depends on the specific method used for choosing the landmark points. In general, we may assume that choosing $r$ landmark points costs $O(r)$. Then, because each of the $O(n/r)$ nonleaf nodes has a set of landmark points, the cost is $O(n/r\times r)=O(n)$.

Part (iii) is a tree walk that visits each of the $O(n/r)$ nodes once. The per-node computation is dominated by constructing one or a few $r\times r$ covariance matrices and performing matrix-matrix multiplications and inversions. We assume that constructing an $r\times r$ covariance matrix costs $O(r^2)$, which is less expensive than the $O(r^3)$ cost of matrix-matrix multiplications and inversions. Then, the overall cost for instantiating the overall matrix is $O(n/r\times r^3)=O(nr^2)$.

\subsection{Arithmetic Cost of Computing $\bm{w}^T\bm{v}$ (Algorithm~\ref{algo:inprod})}
The algorithm consists of two tree walks (one full and one partial): the first one is $\bm{x}$-independent preprocessing and the second one is $\bm{x}$-dependent.

For preprocessing, the tree walk visits all the $O(n/r)$ nodes. Inside each tree node, the computation is dominated by $O(1)$ matrix-vector multiplications with $r\times r$ matrices; hence the per-node cost is $O(r^2)$. Then, the overall preprocessing cost is $O(n/r\times r^2)=O(nr)$.

For the $\bm{x}$-dependent computation, only $O(h)=O(\log_2(n/r))$ tree nodes are visited. Inside each visited node, the computation is dominated by $O(1)$ matrix-vector multiplications with $r\times r$ matrices; hence the per-node cost is $O(r^2)$. Here, we assume that finding the child node where $\bm{x}$ lies in has $O(1)$ cost. Note also that although the computation of the $\bm{d}$ vectors requires a matrix inverse, the matrix in fact has been prefactorized when constructing $K_{\rm{h}}$ (that is, inside Algorithm~\ref{algo:constructing.A}). Hence, the per-node cost is not $O(r^3)$. To conclude, the $\bm{x}$-dependent cost is $O(r^2\log_2(n/r))$.

\subsection{Arithmetic Cost of Computing $\bm{v}^T\widetilde{A}\bm{v}$ (Algorithm~\ref{algo:quadratic.form})}
The cost analysis of this algorithm is almost the same as that of Algorithm~\ref{algo:inprod}, except that in the preprocessing phase, the dominant per-node computation is $O(1)$ matrix-matrix multiplications with $r\times r$ matrices. Hence, the preprocessing cost is $O(n/r\times r^3)=O(nr^2)$ whereas the $\bm{x}$-dependent cost is still $O(r^2\log_2(n/r))$.

\end{document}